\documentclass[sii]{ipart}
\pdfoutput=1
\usepackage{ulem}
\usepackage{ragged2e}
\usepackage{algorithm}
\usepackage{multirow}
\usepackage{amsmath, amsthm, amssymb, booktabs}
\usepackage{graphicx}
\usepackage{subcaption}
\RequirePackage{hyperref}

\startlocaldefs

\theoremstyle{plain}
\newtheorem{thm}{Theorem}
\newtheorem{lemma}{Lemma}

\newtheorem{prop}{Proposition}
\newtheorem{assumption}{Assumption}

\theoremstyle{definition}

\newtheorem{remark}{Remark}

\endlocaldefs

\pubyear{2024}
\volume{0}
\issue{0}
\firstpage{1}
\lastpage{12}

\begin{document}

\begin{frontmatter}
\title{Novel Subsampling Strategies for Heavily Censored Reliability Data\protect\thanksref{T1}}
\thankstext{T1}{This research is supported partly by National Key  R\&D Program of China 2021YFA1000300, 2021YFA1000301.}

\begin{aug}
\author{\inits{YX}\fnms{Yixiao} \snm{Ruan}\ead[label=e1]{ruanyixiao16@mails.ucas.ac.cn}},
    \address{Academy of Mathematics and Systems Science, Chinese Academy of Sciences, Beijing, China.\\
  School of Mathematical Sciences, University of Chinese Academy of Sciences, Beijing, China.\\}

   \author{\inits{Z}\fnms{Zan} \snm{Li}\ead[label=e2]{lizan2022@nankai.edu.cn}},
    \address{School of Statistics and Data Science, KLMDASR and LPMC, Nankai University, Tianjin, China.\\}
              
    \author{\inits{ZH}\fnms{Zhaohui} \snm{Li}\thanksref{t2}\ead[label=e3]{lizh@amss.ac.cn}},
   \address{Academy of Mathematics and Systems Science, Chinese Academy of Sciences, Beijing, China.\\
     School of Mathematical Sciences, University of Chinese Academy of Sciences, Beijing, China.\\}

   \author{\inits{DKJ}\fnms{Dennis K. J.} \snm{Lin}\ead[label=e4]{dkjlin@purdue.edu}},   \address{Department of Statistics, Purdue University.\\}
    
   \author{\inits{QP}\fnms{Qingpei} \snm{Hu}\ead[label=e5]{qingpeihu@amss.ac.cn}},
  \address{Academy of Mathematics and Systems Science, Chinese Academy of Sciences, Beijing, China.\\
  School of Mathematical Sciences, University of Chinese Academy of Sciences, Beijing, China.\\}
    \author{\inits{D.Y}\fnms{Dan} \snm{Yu}\ead[label=e6]{dyu@amss.ac.cn}}
  \address{Academy of Mathematics and Systems Science, Chinese Academy of Sciences, Beijing, China.\\
   School of Mathematical Sciences, University of Chinese Academy of Sciences, Beijing, China.\\}
     \thankstext{t2}{Corresponding author: lizh@amss.ac.cn.}
\end{aug}
%

\begin{abstract}
Computational capability often falls short when confronted with massive data, posing a common challenge in establishing a statistical model or statistical inference method dealing with big data. 
While subsampling techniques have been extensively developed to downsize the data volume, there is a notable gap in addressing the unique challenge of handling extensive reliability data, in which a common situation is that a large proportion of data is censored. 
In this article, we propose an efficient subsampling method for reliability analysis in the presence of censoring data, intending to estimate the parameters of lifetime distribution. 
Moreover, a novel subsampling method for subsampling from severely censored data is proposed, i.e., only a tiny proportion of data is complete. 
The subsampling-based estimators are given, and their asymptotic properties are derived.
The optimal subsampling probabilities are derived through the L-optimality criterion, which minimizes the trace of the product of the asymptotic covariance matrix and a constant matrix. 
Efficient algorithms are proposed to implement the proposed subsampling methods to address the challenge that optimal subsampling strategy depends on unknown parameter estimation from full data.
Real-world hard drive dataset case and simulative empirical studies are employed to demonstrate the superior performance of the proposed methods.
\end{abstract}

\begin{keyword}
\kwd{Subsampling Design}
\kwd{Massive Data}
\kwd{L-optimality}
\kwd{Reliability Analysis}
\end{keyword}

\end{frontmatter}

\section{Introdction}\label{intro} 
As the volume of data explodes, analyzing massive datasets becomes an urgent task. The huge amount of data presents significant challenges to the existing computing resources. This is also a trend in the reliability analysis, where an increase in censored data is becoming more prevalent \citep{Kleinbaum1996}. Confronted with such a phenomenon, \citet{Meeker2014} explored the opportunities and challenges when dealing with the massive reliability field data.

\citet{Backblaze} reports the daily operational hard disk drive (HDD) data on its website\footnote{\url{https://www.backblaze.com/cloud-storage/resources/hard-drive-test-data}.} every quarter, 
and a huge amount of data has been collected. Existing data reflects the product's high reliability, resulting in a severe imbalance, i.e., only a small proportion of HDDs fail (e.g., only 1.2\% of data are uncensored in 2016).
These are massive hard drive failure data with the majority of censored data, and we will surely encounter big data problems when we try to fit the HDDs' lifetime distribution. A few attempts have been made. For example, \citet{mittman2019} proposed a hierarchical model based on the generalized limited failure population (GLFP) to describe the failure time of HDD. \citet{Lewis-Beck2022} studied the prediction method using this model. These methods are effective. However, only a small amount of the data is analyzed. Thus, they still face the issue of massive data.

This example highlights the challenge of dealing with highly censored lifetime data, which is becoming a common occurrence. 
Censoring and truncating will frequently happen in collecting such data due to the high reliability of products and the limited time for data collection \citep{Emura_review}.  
In reliability analysis, the lifetime distribution model is a fundamental statistical model that describes the distribution of failure time/lifetime by using a parametric model such as Weibull distribution \citep{Lawless2011}. 
Estimating the parameters in these lifetime models and quantifying the uncertainty of these estimations are vital tasks for reliability assessment in real-world applications. 
In this example, however, the huge sample size makes it difficult to calculate full-data Maximum likelihood estimation (MLE) due to the limited computing resources. 

To address this important issue, one potential method is the \textit{subsampling}. The general idea of a subsampling scheme is to assign appropriate subsampling probabilities to each data point and draw a representative small-size subsample from the original data \citep{Wang2021}. 
Next, the parameter estimation and the hypothesis testing, etc, proceed, based on the subsample with appropriate adjustments to ensure some desired properties for estimation (such as unbiasedness).
As a result, subsampling can significantly reduce the computational burden by shrinking the volume of data.
Numerous work on massive data subsampling have been published recently.
Most focus on deriving optimal subsampling probabilities under various statistical models, such as linear regression \citep{Ma2015, Wang2021}, generalized linear models \citep{Zhang2021,Ai2021}, logistic regression \citep{Fithian2014,Cheng2020}, multiclass logistic regression \citep{Han2020}, and maximum quasi-likelihood estimators \citep{Yu2022}.
Other studies include the optimal subsampling procedure motivated by A-optimality criterion (OSMAC, \citealp{Wang2018}), influence-based subsampling method \citep{wang2020less} and trimmed-space-filling sampling procedure \citep{deng2023big} and  subsampling into the Bayesian hierarchical model \citep{Bradley2021}. However, all these methods focused on the uncensored data, which cannot be applied to the dataset in \cite{Backblaze}.

To the best of our knowledge, there is no existing subsampling methods have been proposed for analyzing \textit{highly censored massive lifetime data}. 
This article aims to fill this gap by proposing two subsampling strategies to obtain the subsampling-based estimators, called RDS (Reliability Data general Subsampling method) and RDCS (Reliability Data Censoring Subsampling method) for the massive censored data when facing different censoring rates. 
The main contributions of our proposed subsampling methods include three aspects: (\romannumeral1) we provide two optimal subsampling schemes for massive lifetime data based on the L-optimality criterion. In particular, the RDS is used when the censoring rate is moderately high, while the RDCS is used when the censoring rate is extremely high. Moreover, we derive the precise formulas for the subsampling probabilities of our two methods and propose two practical algorithms to implement the subsampling procedures. (\romannumeral2) The theoretical properties, including asymptotic normality, are provided, which are crucial for calculating optimal subsampling probabilities and constructing confidence intervals. And (\romannumeral3) the proposed two subsampling-based estimators are more computationally efficient than the full-data MLE and perform much better than the uniform subsampling method in statistical efficiency and estimation accuracy. 

The remainder of this article is organized as follows. In Section 2, we formulate the subsampling problem and propose two subsampling procedures. Two subsampling-based estimators are also provided to approximate the full-data MLE. In Section 3, we derive consistency and asymptotic normality of two subsampling-based estimators. Section 4 proposes the method to obtain the optimal subsampling probabilities in the context of the L-optimality criterion. Moreover, we give two practical algorithms for the two proposed methods and their theoretical properties are studied. 
In Section 5, simulations and practical applications are conducted to illustrate the superiority of our methods compared with uniform subsampling. Conclusions and future work are given in Section 6. Proofs and additional examples are provided in the appendix.

\section{The Subsampling Methods for Massive Data with Censoring}
In Section \ref{subsec:formulation}, we formulate the problem of the subsampling methods for censored lifetime data and introduce the notations used in this paper. 
This section then proposes a general subsampling method in Section \ref{subsec:gse} and a censored subsampling method in Section \ref{subsec:cse}.
\subsection{Notations and Problem Formulation}\label{subsec:formulation}

Here, we consider the lifetime data given by the form $\mathcal{D}_n=\lbrace \boldsymbol{x}_i=(t_i, C_i, t_{il}), i = 1,2,...,n\rbrace$. 
For each individual, either the failure time or the right censoring time can be observed and is denoted by $t_i$. 
$C_i$ is the indicator for censoring: $C_i=1$ if the $i$th observation is censored, otherwise $C_i=0$.
Moreover, $t_{il}$ denotes the left-truncation time, the age of the $i$th individual at the beginning of data reporting. 
Truncation is a common feature in survival data. If the observation is left-truncated, it means that it would not have
been observed if it had failed before a particular time. 
Biased estimations can be made if truncation is neglected \citep{Klein2003}. 
We set $t_{il}=0$ if the $i$th individual is not left-truncated.

The size of the dataset $n$ is huge.
Only a relatively small $n_0$ individuals fail and the rest of $n_1 = n - n_0$ individuals are right-censored. 
Define the censoring rate as $\alpha= n_1/n \in [0,1]$, to describe the proportion of the censoring data in the full dataset. 
Here, we focus on the situation where the data is heavily censored, that is, $\alpha$ approaches 1.
Moreover, it is assumed that the censoring time and the truncating time are independent of failure time, and all the individuals are independent and identically distributed. 
Consider the motivating example, we will choose the data to analyze from a certain period. 
The censoring and truncation time is determined by the install time and the certain period selected by the reliability engineer, which has no relation to the failure time. 
Thus, treating the truncation and the censoring time to be independent of the failure time is reasonable.

The lifetime of the product is typically modeled using a non-negative random variable. 
Let $f(t;\boldsymbol\theta)$ be the probability density function (PDF) of the lifetime, which is parameterized by $\boldsymbol\theta \in \mathbb{R}^d$. 
The following PDFs are commonly used lifetime models: 
\begin{itemize}
    \item Exponential Distribution: $f(t)=\theta e^{-\theta t}$, $\theta>0$.
    \item Weibull Distribution: $f(t)=(\beta/\eta)(t/\eta)^{\beta-1} e^{(-t/\eta)^{\beta}}$, $\beta>0$, $\eta>0$.
    \item GLFP Distribution \citep{Chan1999}: 
    $f(t) = \pi_gf_1(1-F_2)+f_2(1-\pi_g F_1)$, $F_1$, $F_2$ are the cumulative distribution function of two Weibull distributions with parameters $(\beta_1,\eta_1)$ and $(\beta_2,\eta_2)$, and $f_1$, $f_2$ are the corresponding PDFs. 
\end{itemize}

The first two distributions are widely used in the reliability analysis \citep{Lawless2011} and the last one turns out to have good performance when fitting the failures of hard drives \citep{mittman2019}. For each specific parametric lifetime model, the unknown parameters in the model have explicit reliability interpretations. 
For example, the parameter $\theta$ in the exponential model indicates the failure rate of the product. 
Estimating these parameters from the data is a vital task for reliability analysis \citep{Practical_reliability}.
One of the most popular parameter estimation methods is called maximum likelihood estimation (MLE).
Given the full data and the parametric PDF for the lifetime model, the likelihood function can be written as:
$L(\boldsymbol\theta)=\prod_{i=1}^{n}\left[\frac{f(t_i;\boldsymbol\theta)}{1-F(t_{il};\boldsymbol\theta)}\right]^{1-C_i}\left[\frac{1-F(t_i;\boldsymbol\theta)}{1-F(t_{il};\boldsymbol\theta)}\right]^{C_i},$
where $\boldsymbol\theta=(\theta_1,\theta_2,...,\theta_d)$ is an unknown $d$-dimension parameter.
The MLE is typically obtained by maximizing the log-likelihood function of the full data, i.e., $\boldsymbol{\hat{\theta}}=\mathop{\arg\max}\limits_{\theta}l_f(\boldsymbol\theta)$, where the log-likelihood is given by 
\begin{equation}
    \begin{aligned}
     l_f(\boldsymbol\theta)
     =&\frac{1}{n}\sum_{i=1}^{n}(1-C_i)\log f(t_i;\boldsymbol\theta)+
     C_i\log\left[1-F(t_i;\boldsymbol\theta)\right]\\
     -&\log\left[1-F(t_{il};\boldsymbol\theta)\right] \\
     \triangleq &\frac{1}{n}\sum_{i=1}^{n}l(\boldsymbol{x}_i,\boldsymbol\theta).
    \end{aligned}
    \label{Equ:nlog-likelihood}
\end{equation}
In most cases, there is no closed-form expression for MLE.
Numerical optimization algorithms, such as the Newton method, are required to solve the aforementioned optimization problem.
For example, the Newton method requires iterative number of evaluations of $\boldsymbol{\theta}^{(m)}=\boldsymbol{\theta}^{(m-1)}-\{\ddot{l}(\boldsymbol{\theta}^{(m-1)})\}^{-1}{\dot{l}(\boldsymbol{\theta}^{(m-1)})}$ for $m=1,2,\cdots$. 
Such an algorithm will be time-consuming if the size of the data is large.
Hence, it is vital to develop a subsampling method that draws a small representative subsample from the original massive data and then approximate the MLE based on the subsample. 

Generally, a subsampling algorithm is implemented by first assigning subsampling probabilities for each data point, and then drawing samples iteratively according to the probabilities determined in the first stage \citep{Politis_1999} until the desired sample size is reached. Therefore, the determination of the subsampling probabilities is crucial because the quality of the subsample relies on the subsampling probabilities.
The significant effects of the subsampling probabilities on the quality of the selected subsample will be demonstrated in Section \ref{simulation} by comparing the proposed subsampling strategy with the uniform subsampling.
In cohort follows, we will introduce two different subsampling strategies and the corresponding subsampling-based estimators.

\subsection{General Subsampling Based Estimator}\label{subsec:gse}
As stated, the General Subsampling procedure is first assigning subsampling probabilities $\lbrace\pi_i\rbrace_{i=1,2,...,n}$ to the corresponding individuals, where $0<\pi_i<1$, $\sum_{i=1}^n \pi_i = 1$, and then drawing a subsample with size $r$ from the full data according to the subsampling probabilities with replacement. 
We assume $n$ is large and $r\ll n$ significantly, which is helpful to reduce the computation burden because enough data are discarded.

Let $\mathcal{D}_n^* = \lbrace (\boldsymbol{x}_i^*, \pi_i^*)=(t^*_i, C_i^*, t_{il}^*, \pi_i^*), i=1,2,...,r\rbrace$ denote the selected subsample. Then the General Subsampling based estimator $\tilde{\boldsymbol{\theta}}_g$ is obtained by maximizing the weighted pseudo log-likelihood of the subsample, i.e., $\tilde{\boldsymbol{\theta}}_g=\mathop{\arg\max}\limits_{\theta}l_g^*(\boldsymbol\theta)$,
\begin{equation}
    \begin{aligned}
     l_g^*(\boldsymbol\theta)=& \frac{1}{r}\sum_{i=1}^{r}\frac{1}{\pi_i^*}l^*(\boldsymbol{x}^*_i,\boldsymbol\theta)\\
     =&\frac{1}{r}\sum_{i=1}^{r}\frac{1}{\pi_i^*} \{(1-C_i^*)\log f(t_i^*;\boldsymbol\theta)+
     C_i^*\log\left[1-F(t_i^*;\boldsymbol\theta)\right]\\
     -&\log[1-F(t_{il}^*;\boldsymbol\theta)] \}.
    \end{aligned}
    \label{Equ:rlog-likelihood}
\end{equation}
The inverse probability weight ${\pi_i^*}^{-1}$ ensures an unbiased or asymptotically unbiased estimation, which will be shown in the theoretical analysis of Section \ref{section:theorem supports}.

It can be seen that the main challenge of general subsampling is evaluating the subsampling quality and then determining the optimal subsample probabilities $\pi_i^*$ in terms of the evaluation metric.
One possible evaluation is the accuracy of subsample estimate $\tilde{\boldsymbol{\theta}}_g$ and its (asymptotic) variance. 
Sections 3 and 4 will discuss how to determine the optimal subsample probabilities regarding the L-optimality criterion. 

In the simulations of Appendix \ref{extra_sim}, we will show that the performance of this strategy (over uniform subsampling) diminishes significantly when the censoring rate becomes quite high. 
One important reason is that when the censoring rate approaches $1$, the subsample will ignore some of the complete data despite the superiority of the subsampling probabilities on them. 
This causes the proposed subsampling methods to draw abundant samples from the censored subset of the full data, similar to uniform subsampling.
To address this limitation, a refined subsampling method tailored for extremely high censoring rates is presented in the subsequent subsection.

\subsection{Censoring Subsampling Based Estimator}\label{subsec:cse}

In the reliability analysis for highly reliable products, extremely high censoring rate data (complete data are rare) are frequently encountered (see motivating example).
When implementing the general subsampling procedure to extract a small subsample, chances are that some uncensored data will not be included in the subsample, making the uncensored data even rarer. 
Compared to the censored data, the uncensored data possesses a specific failure time rather than a range and may contain more information.
Thus, if dropping partial of them due to the subsampling strategy, the accuracy of the estimation based on the subsample might deteriorate drastically. 
Moreover, the numerical investigation shows that the general subsampling method exhibits instability and inaccuracy reflected by the high variance of estimators and some outrageous estimators deviating significantly from the true value.
To address these problems, a novel method called the Censoring Subsampling is proposed for the case where $r>n_0$.
Broadly speaking, we intend to include as much complete data as we can.
A two-stage approach is adopted to implement the new subsampling strategy as follows, which is simple and straightforward. 

\noindent\textbf{Censoring Subsampling}: 
\begin{itemize}
\item Initialize the desired size of the subsample $r$ $(r>n_0)$. Collect all uncensored data into the subsample.
\item According to the designed subsampling probabilities $\tilde{\boldsymbol{\pi}}$, we can apply the General Subsampling procedure to the censored data and select $r-n_0$ censored data into the subsample with replacement.
\end{itemize}
Let $\tilde{\boldsymbol{\pi}} = (\tilde{\pi}_1,...,\tilde{\pi}_{n})$ be the subsampling probabilities assigned to each data.
When $C_i=0$, let $\tilde{\pi}_i = 1$, which corresponds to collecting the complete data with probability $1$. 
Assume that $\sum_{C_i = 1} \tilde{\pi}_i =1$, i.e., the acceptance probability is applied to the censored dataset.
For notational convenience, define $\boldsymbol{\omega}=(\omega_1,...,\omega_n)$ as a weight vector.  
We set $\omega_i=1$ if $C_i = 0$ and $\omega_i=\frac{1}{(r-n_0)\tilde{\pi}_i}$ if $C_i = 1$. Hence, $\omega_i$ can be written as $\omega_i = \frac{1}{(1-C_i)+C_i (r-n_0)\tilde{\pi}_i} $. Note that in the rest of the paper, without loss of generality, we will let the first $n_0$ individuals be the uncensored data and the others be censored. That is, for $i=1,2,...,n_0$, $C_i=0$ and $\omega_i=1$, while for $i>n_0$, $C_i=1$ and $\omega_i=\frac{1}{(r-n_0)\tilde{\pi}_i}$. 

Let $\lbrace (\boldsymbol{x}_i^c, \tilde{\pi}_i^c)=(t^c_i, C_i^c, t_{il}^c, \tilde{\pi}_i^c), i=1,2,...,r\rbrace$ represent the information of $r$ selected data and the first $n_0$ data are uncensored. Besides, $\omega^c_i=\frac{1}{(1-C_i)+C_i (r-n_0)\tilde{\pi}^c_i}$ for all $i$. Then, based on the Censoring Subsampling procedure, the weighted pseudo log-likelihood equation can be written in the following form:
\begin{equation}
    \begin{aligned}
        l^{*}_c(\boldsymbol{\theta}) =
        & \sum_{i=1}^r \omega_i^c l^c(\boldsymbol{x}^c_i,\boldsymbol{\theta})\\
        =& \sum_{i=1}^{n_0} l^c(\boldsymbol{x}_i^c,\boldsymbol{\theta}) + \sum_{i=n_0+1}^{r} \frac{1}{(r-n_0)\tilde{\pi}_i^c}l^c(\boldsymbol{x}_i^c,\boldsymbol{\theta})
        \\
        = & \sum_{i=1}^r \omega_i^c \{(1-C_i^c)\log f(t_i^c;\boldsymbol\theta)+ C_i^c\log\left[1-F(t_i^c;\boldsymbol\theta)\right] \\
        - &\log\left[1-F(t_{il}^c;\boldsymbol\theta)\right]\}.
    \end{aligned}
\end{equation}
Through maximizing $l^*_{c}(\boldsymbol{\theta})$, we obtain the subsampling--based MLE, i.e., $\tilde{\boldsymbol{\theta}}_c = \mathop{\arg\max}\limits_{\boldsymbol{\theta}}l^*_{c}(\boldsymbol{\theta})$. For convenience, we let $l_i(\boldsymbol{\theta})=l(\boldsymbol{x}_i,\boldsymbol\theta)$, $l_i^*(\boldsymbol\theta) = l^*(\boldsymbol{x}^*_i,\boldsymbol\theta)$ and $l^{c}_i(\boldsymbol{\theta}) = l^{c}(\boldsymbol{x}^c_i,\boldsymbol{\theta})$ in this article.

\section{Theoretical Supports}\label{section:theorem supports}

Section 2 proposes two types of subsampling procedures. 
As stated, the accuracy of the MLE obtained from the subsample is affected by the chosen subsampling probabilities $\pi_i$.
Therefore, determining the \textit{optimal subsampling probabilities $\pi_i$} becomes a vital task. 
In this section, we investigate the theoretical properties of the two subsampling-based estimators, $\tilde{\boldsymbol{\theta}}_g$ and $\tilde{\boldsymbol{\theta}}_c$. 
The theoretical results proposed in this section are crucial for designing our subsampling algorithm which attempts to find the subsampling probabilities under which the trace of the asymptotic variance-covariance matrix of the estimator is minimized. 
The detailed algorithm will be provided in Section 4.

Note that in this article, if $A$ is a matrix, then $A^2 \triangleq AA^\top$. Besides, $\ddot{l}_{(jk)} (\boldsymbol{\theta}) \triangleq \frac{\partial^2 l(\boldsymbol{\theta})}{\theta_j \theta_k}$ is the second order derivative of the log-likelihood function. The following assumptions are necessary for the proof:

\begin{assumption}\label{ass1}
   The $\hat{\boldsymbol\theta}$ is unique. As $n\xrightarrow{}\infty$, $\boldsymbol{M}_g =\frac{1}{n}\sum_{i=1}^n  \ddot{l}_i(\hat{\boldsymbol\theta}) $ is a negative definite matrix in probability, and $\frac{1}{n^2}\sum_{i=1}^n \frac{1}{\pi_i}
   \Vert\ddot{l}_i(\hat{\boldsymbol{\theta}})\Vert^2 = O_{p|\mathcal{D}_n}(1)$.
\end{assumption}

\begin{assumption}\label{ass2}
    $\frac{1}{n^2}\sum_{i=1}^n\frac{1}{\pi_i} \Vert \dot{l}_i(\hat{\boldsymbol\theta}) \Vert^2=O_{p|\mathcal{D}_n}(1)$.
\end{assumption}
   
\begin{assumption}\label{ass3}
     For all $\boldsymbol{\theta} = \hat{\boldsymbol{\theta}} + s(\tilde{\boldsymbol{\theta}}_g - \hat{\boldsymbol{\theta}})$, $s \in [0,1]$, $\frac{1}{n}\left\Vert  \frac{\partial \ddot{l}^*_{g}(\boldsymbol{\theta})}{\partial \theta_i} \right\Vert=O_{p|\mathcal{D}_n}(1)$, $i=1,2,...,d$. 
\end{assumption}

\begin{assumption}\label{ass4}
   Exist $\delta>0$, such that $\frac{1}{n^{2+\delta}} \sum_{i=1}^n \frac{1}{\pi_i^{1+\delta}} \Vert \dot{l}_i (\hat{\boldsymbol{\theta}})\Vert^{2+\delta} = O_{p}(1)$. 
\end{assumption}
These assumptions require some bound restrictions on $l(\hat{\boldsymbol\theta})$ and $l^*(\hat{\boldsymbol\theta})$, which are typically easily satisfied. These constraints are naturally met when the failure time adheres to commonly used lifetime distributions such as the Weibull and exponential distributions.
Indeed, when the subsampling probability is uniform, i.e., $\pi_1=\pi_2 = \dots = \pi_n = 1/n$, Assumptions \ref{ass1}, \ref{ass2}, and \ref{ass4} are essentially upheld. For further discussions, please refer to Appendix \ref{appendix:Exmaples}.

\subsection{Theoretical Properties of General Subsampling}
Assume the data are randomly censored. Under General Subsampling strategy with the certain subsampling probabilities $\{\pi_i\}_{i=1,2,...,n}$, suppose the assumptions 1-4 holds. 
Conditional on the full data $\mathcal{D}_n$, the asymptotic property of $\tilde{\boldsymbol\theta}_g$ is given below:
\begin{thm} \label{thm:genenral-mle}
    Under Assumption \ref{ass1}-\ref{ass4}, conditionally on $\mathcal{D}_n$ in probability, as $n\xrightarrow{}\infty$ and $r\xrightarrow{}\infty$, the following holds:
    \begin{align}
        \sqrt{r}\boldsymbol{V}_g^{-\frac{1}{2}}(\tilde{\boldsymbol\theta}_g-\hat{\boldsymbol\theta})\xrightarrow{d}N(0,\boldsymbol{I}),
    \end{align} 
    where $\boldsymbol{V}_g=\boldsymbol{M}_g^{-1}
         \boldsymbol{\Lambda}_g
         \boldsymbol{M}_g^{-1}$
 and $\boldsymbol{\Lambda}_g=\frac{1}{n^2}\sum_{i=1}^n\frac{1}{\pi_i}\dot{l}_i(\hat{\boldsymbol\theta})^2$.
\end{thm}
Theorem \ref{thm:genenral-mle} shows that when $n$ and $r$ tend to infinity, the difference between $\tilde{\boldsymbol\theta}_g$ and $\hat{\boldsymbol\theta}$ diminishes. Note that we do not need $r/n \rightarrow 0$ to get this result. However, in practice, $r\ll n$ is necessary because we aim to reduce the computation burden. There are two reasons to approximate $\hat{\boldsymbol\theta}$ instead of the true value. First, the computation of $\hat{\boldsymbol\theta}$ involves all the available information and thus is supposed as one of the best estimations for the true parameter $\boldsymbol{\theta}_0$. Second, the true value is typically unavailable in practice, compared with $\hat{\boldsymbol\theta}$ provides an evaluation metric for the estimation. 
\begin{remark} \label{remark:general_true}
    It will be advantageous if  $\tilde{\boldsymbol\theta}_g$ possesses some asymptotic properties related to the true value $\boldsymbol{\theta}_0$. Moreover, these properties become valuable when providing interval estimations or other inferences for the true parameter. We derive the asymptotic normality property of $\tilde{\boldsymbol\theta}_g - \boldsymbol{\theta}_0$ in Appendix \ref{appendix:Theorems} (see Theorem \ref{thm:general-true}).
\end{remark}

In summary, Theorem \ref{thm:genenral-mle} and Theorem \ref{thm:general-true} in Appendix \ref{appendix:Theorems} illustrate that for any subsampling policy that satisfies some mild assumptions, the general subsampling-based weighted estimator from the subset can approximate the MLE of the full data or the true value fairly well, as long as $r\rightarrow \infty$, $r/n\rightarrow0$.  
The asymptotic variance is significantly affected by the choice of the subsampling probabilities. 
We will discuss in detail how to determine the subsampling probabilities to ensure the optimal performance subsampling-based weighted estimator in Section 4.1.

\subsection{Theoretical Properties of Censoring Subsampling}\label{subsec:Thm of CS}

This section investigates the theoretical properties of $\tilde{\boldsymbol{\theta}}_c$.
The data are assumed to be random censored. 
It should be pointed out that we only recommend the censoring subsampling method when $\alpha=n_0/n$ is small enough (close to $0$).
One sufficient condition to employ the censoring subsampling method is $n_0<r$, i.e., the size of the complete data set is smaller than the desired sample size of the subsample. 
Under such conditions, the properties of the estimator can still be maintained when using the censoring subsampling method, including consistency and asymptotic normality. 
The assumptions below are necessary for the theoretical results.  
\begin{assumption}\label{ass5}
    Assume $\hat{\boldsymbol\theta}$ is unique, $\boldsymbol{\theta} \in \boldsymbol{\Theta}$  and $\boldsymbol{\Theta}$ is compact.  For $\forall \boldsymbol{\theta} \in \boldsymbol{\Theta}$, $\mathbb{E}\{l^2(\boldsymbol{x},\boldsymbol{\theta})\}<\infty, \mathbb{E}\{l(\boldsymbol{x},\boldsymbol{\theta})\}<\infty$.
\end{assumption}

\begin{assumption} \label{ass6}
    Assume $-\mathbb{E}\{\ddot{l}(\boldsymbol{x},\boldsymbol{\theta}) \}$ is positive definite. For $\forall j,k=1,2,...,d$, $\mathbb{E}\{\ddot{l}^2_{(jk)}(\boldsymbol{x},\boldsymbol{\theta})\} < \infty$. Besides, there exists a function $\psi(\boldsymbol{x})$ with $\mathbb{E}\{\psi^2(\boldsymbol{x})\} < \infty$ such that for every $\boldsymbol{\theta}_1, \boldsymbol{\theta}_2 \in \Theta$, $\lvert\ddot{l}_{(jk)}(\boldsymbol{x},\boldsymbol{\theta}_1) - \ddot{l}_{(jk)}(\boldsymbol{x},\boldsymbol{\theta}_2) \rvert  \leq \psi(\boldsymbol{x}) \Vert \boldsymbol{\theta}_1 - \boldsymbol{\theta}_2 \Vert$.
\end{assumption}

\begin{assumption} \label{ass7}
    Assume $\boldsymbol{\Lambda}(\boldsymbol{\theta}) = \mathbb{E}\{\dot{l}(\boldsymbol{x},\boldsymbol{\theta})^2\}$ is positive definite, and $\exists \delta$ such that $\forall$$ \Vert \boldsymbol{\theta} - \hat{\boldsymbol{\theta}} \Vert < \delta$, $\frac{1}{n} \sum_{i=1}^n \Vert \dot{l}_i(\boldsymbol{\theta}) \Vert ^ 4 = O_p(1)$.
\end{assumption}

\begin{assumption} \label{ass8}
    The subsampling probabilities satisfy $\max_{i=n_0+1,...,n} (n\tilde{\pi}_i)^{-1} = O_p(1)$.
\end{assumption}

\begin{thm} \label{thm:censor-all}
    Under Assumptions \ref{ass5}-\ref{ass8}, conditionally on $\mathcal{D}_n$, when $n,r \rightarrow \infty$, $n_0/r=o(1)$, 
    \begin{equation}
        \sqrt{r}\boldsymbol{V}_c^{-\frac{1}{2}}(\tilde{\boldsymbol{\theta}}_c - \hat{\boldsymbol{\theta}}) \xrightarrow{d}N(0,\boldsymbol{I}),
    \end{equation}
    \justifying{where $\boldsymbol{V}_c = \boldsymbol{M}_g^{-1} \boldsymbol{\Lambda}_c \boldsymbol{M}_g^{-1}$ and $\boldsymbol{\Lambda}_c = \frac{1}{n^2} \sum_{i=n_0+1}^n \frac{1}{\tilde{\pi}_i}\dot{l}_i(\hat{\boldsymbol{\theta}})^2 - \left[ \frac{1}{n}\sum_{i=n_0+1}^n \dot{l}_i(\hat{\boldsymbol{\theta}}) \right]^2,$  
    $\dot{l}_i(\hat{\boldsymbol{\theta}})= \frac{\partial \log(1-F(t_i,\hat{\boldsymbol{\theta}}))}{\partial \boldsymbol{\theta}} -  \frac{\partial \log(1-F(t_{il},\hat{\boldsymbol{\theta}}))}{\partial \boldsymbol{\theta}}.$}
   
\end{thm}

The forms of the covariance matrix $\boldsymbol{V}_c$ reflect the significance of choosing subsampling probabilities $\tilde{\pi}_i$ in this case. Note that $\boldsymbol{\Lambda}_c$ only relates to the censored data. We also derive the unconditional asymptotic distribution between $\tilde{\boldsymbol{\theta}}_c$ and true parameter $\boldsymbol{\theta}_0$, which is given in Appendix \ref{appendix:Theorems} (see Theorem \ref{thm:censor-true}).

Similar to conclusion of the general subsampling, Theorem \ref{thm:censor-all} and Theorem \ref{thm:censor-true} in Appendix \ref{appendix:Theorems} show that the censoring subsampling-based estimator can be a good approximation of the MLE of the full data or the true value, as long as $n,r\rightarrow \infty$ and $n_0/r \rightarrow 0$. Both general subsampling and censoring subsampling estimators have a convergence rate $O_p(r^{-1/2})$ relating to the size of subsample $r$.
For a given computational budget, one is restricted to choosing a subsample with a sample size of no more than $r_M$, since a large $r$ will increase computation time.
Hence, it is crucial to select a set of `good' samples (with respect to some specific evaluation metrics) with the sample size constraint. 

\section{Optimal Subsampling Strategies}
In this section, we propose two subsampling methods designed for the reliability of data with moderate or high censoring rates, corresponding to the general and censoring subsampling strategy. 
The proposed subsampling method is designed based on the theoretical properties of the previous section. 
Section \ref{subsec:ogss} presents an optimal subsampling strategy for a moderate censoring rate case and 
Section \ref{opt_censor} presents an optimal censoring subsampling strategy for high censoring rates, which gives dramatically outstanding performance through numerical experiments.

\subsection{Optimal General Subsampling Strategy}\label{subsec:ogss}
Adopting the uniform subsampling strategy, i.e., $\pi_1= \cdots = \pi_n = 1/n$, is the simplest approach. However, the information content of each data point can vary significantly. 
Intuitively, it is reasonable to set a higher probability to the data point with more information. Additionally, we are eager to determine the optimal non-uniform subsampling probabilities (with respect to some specific criteria) such that more accurate estimates can be obtained with a smaller size of the subsample. The A-optimality criterion \citep{kiefer1959optimum} is used to minimize the trace of the asymptotic covariance of $\tilde{\boldsymbol{\theta}} - \hat{\boldsymbol{\theta}}$. This involves minimizing the trace of $\boldsymbol{V}_g$. Notably, this approach is equivalent to minimizing the asymptotic expectation of mean square error (MSE) of the resultant estimator. The following theorem provides the optimal subsampling probabilities that minimize $tr(\boldsymbol{V}_g)$.

\begin{thm} \label{ssp:general A-opt}
    If the subsampling probabilities $\boldsymbol{\pi}$ are given by 
    \begin{equation}
    \pi_i^{A} \propto  \left\Vert\boldsymbol{M}_g^{-1} \dot{l}_i(\hat{\boldsymbol{\theta}})\right\Vert, i=1,2,...,n,
    \end{equation}
    then the $tr(\boldsymbol{V}_g)$ reaches its minimum. 
\end{thm}

The calculation of $\boldsymbol{M}_g$ requires a second derivative of $l$, which is a heavy computation burden given that the log-likelihood $l$ is a summation of $n$ term.
Besides, $\boldsymbol{M}_g$ is a constant matrix which does not consist of subsampling probabilities $\{\pi_i\}_{i=1,2,...,n}$.
Thus, the alternative subsampling probabilities can be chosen by minimizing $tr(\boldsymbol{\Lambda}_g) = tr(\boldsymbol{M_g}\boldsymbol{V}_g \boldsymbol{M_g}) = tr(\boldsymbol{V}_g\boldsymbol{M_g}^2)$, corresponding to the L-optimality criterion \citep{Atkinson_2007}, which minimizes the trace of product of the asymptotic covariance matrix and a constant matrix. Compared to the former strategy, the latter is more computationally efficient and easier to implement. 
The specific subsampling probabilities are demonstrated below.

\begin{thm} \label{ssp:general L-opt}
    If the subsampling probabilities $\boldsymbol{\pi}$ are given by 
    \begin{equation}
        \pi_i^{RDS}(\hat{\boldsymbol{\theta}}) \propto \Vert  \dot{l}_i(\hat{\boldsymbol{\theta}}) \Vert, i=1,2,...,n,
    \end{equation}
    then the $tr(\boldsymbol{\Lambda}_g)$, reaches its minimum.
\end{thm}

\begin{remark}
Here are some intuitions on the choice of L-optimality criterion. 
      Loewner order \citep{Pukelsheim_Optimal_design} is the partial order defined by the convex cone of positive semi-definite matrices. Let $A$ and $B$ be two Hermitian matrices of order $n$, and define $A \geq B$ if $A-B$ is positive semi-definite. For any subsampling probabilities $\boldsymbol{\pi}_0$,  $\boldsymbol{V}_g(\boldsymbol{\pi}_0)$ and $\boldsymbol{\Lambda}_g(\boldsymbol{\pi}_0)$ are used if we replace $\boldsymbol{\pi}$ by $\boldsymbol{\pi}_0$ in $\boldsymbol{V}_g$ and $\boldsymbol{\Lambda}_g$.
      In this sense, for two different choices of subsampling probabilities $\boldsymbol{\pi}^a$ and $\boldsymbol{\pi}^b$, $\boldsymbol{V}_g(\boldsymbol{\pi}^a)\geq\boldsymbol{V}_g(\boldsymbol{\pi}^b) \iff \boldsymbol{\Lambda}_g(\boldsymbol{\pi}^a)\geq\boldsymbol{\Lambda}_g(\boldsymbol{\pi}^b)$. 
      This fact also explains the rationality of minimizing $tr(\boldsymbol{\Lambda}_g)$ instead of $tr(\boldsymbol{V}_g)$, because they are in the same order in the sense of Loewner order.
\end{remark}

Although we derive the form of subsampling probabilities, it contains full-data maximum likelihood estimation $\hat{\boldsymbol{\theta}}$, which cannot be obtained in practice. 
To address this, a three-step algorithm to approximate the subsampling probabilities is proposed. Firstly, a pilot estimation $\tilde{\boldsymbol{\theta}}_0$ is obtained from a subsample of $r_0$ ($r_0 \ll r$). Secondly, substituting $\tilde{\boldsymbol{\theta}}_0$ to $\hat{\boldsymbol{\theta}}$ to obtain the approximately optimal subsampling probabilities. Finally, use the subsampling probabilities obtained in the second step to choose a subsample and obtain the final estimation. This subsampling policy is named as RDS. 
The algorithm is presented below.

The constant $\xi$ ensures $n\pi_{\xi,i}^{RDS}\geq\xi,\quad i=1,2,...,n,$ so that $\max_{i=1,...,n}(n\pi_{\xi,i}^{RDS})^{-1} = O_{p|\mathcal{D}_n,\tilde{\boldsymbol{\theta}}_p}(1)$. This can prevent $(n\pi_{\xi,i}^{RDS})^{-1}$ from getting too large to make the estimation unstable. 
The asymptotic property of $\tilde{\boldsymbol{\theta}}_{RDS}$ towards the full-data MLE $\hat{\boldsymbol{\theta}}$ is derived as follows.
\begin{prop} \label{thm:RDS-mle}
    Assume the pilot estimation $\tilde{\boldsymbol{\theta}}_p$ is obtained from Algorithm \ref{ass1}.
    Under assumptions \ref{ass5}-\ref{ass7}, as $r_0, r$ and $n\rightarrow \infty$, conditionally on $\mathcal{D}_n$ and $ \tilde{\boldsymbol{\theta}}_p$, we have:
    \begin{equation}
        \sqrt{r}\boldsymbol{V}_{RDS}^{-\frac{1}{2}}(\tilde{\boldsymbol{\theta}}_{RDS} - \hat{\boldsymbol{\theta}}) \xrightarrow{d} N(0,\boldsymbol{I}),
    \end{equation}
where $\boldsymbol{V}_{RDS}=\boldsymbol{M}_{g}^{-1}\boldsymbol{\Lambda}_{RDS}\boldsymbol{M}_{g}^{-1}$, and $\boldsymbol{\Lambda}_{RDS}=\frac{1}{n^2}\sum_{i=1}^n \frac{1}{\pi_{\xi,i}^{RDS}(\hat{\boldsymbol{\theta}})}\dot{l}_i(\hat{\boldsymbol{\theta}})^2$.
\end{prop}

 \begin{algorithm}[H]
    \caption{Three-step algorithm for RDS} \label{algorithm1}
    \justifying{
        \begin{description}
        \item[Pilot estimation]Use uniform subsampling to obtain subsample of $r_0$ ($r_0 \ll r$). Obtain a pilot estimate $\tilde{\boldsymbol{\theta}}_{p}$ by
        calculating the pseudo maximum likelihood estimate of this subsample:
        $\tilde{\boldsymbol{\theta}}_{p} = \mathop{\mathrm{argmax}}\limits_{\boldsymbol{\theta}} \frac{1}{r_0}\sum_{i=1}^{r_0} l_{i}^p(\boldsymbol{\theta})$, where $l_{i}^p, i = 1,...,r_0$ is likelihood of the chosen data.
        \item[Approximately optimal subsampling probabilities]
        Replace $\hat{\boldsymbol{\theta}}$ with $\tilde{\boldsymbol{\theta}}_p$ in $\pi_i^{RDS}(\cdot)$. Then, let $\pi_{\xi,i}^{RDS}(\tilde{\boldsymbol{\theta}}_p)= (1-\xi)\pi_i^{RDS} (\tilde{\boldsymbol{\theta}}_p)+\frac{\xi}{n}$ to obtain approximately optimal subsampling probabilities for $i=1,2...,n$, where $\xi \in (0,1)$ is a constant.   
        \item[RDS-based estimation] Use $\pi_{\xi,i}^{RDS}(\tilde{\boldsymbol{\theta}}_p)$ to obtain a subsample of size $r$, with the corresponding subsampling probabilities $\pi_{\xi,i}^{RDS*}(\tilde{\boldsymbol{\theta}}_p)$ and likelihood $l_{i}^*(\boldsymbol{\theta})$ of the subsample, $i=1,...,r$. Calculating the estimate $\tilde{\boldsymbol{\theta}}_{RDS}$ by solving the weighted pseudo maximum likelihood of this subsample: $\tilde{\boldsymbol{\theta}}_{RDS} = \mathop{\mathrm{argmax}}\limits_{\boldsymbol{\theta}} \frac{1}{r}\sum_{i=1}^{r} \frac{1}{\pi_{\xi,i}^{RDS*}(\tilde{\boldsymbol{\theta}}_p)}l_{i}^*(\boldsymbol{\theta}). $
    \end{description}
    }    
\end{algorithm}

\begin{remark}\label{remark:lambda_g}
Let $\boldsymbol{\Lambda}_{g}^{opt}=\frac{1}{n^2}\sum_{i=1}^n \frac{1}{\pi_i^{RDS}(\hat{\boldsymbol{\theta}})}\dot{l}_i(\hat{\boldsymbol{\theta}})^2$ be the asymptotic covariance matrix with optimal subsampling probabilities, which minimizes the trace of $\boldsymbol{\Lambda}_{g}$.
Although $\boldsymbol{\Lambda}_{RDS}$ is different from the  $\boldsymbol{\Lambda}_{g}^{opt}$, we have
$tr(\boldsymbol{\Lambda}_{g}^{opt}) \leq tr(\boldsymbol{\Lambda}_{RDS}) \leq \frac{tr(\boldsymbol{\Lambda}_{g}^{opt})}{1-\xi}$.
Then, $tr(\boldsymbol{\Lambda}_{g}^{opt})$ and $tr(\boldsymbol{\Lambda}_{RDS})$ can be close enough if $\xi \rightarrow 0$. Therefore, the estimation $\tilde{\boldsymbol{\theta}}_{RDS}$ is reasonable because the trace of its asymptotic variance matrix is approximately optimal.
As we have derived general subsampling strategies, we wonder how censoring affects the optimal subsampling probabilities. In the simulations of Appendix \ref{extra_sim}, we will show that higher subsampling probabilities are assigned to complete data, which fits our intuition.

\end{remark}

\subsection{Optimal Censoring Subsampling Strategy} \label{opt_censor}
As discussed in Sections \ref{subsec:cse} and \ref{subsec:Thm of CS}, we proposed a new subsampling strategy for the high censoring rate case.
Similar to the General Subsampling situation, we obtain the optimal subsampling probabilities through L-optimality by minimizing the trace of $\boldsymbol{\Lambda}_{c}$. 
The optimal censoring sampling function is given below. 

\begin{thm} \label{ssp:cen-sub}
If the censoring subsampling probabilities are given by
\begin{equation}
\begin{aligned}
    \tilde{\pi}_i^{RDCS}(\hat{\boldsymbol{\theta}}) \propto \left\Vert \frac{\partial log(1-F(t_i^c,\hat{\boldsymbol{\theta}}))}{\partial \boldsymbol{\theta}} - \frac{\partial log(1-F(t_{il}^c,\hat{\boldsymbol{\theta}}))}{\partial \boldsymbol{\theta}}\right\Vert, \\
    i=n_0+1,...,n. 
\end{aligned}
\end{equation}
Then the $tr(\boldsymbol{\Lambda}_{c})$, reaches its minimum.
\end{thm}

We call this strategy RDCS. As the full-data MLE is contained in the form of censoring subsampling probabilities, a three-step algorithm is proposed as well.

\begin{algorithm}[H]
    \caption{Three-step algorithm for RDCS}
    \label{algorithm2} 
    \justifying{ \begin{description}
        \item[Pilot estimation] Use uniform subsampling to obtain subsample of $r_0$ ($r_0 \ll n$). Obtain a pilot estimation $\tilde{\boldsymbol{\theta}}_{cp}$ by
        calculating the pseudo maximum likelihood estimate of this subsample: $\tilde{\boldsymbol{\theta}}_{cp} = \mathop{\mathrm{argmax}}\limits_{\boldsymbol{\theta}} \frac{1}{r_0}\sum_{i=1}^{r_0} l_{ci}^p(\boldsymbol{\theta})$, where $l_{ci}^p$, $i = 1,2,...,r_0$ is likelihood of the chosen data.
        \item[Approximately optimal subsampling probabilities]
        Replace $\hat{\boldsymbol{\theta}}$ with $\tilde{\boldsymbol{\theta}}_{cp}$ in  $\tilde{\pi}_i^{RDCS}$. Then, let $\tilde{\pi}_{\xi,i}^{RDCS}(\tilde{\boldsymbol{\theta}}_{cp})= (1-\xi_c)\tilde{\pi}_i^{RDCS} (\tilde{\boldsymbol{\theta}}_{cp})+\frac{\xi_c}{n-n_0}$ to obtain approximately optimal subsampling probabilities, where $\xi_c \in (0,1)$ is a constant.
        The corresponding weights are $\omega_{\xi,i}^{RDCS}(\tilde{\boldsymbol{\theta}}_{cp}) = \frac{1}{(1-C_i)+C_i (r-n_0) \tilde{\pi}_{\xi,i}^{RDCS}(\tilde{\boldsymbol{\theta}}_{cp})},$ for $i=1,2,...,n$. 
        \item[RDCS-based estimation]Use the Censoring Subsampling with subsampling probabilities $\tilde{\pi}_{\xi,i}^{RDCS}(\tilde{\boldsymbol{\theta}}_{cp})$ to obtain a subsample of size $r$. $\omega_{\xi,i}^{c,RDCS}(\tilde{\boldsymbol{\theta}}_{cp})$, $\tilde{\pi}_{\xi,i}^{c,RDCS}(\tilde{\boldsymbol{\theta}}_{cp})$ and $l_{i}^c(\boldsymbol{\theta})$ are the corresponding weights, subsampling probabilities and likelihood of the subsample. Let the first $n_0$ data to be the uncensored data and calculate the estimation $\tilde{\boldsymbol{\theta}}_{RDCS}$ by solving the weighted pseudo maximum likelihood equation: $\tilde{\boldsymbol{\theta}}_{RDCS} = \mathop{\mathrm{argmax}}\limits_{\boldsymbol{\theta}} \frac{1}{r}\sum_{i=1}^r\omega_{\xi,i}^{c,RDCS} (\tilde{\boldsymbol{\theta}}_{cp})l_{i}^c(\boldsymbol{\theta}).$
         \end{description}}
\end{algorithm}

Similar to $\xi$, the constant $\xi_c$ makes sure $(n\tilde{\pi}_{\xi,i}^{RDCS})^{-1}$ are bounded for $i=n_0+1,...,n$, which adds robustness to the estimation. 
The estimator $\tilde{\boldsymbol{\theta}}_{RDCS}$ from the aforementioned algorithm retains the asymptotic property, as demonstrated below.

\begin{prop}\label{thm:RDCS-mle}
   Assume the pilot estimation $\tilde{\boldsymbol{\theta}}_{cp}$ is obtained from Algorithm \ref{algorithm2}.
   Under assumptions \ref{ass5}-\ref{ass7}, as $r_0, r,n\rightarrow \infty$ and $n_0/r \rightarrow 0$, conditionally on $\mathcal{D}_n$ and $\tilde{\boldsymbol{\theta}}_{pc}$, we have:
    \begin{equation}
        \sqrt{r}\boldsymbol{V}_{RDCS}^{-\frac{1}{2}}(\tilde{\boldsymbol{\theta}}_{RDCS} - \hat{\boldsymbol{\theta}}) \xrightarrow{d}N(0,\boldsymbol{I}),
    \end{equation}
    \justifying{where $\boldsymbol{V}_{RDCS} = \boldsymbol{M}_g^{-1} \boldsymbol{\Lambda}_{RDCS} \boldsymbol{M}_g^{-1}$, and $\boldsymbol{\Lambda}_{RDCS} = \frac{1}{n^2} \sum_{i=n_0+1}^n \frac{1}{\tilde{\pi}_{\xi,i}^{RDCS}(\hat{\boldsymbol{\theta}})}\dot{l}_i(\hat{\boldsymbol{\theta}})^2  -\\ \left[ \frac{1}{n}\sum_{i=n_0+1}^n \dot{l}_i(\hat{\boldsymbol{\theta}}) \right]^2$,  $\dot{l}_i(\hat{\boldsymbol{\theta}}) = \frac{\partial \log(1-F(t_i,\hat{\boldsymbol{\theta}}))}{\partial \boldsymbol{\theta}} -  \frac{\partial \log(1-F(t_{il},\hat{\boldsymbol{\theta}}))}{\partial \boldsymbol{\theta}}$.}
\end{prop} 

\begin{remark}\label{remark:lambda_c}
Let $\boldsymbol{\Lambda}_{c}^{opt}=\frac{1}{n^2} \sum_{i=n_0+1}^n \frac{\dot{l}_i(\hat{\boldsymbol{\theta}})^2}{\tilde{\pi}^{RDCS}_i(\hat{\boldsymbol{\theta}})}  - \left[ \frac{1}{n}\sum_{i=n_0+1}^n \dot{l}_i(\hat{\boldsymbol{\theta}}) \right]^2$ be the asymptotic covariance matrix with optimal subsampling probabilities, which minimizes the trace of $\boldsymbol{\Lambda}_{c}$.
Straightforward calculations can lead to $tr(\boldsymbol{\Lambda}_{c}^{opt}) \leq tr(\boldsymbol{\Lambda}_{RDCS}) \leq \frac{tr(\boldsymbol{\Lambda}_{c}^{opt})}{1-\xi_c}$.
Thus, $tr(\boldsymbol{\Lambda}_{c}^{opt})$ and $tr(\boldsymbol{\Lambda}_{RDCS})$ can be arbitrarily close if $\xi_c$ is small enough. 
\end{remark}

\subsection{Variance Estimation} \label{Variance Estimation}
Proposition \ref{thm:RDS-mle} and Proposition \ref{thm:RDCS-mle} present the variance estimations for RDS and RDCS respectively. 
However, evaluation of these asymptotic covariance matrices involves the full-data MLE $\hat{\boldsymbol{\theta}}$, which is unknown in practice and requires calculations on the full data, thus time-consuming.
To provide a reasonable computationally efficient uncertainty estimation, we derive the formula that only involves the selected subsample to estimate the asymptotic covariance matrix in these two theorems.
In particular, by substituting $\hat{\boldsymbol{\theta}}$ with the subsampling-based estimator $\tilde{\boldsymbol{\theta}}_{RDS}$ or $\tilde{\boldsymbol{\theta}}_{RDCS}$, the estimated covariance only depends on the selected subsample the subsampling-based MLE. 
Through these approximations, the following results are calculated, which present the estimations of the covariance matrix for the two methods.

\textbf{RDS}: 
\begin{equation*}
\hat{\boldsymbol{V}}_{RDS}=\hat{\boldsymbol{M}}_{RDS}^{-1}\hat{\boldsymbol{\Lambda}}_{RDS}\hat{\boldsymbol{M}}_{RDS}^{-1},
\end{equation*}

where 
\begin{equation*}
\begin{aligned}
\hat{\boldsymbol{M}}_{RDS} &= \frac{1}{nr}\sum_{i=1}^r\frac{1}{\pi_{\xi,i}^{RDS*}(\tilde{\boldsymbol{\theta}}_p)}\ddot{l}^*_i(\tilde{\boldsymbol{\theta}}_{RDS}),\\
\hat{\boldsymbol{\Lambda}}_{RDS}&=\frac{1}{n^2r} \sum_{i=1}^r \frac{1+n\gamma \pi_{\xi,i}^{RDS*}(\tilde{\boldsymbol{\theta}}_p)}{\pi_{\xi,i}^{RDS*}(\tilde{\boldsymbol{\theta}}_p)^2}\dot{l}^*_i(\tilde{\boldsymbol{\theta}}_{RDS})^2.
\end{aligned}
\end{equation*}

\textbf{RDCS}:
\begin{equation*}
     \hat{\boldsymbol{V}}_{RDCS} = \hat{\boldsymbol{M}}_{RDCS}^{-1} \hat{\boldsymbol{\Lambda}}_{RDCS} \hat{\boldsymbol{M}}_{RDCS}^{-1},
\end{equation*}

where
\begin{equation*}
\begin{aligned}
    \hat{\boldsymbol{M}}_{RDCS} &= \frac{1}{nr}\sum_{i=1}^{r}\omega_{\xi,i}^{c,RDCS}(\tilde{\boldsymbol{\theta}}_{cp})\ddot{l}^c_i(\tilde{\boldsymbol{\theta}}_{RDCS}),\\
    \hat{\boldsymbol{\Lambda}}_{RDCS} &= \frac{1}{n^2r} \sum_{i=n_0+1}^{r}\frac{1+n\gamma \tilde{\pi}_{\xi,i}^{c,RDCS}(\tilde{\boldsymbol{\theta}}_{cp})}{\tilde{\pi}_{\xi,i}^{c,RDCS}(\tilde{\boldsymbol{\theta}}_{cp})^2}\dot{l}^c_i(\tilde{\boldsymbol{\theta}}_{RDCS})^2 \\
    &-\left[ \frac{1}{nr}\sum_{i=n_0+1}^r \frac{\dot{l}^c_i(\tilde{\boldsymbol{\theta}}_{RDCS})}{\tilde{\pi}_{\xi,i}^{c,RDCS}(\tilde{\boldsymbol{\theta}}_{cp})} \right]^2.
\end{aligned}
\end{equation*}

Note that compared to $\boldsymbol{\Lambda}_{RDS}$ and $\boldsymbol{\Lambda}_{RDCS}$, there are extra terms consist of $\gamma = r/n$ in $\hat{\boldsymbol{\Lambda}}_{RDS}$ and $\hat{\boldsymbol{\Lambda}}_{RDCS}$, which are related to $\gamma\boldsymbol{\Lambda}(\boldsymbol{\theta}_0)$ (defined in Assumption \ref{ass7}) in Theorem \ref{thm:general-true} and \ref{thm:censor-true} (in Appendix \ref{appendix:Theorems}). They can be neglected as we mainly focus on the case where $r/n\rightarrow0$.
It is easy to see that $\hat{\boldsymbol{V}}_{RDS}$ and $\hat{\boldsymbol{V}}_{RDCS}$ are computationally efficient since they only depend on the selected subsample.  
The performance of these estimations will be demonstrated in Section \ref{simulation}.

\section{Case Study - Backblaze Hard Drive Data Analysis} \label{subsec:case}
In this section, the proposed subsampling methods RDS and RDCS are applied to the Backblaze Hard Drive Data, as mentioned in Section \ref{intro}. We will show that the proposed methods can handle this real-world problem well.
Backblaze is a company that offers data storage services. It has been collecting the daily stats of all their hard drives since 2013. The data will not stop collecting until the drive fails. Besides, when the drive fails, it will be replaced by a new one. The data used are from 2016 to 2018. As multiple models of hard drives are installed, the data from the \textit{ST4000DM000} model are extracted to ensure their lifetime follows the identical distribution. 
The RDS, RDCS, and uniform subsampling are applied to these three datasets of different periods. 

For the lifetime distribution model, we adopt the GLFP distribution, which is shown to perform well in fitting the hard drives' lifetime distribution \citep{mittman2019}.  
\citet{mittman2019} stated that a fixed $\pi_g$ across the different models have little influence to the GLFP distribution, so we set $\pi_g=0.054$ which is recommended in that article and estimate the other four parameters, i.e., $\beta_1, \beta_2, \eta_1$ and $\eta_2$. Considering the extended duration of hard drive lifetimes, all time values are normalized by dividing them by $10^6$ to maintain parameter values within a practical range.
As $\xi$ and $\xi_c$ need to be small to ensure the properties in Remark \ref{remark:lambda_g} and \ref{remark:lambda_c}, they are simply set to be 0.1.
Throughout this section, $n$ denotes the sample size of the full data,
$r$ denotes the number of subsample, while $r_0$ denotes the sample size of pilot data. 
The experiments are repeated 500 times.

\textbf{Evaluation metrics:}
The root mean square error (RMSE) of the parameter estimation is calculated as an evaluation metric of the proposed and alternative methods. 
Let $\breve{\boldsymbol{\theta}}$ and $\boldsymbol{\theta}_0$ denote the estimator and the true value, respectively. $\breve{\boldsymbol{\theta}}^{(b)}$ is the estimation from the $b$th replicate, then the RMSE is give by:
$\text{RMSE} = \frac{1}{m}  (\sum_{b=1}^m\Vert \breve{\boldsymbol{\theta}}^{(b)} -  \boldsymbol{\theta}_0\Vert^2)^{\frac{1}{2}}.$ In this section, as the true value cannot be captured, we calculate the RMSE between the estimator and full-data MLE.
A lower RMSE indicates that the corresponding estimator is more accurate. 
Moreover, empirical biases (Bias) of the parameter estimations are calculated. 
To compare the computational efficiency of the alternative methods, we compare the mean computational time (Time) over the $m$ simulations. 

\begin{table*}[!ht]
\centering 
\caption{The results of analyzing Backblaze HDD in 2016-2018 using three methods. In 2016, $n=35678$ and $\alpha=0.97$. In 2016-2017, $n=36129$ and $\alpha=0.95$. In 2016-2018, $n=36153$ and $\alpha=0.93$. $r=3000$ and $r_0=500$ in all cases.} \label{backblaze}

\begin{tabular}{ccccccccccccc} \toprule[\heavyrulewidth]
\multirow{2}*{Year}& & \multicolumn{3}{c}{RMSE} & \multicolumn{3}{c}{Bias} &\multicolumn{4}{c}{Time(s)}
\\ \cmidrule(lr){3-5} \cmidrule(lr){6-8}\cmidrule(lr){9-12} 
&& RDS & RDCS& UNIF & RDS & RDCS & UNIF & RDS & RDCS & UNIF & FULL
\\ \midrule[\heavyrulewidth]
\multirow{4}*{2016}&$\beta_1$ & 1.3161 & \textbf{0.2152} &5.7042 & 0.2793 & \textbf{0.0101} & 1.3106 & \multirow{4}*{3.076} & \multirow{4}*{3.047} & \multirow{4}*{2.231} & \multirow{4}*{23.80} 
\\ 
&$\eta_1$ & \textbf{0.0140} & 0.1936 & 4.2673 & \textbf{0.0015}& 0.0071 & 0.4035 & & &&
\\ 
&$\beta_2$ & 0.1034 & \textbf{0.0695} & 0.2865 & \textbf{0.0107} & -0.0430 & 0.0342 & & &&
\\ 
&$\eta_2$& 8.0213 & \textbf{1.8952} & 78.736 & 1.5650 & \textbf{1.2357} & 7.7066 & & & & 
\\ \hline
\multirow{4}*{2016-2017} & $\beta_1$ & 1.9026 & \textbf{0.9384} & 4.9515 & 0.7247 & \textbf{0.5545} & 1.6116 & \multirow{4}*{3.135} & \multirow{4}*{3.028} & \multirow{4}*{2.624} & \multirow{4}*{19.76} 
\\ 
&$\eta_1$ & 0.2950 & \textbf{0.0165} & 0.6243 & 0.0544 & \textbf{0.0131} & 0.1161 & & &&
\\ 
&$\beta_2$ & 0.1581 & \textbf{0.1246} & 0.2415 & 0.1062 & \textbf{0.1013} & 0.1293 & & &&
\\ 
&$\eta_2$& 2.0067 & \textbf{1.8865} & 4.0752 & -1.4108 & -1.5899 & \textbf{-0.7573} & & & & 
\\ \hline
\multirow{4}*{2016-2018} & $\beta_1$ & 1.2053 & \textbf{1.0274} & 2.9654 & 0.1004 & \textbf{-0.0687} & 0.6310 & \multirow{4}*{3.556} & \multirow{4}*{2.976} & \multirow{4}*{2.332} & \multirow{4}*{34.15} 
\\ 
&$\eta_1$ & \textbf{0.0782} & 0.0965 & 0.3405 & \textbf{0.0003}& 0.0028 & 0.0400 & & &&
\\ 
&$\beta_2$ & 0.0588 & \textbf{0.0459} & 0.1908 &-0.0111 & \textbf{-0.0107} & 0.0157 & & &&
\\ 
&$\eta_2$& 1.4079 & \textbf{0.9709} & 4.6551 &0.4794 & \textbf{0.3411} & 1.3172& & & & 
\\ \bottomrule[\heavyrulewidth]
\end{tabular}
\end{table*}

\begin{table}[!ht]
\centering 
\setlength{\belowcaptionskip}{2mm}
\caption{The estimated MTTF (hours) in 2016 and 2017 using RDS, RDCS, uniform subsampling and full data} \label{MTTF}
\begin{tabular}{ccccc} \toprule[\heavyrulewidth]
year & RDS & RDCS & UNIF &FULL
\\ \midrule[\heavyrulewidth]
2016 &1083772&1020375&4036786&1050008
\\ 
2017 &2577379&1654480&2623846&1485223
\\ \bottomrule[\heavyrulewidth]
\end{tabular}
\end{table}

Table \ref{backblaze} gives the results when $r=3000$ and $r_0=500$.
The RMSE of RDCS and RDS are significantly smaller than that of the uniform subsampling method.
Moreover, RDCS has the smallest RMSE in almost all situations, and the RMSE of RDS is close to RDCS.
For the empirical bias of the parameter estimation, RDCS performs the best most of the time, and RDS is comparable to RDCS in some cases. 
The RMSE and biases of both proposed methods are smaller than those of the uniform subsampling method in magnitudes, especially on the estimation of $\eta_2$ in 2016.
Besides, the results of mean calculation time illustrate that the RDS and RDCS can reduce computation time dramatically. Table \ref{MTTF} shows the estimated MTTF using the data from 2016 and 2017.
The MTTF is calculated by: $\text{MTTF} = \int_{0}^{\infty} R(t)dt,$ where $R(t)=1-F(t)$ is the reliability function and $F(t)$ is the cumulative distribution function (c.d.f). We use the subsampling-based estimation to construct the c.d.f and calculate the estimated MTTF. For reference, it is reported on the internet that the MTTF of \textit{ST4000DM000} is 1.2 million hours\footnote{https://www.diskdrivefinder.com/product/108160/.}. The results show that for each method we employ, the estimated MTTF for hard drives installed in 2017 is significantly larger than those installed in 2016, possibly resulting from improved product quality. The estimated MTTF based on RDCS is the closest to the MTTF based on full data, reflecting the effectiveness of our method. The estimated MTTF based on RDS does not perform as well as RDCS because of the high censoring rate. 
Consequently, our proposed methods demonstrate superior efficiency and accuracy in practical applications, offering substantial time savings compared to uniform subsampling and full data methods. 

\section{Simulation Studies} \label{simulation}
In this section, we demonstrate the performance of the proposed subsampling methods via various simulation examples. 
The uniform subsampling estimator and the full-data estimator are used for comparison to show the accuracy and efficiency of our methods. 
To evaluate the performance of the alternative methods for estimating parameters in various lifetime models, three different distributions are considered for the failure time: exponential distribution, Weibull distribution and GLFP distribution, as detailed introduced in Section 2.1. 

\textbf{Synthetic data generation:}
Generate $n$ ($n=10^6$ throughout this section except for the specific statement) left-truncated time from a uniform distribution $U(a,b)$, failure time data according to the corresponding distribution, censored time from another uniform distribution $U(c,d)$ ($a<b<c<d$ are prespecified and help to adjust the censoring rate) respectively. 
The sizes of censored time and failure time determine the value of $(t_i,C_i)_{i=1,2,...,n}$. 
We still use $n$, $r$ and $r_0$ to denote the sample size of the full data, subsample and pilot data. 
Each simulation is repeated $m=500$ times. 

Empirically, RDS is recommended to use when $n_0>r/2$ while RDCS is recommended when $n_0<r/2$. 
As for the evaluation metrics, the RMSE and empirical bias between the subsampling-based estimator and true value are calculated. The average calculation time of RDCS, RDS, and uniform subsampling, as well as the calculating time of the full data estimation, are also given.
To evaluate the performance of the uncertainty estimation, we compute the 95\% empirical coverage probability (CP) towards the MLE.

\begin{table*}[!ht] 
\centering
\caption{The RMSE of subsampling-based estimators with different $\xi$ or $\xi_c$, where $r = 1000$, $r_0 = 400$. $\alpha = 0.90$ when using RDS and $\alpha = 0.9993$ when using RDCS.}
\label{table:zeta}
\begin{tabular}{cccccccccccc}
\toprule[\heavyrulewidth]
\multicolumn{2}{l}{} & \multicolumn{5}{c}{RDS} & \multicolumn{5}{c}{RDCS} \\ \cmidrule(lr){3-7} \cmidrule(lr){8-12}
\multicolumn{2}{l}{} & $\xi = 0$ & \multicolumn{1}{c}{$\xi = 0.1$} & \multicolumn{1}{c}{$\xi = 0.2$} & \multicolumn{1}{c}{$\xi = 0.3$} & \multicolumn{1}{c}{$\xi = 0.4$} & \multicolumn{1}{c}{$\xi_c = 0$} & \multicolumn{1}{c}{$\xi_c = 0.1$} & \multicolumn{1}{c}{$\xi_c = 0.2$} & \multicolumn{1}{c}{$\xi_c = 0.3$} & \multicolumn{1}{c}{$\zeta = 0.4$} \\ \midrule[\heavyrulewidth]
Exp & $\theta$ & 0.0064 & \textbf{0.0060} & 0.0063 & 0.0061 &0.0065 & 0.0021 & \textbf{0.0021} & 0.0022 & 0.0022 & 0.0023 \\ \hline
\multicolumn{1}{r}{\multirow{2}{*}{Weibull}} & $\beta$ & 0.0854 &	0.0824 & \textbf{0.0815} & 0.0886& 	0.0937 
 & 0.0763 & 0.0621 & \textbf{0.0606} & 0.0641 & 0.0633 \\
\multicolumn{1}{r}{} & $\eta$ & 0.2605 & \textbf{0.2451} & 0.2638& 0.2857& 0.2700 & 0.0620 & 0.4417 & \textbf{0.4330} & 0.4563 & 0.4559 \\ \hline
\multirow{4}{*}{GLFP} & $\beta_1$ & \textbf{0.1654} &0.1755& 0.2114& 0.1956& 0.2278 & 0.4412 & \textbf{0.0474} & 0.1741 & 0.0596 & 0.0806 \\
 & $\eta_1$ & 0.3927 & \textbf{0.3828} & 0.5720 & 0.6323 & 0.5519 & 8.3427 & 0.4336 & 1.4622 & \textbf{0.2436} & 0.3155 \\
 & $\beta_2$ & 0.6427 &\textbf{0.5555} &0.6871 &0.7700 &0.8974 & 2.5522 & 1.7180 & 1.9134 & 1.6815 & \textbf{1.2824} \\
 & $\eta_2$ & 0.3675 &\textbf{0.2854} &0.4090 &0.4851 &0.5032 & 2.5865 & 1.4188 & 1.8557 & \textbf{1.2421} & 1.4208 \\ 
 \bottomrule[\heavyrulewidth]
\end{tabular}
\end{table*}

In the first simulation, we decide to find the optimal $\xi$ and $\xi_c$ given in Algorithm \ref{algorithm1} and \ref{algorithm2}. In Table \ref{table:zeta}, we present the RMSE of subsampling-based estimators with different $\xi$ and $\xi_c$. Generally, the RDS-based estimations have a relatively smaller RMSE when $\xi = 0.1$. It can be seen that if $\xi_c=0$, the RMSE of the RDCS-based estimators will be relatively large, which is unacceptable. And when $\xi_c=0.1, 0.2, 0.3$, the results of RDCS are similar. For convenience, we simply let $\xi,\xi_c=0.1$ to retain the properties of the proposed optimal subsampling probabilities to the greatest extent.

After that, the performance of the RDS method is illustrated. 
For all subsampling methods considered in this example, $r=1000$, $r_0=400$, $\alpha=0.90$ are fixed and all three lifetime distributions introduced at the beginning of this section are considered.
Table \ref{table:RDS} presents the RMSE, Bias, CP, and Times of RDS and the uniform subsampling method under different settings. 
The results presented in this table show that nearly all the RMSE and biases of estimates based on RDS are significantly less than those based on uniform subsampling, illustrating the superior performance of the proposed RDS method.

It is worth noting that when the number of parameters increases, the RMSE of uniform subsampling-based estimation becomes unstable, which is shown in the estimation of  $\eta_1$ when $r = 1000$ in Table \ref{table:RDS}.
As a comparison, the RDS-based estimation performs well in all situations. 
To compare the computational efficiency of alternative methods, we report the computational time of the RDS method, the uniform subsampling method, and the full data MLE. 
The computation time of RDS is much faster than the full-data estimators. 
Given the outstanding estimation results of RDS, a little slower than uniform subsampling is acceptable.

\begin{table*}[!ht]
\caption{The simulation results of RDS and uniform subsampling under three distributions, where $r_0=400$, $\alpha=0.90$.}
\centering
\label{table:RDS}
\begin{tabular}{cccccccccccc}\toprule[\heavyrulewidth]
\multirow{2}{*}{} & \multirow{2}{*}{r} & \multirow{2}{*}{} & \multicolumn{2}{c}{RMSE} & \multicolumn{2}{c}{Bias} & \multicolumn{2}{c}{CP} & \multicolumn{3}{c}{Time(s)} \\ 
\cmidrule(lr){4-5} \cmidrule(lr){6-7} \cmidrule(lr){8-9} \cmidrule(lr){10-12}
 &  &  & RDS & UNIF & RDS & UNIF & RDS & UNIF & RDS & UNIF & FULL \\ \midrule[\heavyrulewidth]
\multirow{3}{*}{Exp} & 1000 & \multirow{3}{*}{$\theta$} & \textbf{0.0065} & 0.0104 & 0.0005 & 0.0003 & 0.950 & 0.948 & 0.065 & 0.005 & \multirow{3}{*}{0.687} \\
 & 1500 &  & \textbf{0.0050} & 0.0076 & \textbf{0.0002} & 0.0002 & 0.954 & 0.956 & 0.065 & 0.005 &  \\
 & 2000 &  & \textbf{0.0044} & 0.0066 & \textbf{0.0002} & -0.0003 & 0.948 & 0.942 & 0.066 & 0.005 &  \\ \hline 
 
\multirow{7}{*}{Weibull} & \multirow{2}{*}{1000} & $\beta$ & \textbf{0.0810} & 0.1835 & \textbf{0.0163} & 0.0221 & 0.958 & 0.962 & \multirow{2}{*}{0.787} & \multirow{2}{*}{0.025} & \multirow{7}{*}{21.44} \\
 &  & $\eta$ & \textbf{0.2450} & 0.4781 & \textbf{-0.0361} & 0.0378 & 0.958 & 0.954 &  &  &  \\
\noalign{\vskip 1.5mm}
 & \multirow{2}{*}{1500} & $\beta$ & \textbf{0.0689} & 0.1694 & \textbf{0.0158} & 0.0290 & 0.960 & 0.962 & \multirow{2}{*}{0.795} & \multirow{2}{*}{0.038} &  \\
 &  & $\eta$ & \textbf{0.2017} & 0.4039 & \textbf{-0.0238} & -0.0408 & 0.956 & 0.958 &  &  &  \\
\noalign{\vskip 1.5mm}
 & \multirow{2}{*}{2000} & $\beta$ & \textbf{0.0593} & 0.1386 & \textbf{0.0148} & 0.0254 & 0.924 & 0.938 & \multirow{2}{*}{0.799} & \multirow{2}{*}{0.049} &  \\
 &  & $\eta$ & \textbf{0.1765} & 0.3430 & -0.0428 & -0.0258 & 0.946 & 0.956 &  &  &  \\ \hline
 
\multirow{12}{*}{GLFP} & \multirow{4}{*}{1000} & $\beta_1$ & \textbf{0.4737} & 0.9863 & \textbf{-0.0856} & -0.2630 & 0.942 & 0.948 & \multirow{4}{*}{6.038} & \multirow{4}{*}{0.266} & \multirow{12}{*}{202.614} \\
 &  & $\eta_1$ & \textbf{1.3458} & 8.3930 & \textbf{0.2370} & 1.1369 & 0.942 & 0.992 &  &  &  \\
 &  & $\beta_2$ & \textbf{1.1644} & 1.7282 & \textbf{-0.0327} & -0.1911 & 0.940 & 0.952 &  &  &  \\
 &  & $\eta_2$ & \textbf{0.7246} & 1.1810 & \textbf{0.1154} & 0.3818 & 0.958 & 0.956 &  &  &  \\
\noalign{\vskip 1.5mm}
& \multirow{4}{*}{1500} & $\beta_1$ & \textbf{0.3942} & 0.7665 & \textbf{-0.0432} & -0.2232 & 0.944 & 0.946 & \multirow{4}{*}{9.831} & \multirow{4}{*}{0.414} &  \\
 &  & $\eta_1$ & \textbf{1.4140} & 1.7575 & \textbf{0.2861} & 0.3730 & 0.940 & 0.980 &  &  &  \\
 &  & $\beta_2$ & \textbf{1.1686} & 1.4113 & \textbf{0.0922} & -0.2377 & 0.952 & 0.952 &  &  &  \\
 &  & $\eta_2$ & \textbf{0.6675} & 0.9708 & \textbf{0.0919} & 0.3039 & 0.944 & 0.944 &  &  &  \\
\noalign{\vskip 1.5mm}
 & \multirow{4}{*}{2000} & $\beta_1$ & \textbf{0.2527} & 0.6482 & \textbf{-0.0244} & -0.1512 & 0.956 & 0.944 & \multirow{4}{*}{10.606} & \multirow{4}{*}{0.570} &  \\
 &  & $\eta_1$ & \textbf{0.6035} & 0.9578 & \textbf{0.1375} & 0.1485 & 0.946 & 0.996 &  &  &  \\
 &  & $\beta_2$ & \textbf{1.0983} & 1.1749 & \textbf{-0.0075} & -0.3192 & 0.948 & 0.962 &  &  &  \\
 &  & $\eta_2$ & \textbf{0.5256} & 0.8643 & \textbf{0.0729} & 0.3349 & 0.948 & 0.946 &  &  & \\
\bottomrule[\heavyrulewidth]
\end{tabular}
\end{table*}

\begin{table*}[!ht]
\centering 
\caption{The simulation results of RDCS and uniform subsampling under three distributions, where $r_0=400$, $\alpha=0.9993$.}
\label{table:RDCS}
\begin{tabular}{cccccccccccc}
\toprule[\heavyrulewidth]
\multirow{2}{*}{} & \multirow{2}{*}{r} & \multirow{2}{*}{} & \multicolumn{2}{c}{RMSE} & \multicolumn{2}{c}{Bias} & \multicolumn{2}{c}{CP} & \multicolumn{3}{c}{Time(s)} \\
\cmidrule(lr){4-5} \cmidrule(lr){6-7} \cmidrule(lr){8-9} \cmidrule(lr){10-12}
 &  &  & RDCS & UNIF & RDCS & UNIF & RDCS & UNIF & RDCS & UNIF & FULL \\ \midrule[\heavyrulewidth]
\multirow{3}{*}{Exp} & 1000 & \multirow{3}{*}{$\theta$} & \textbf{0.0012} & 0.1238 & \textbf{0.0012} & -0.0072 & 0.944 & 0.972 & 0.007 & 0.005 & \multirow{3}{*}{0.493} \\
 & 1500 &  & \textbf{0.0012} & 0.1016 & \textbf{0.0012} & -0.0073 & 0.948 & 0.974 & 0.008 & 0.005 &  \\
 & 2000 &  & \textbf{0.0012} & 0.0863 & \textbf{0.0012} & -0.0071 & 0.946 & 0.942 & 0.008 & 0.005 &  \\
\hline
 
\multirow{6}{*}{Weibull} & \multirow{2}{*}{1000} & $\beta$ & \textbf{0.0408} & 5.9941 & \textbf{-0.0346} & 2.9486 & 0.960 & 0.934 & \multirow{2}{*}{0.687} & \multirow{2}{*}{0.040} & \multirow{6}{*}{18.993} \\
 &  & $\eta$ & \textbf{0.2855} & 15.1297 & \textbf{0.2321} & 5.5613 & 0.985 & 0.936 &  &  &  \\
 \noalign{\vskip 1.5mm}
 & \multirow{2}{*}{1500} & $\beta$ & \textbf{0.0402} & 6.6856 & \textbf{-0.0385} & 2.7412 & 0.950 & 0.920 & \multirow{2}{*}{0.687} & \multirow{2}{*}{0.065} &  \\
 &  & $\eta$ & \textbf{0.2722} & 14.8707 & \textbf{0.2572} & 4.1895 & 0.922 & 0.950 &  &  &  \\
 \noalign{\vskip 1.5mm}
 & \multirow{2}{*}{2000} & $\beta$ & \textbf{0.0375} & 5.0358 & \textbf{-0.0363} & 2.3896 & 0.982 & 0.970 & \multirow{2}{*}{0.705} & \multirow{2}{*}{0.088} &  \\
 &  & $\eta$ & \textbf{0.2521} & 14.8434 & \textbf{0.2420} & 5.5784 & 0.984 & 0.924 &  &  &  \\
\hline
 
\multirow{12}{*}{GLFP} & \multirow{4}{*}{1000} & $\beta_1$ & \textbf{0.0513} & 4.5645 & \textbf{-0.0459} & 3.1220 & 0.950 & 0.950 & \multirow{4}{*}{4.356} & \multirow{4}{*}{0.317} & \multirow{12}{*}{249.609} \\
 &  & $\eta_1$ & \textbf{0.2760} & 17.1594 & \textbf{0.2323} & 3.7607 & 0.942 & 0.962 &  &  &  \\
 &  & $\beta_2$ & \textbf{1.4038} & 6.7585 & \textbf{-0.7882} & 1.8256 & 0.974 & 0.978 &  &  &  \\
 &  & $\eta_2$ & \textbf{1.7508} & 24.0151 & \textbf{-1.4767} & 3.3122 & 0.970 & 0.992 &  &  &  \\
  \noalign{\vskip 1.5mm}
 & \multirow{4}{*}{1500} & $\beta_1$ & \textbf{0.0436} & 4.7910 & \textbf{-0.0408} & 2.6242 & 0.968 & 0.958 & \multirow{4}{*}{4.454} & \multirow{4}{*}{0.511} &  \\
 &  & $\eta_1$ & \textbf{0.2355} & 78.5543 & \textbf{0.2015} & 20.2411 & 0.972 & 0.966 &  &  &  \\
 &  & $\beta_2$ & \textbf{1.3085} & 41.3323 & \textbf{-0.8013} & 7.6409 & 0.956 & 0.980 &  &  &  \\
 &  & $\eta_2$ & \textbf{1.5557} & 9.1978 & \textbf{-1.4094} & 1.1273 & 0.974 & 0.976 &  &  &  \\
  \noalign{\vskip 1.5mm}
  & \multirow{4}{*}{2000} & $\beta_1$ & \textbf{0.0392} & 3.3947 & \textbf{-0.0383} & 1.8053 & 0.966 & 0.958 & \multirow{4}{*}{4.540} & \multirow{4}{*}{0.685} &  \\
 &  & $\eta_1$ & \textbf{0.1798} & 85.3416 & \textbf{0.1758} & 25.6368 & 0.972 & 0.970 &  &  &  \\
 &  & $\beta_2$ & \textbf{1.1759} & 57.4013 & \textbf{-0.7803} & 10.8738 & 0.964 & 0.980 &  &  &  \\
 &  & $\eta_2$ & \textbf{1.5355} & 20.8062 & \textbf{-1.4237} & 3.3282 & 0.978 & 0.976 &  &  & 
\\ \bottomrule[\heavyrulewidth]
\end{tabular}
\end{table*}

Next, we illustrate the performance of RDCS under the situation where $\alpha$ is extremely high. 
Table \ref{table:RDCS} shows the RMSE, Bias, CP, and Times under the three aforementioned lifetime distributions.  In all cases, $\alpha = 0.9993$.
It is known that parameter estimation using the data with a high censoring rate leads to less accuracy, a larger size of subsample is chosen compared to the situations with a relatively low censoring rate full data. 
Note that when failure time follows the exponential distribution, the pilot subsample is not needed when using RDCS, because the subsampling probabilities of RDCS: $\tilde{\pi}^{RDCS}_i = \frac{-t_i+t_{il}}{\sum_{i=n_0+1}^n -t_i+t_{il}}$ do not consist of the pilot estimator.

The results presented in Table \ref{table:RDCS} show that the uniform subsampling-based estimation gives bad results but the RDCS maintains a quite good performance in the aspect of RMSE and Bias. 
To be specific, the RMSE of the RDCS-based estimation is an order of magnitude smaller than that of uniform subsampling-based estimation. Besides, uniform subsampling-based estimation has trouble approximating the parameters of Weibull and GLFP distributions, giving rise to large RMSE and Bias. By contrast, the RMSE and bias of RDCS still maintain a low level and are much smaller than the uniform subsampling.
Moreover, RDCS method also has a much shorter computational time than the full-data estimators, which meets our aim of reducing the computational burden.

\begin{figure}[!ht]
\centering
\begin{minipage}[c]{0.32\textwidth}
\centering
\includegraphics[width=\textwidth]{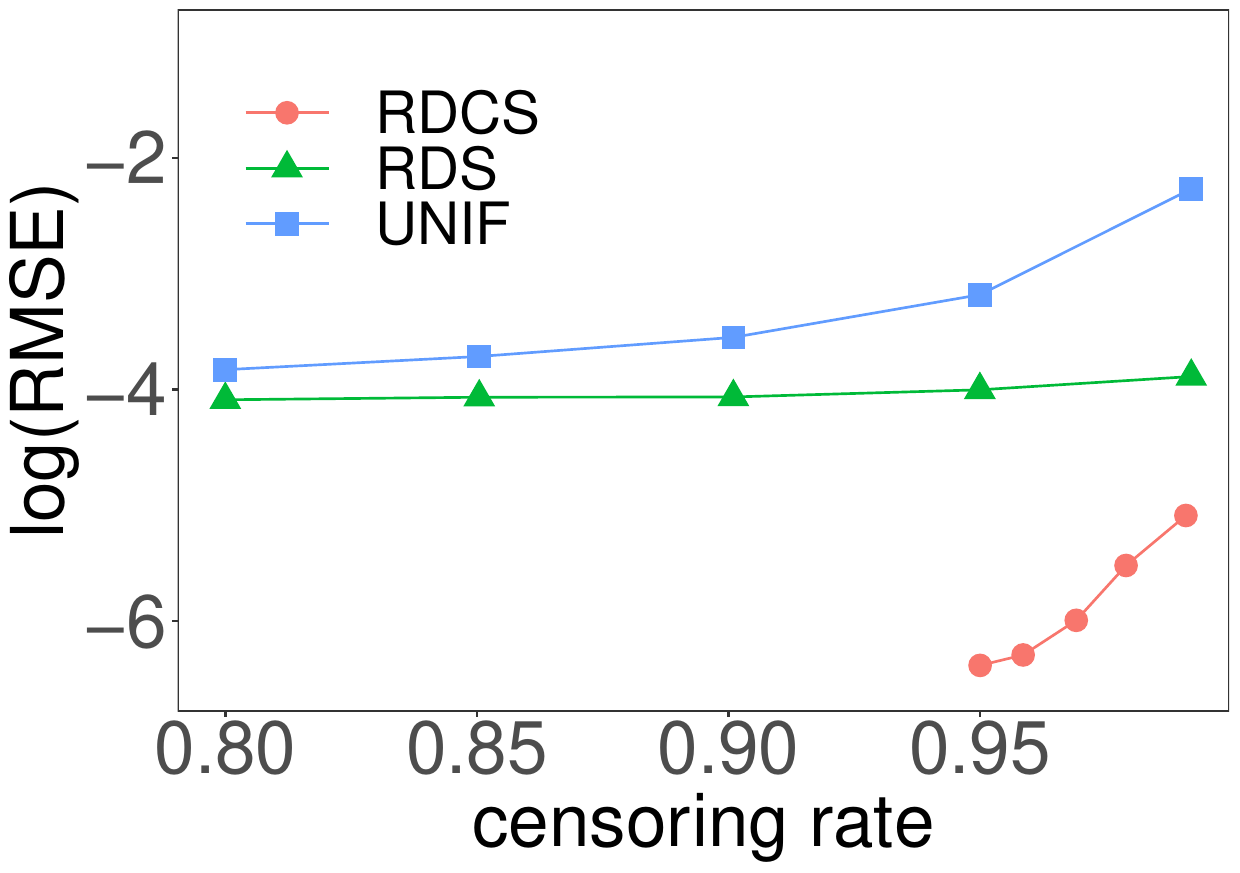}
\subcaption{Exponential distribution with $r=3000$, $r_0=500$.}
        \label{exp}
\end{minipage} 
\begin{minipage}[c]{0.32\textwidth}
\centering
\includegraphics[width=\textwidth]{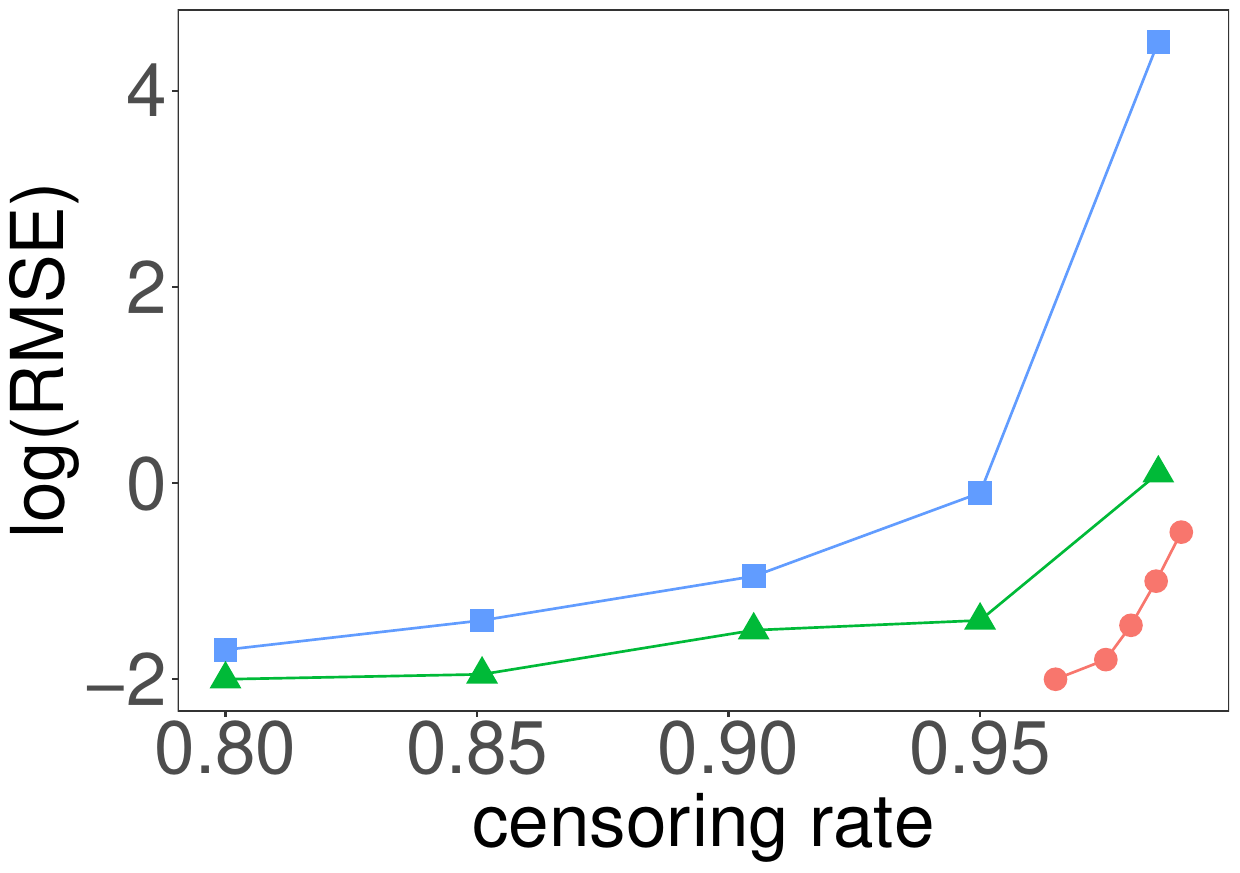}
\subcaption{Weibull distribution with $r=3000$, $r_0=500$.}
        \label{weibull}
\end{minipage} 
\begin{minipage}[c]{0.32\textwidth}
\centering
\includegraphics[width=\textwidth]{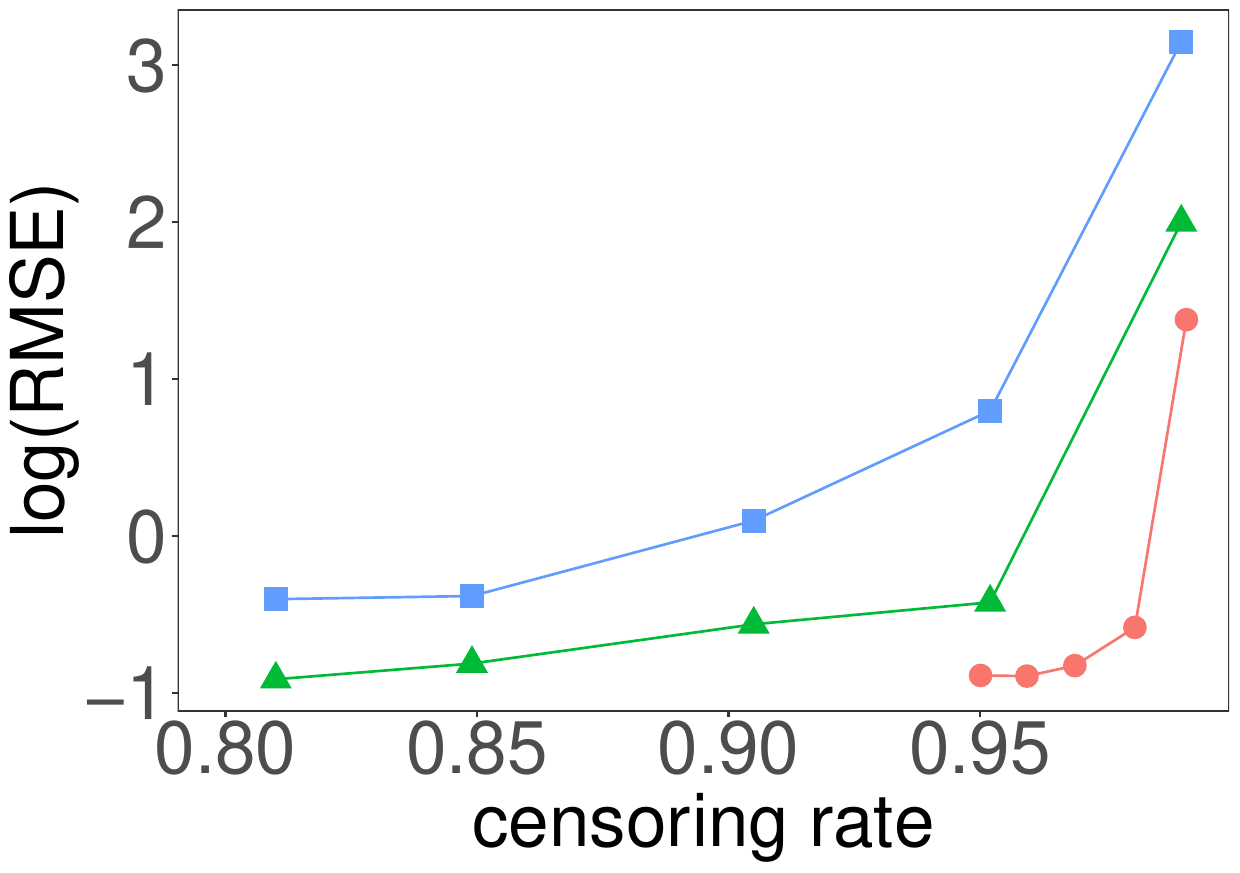}
\subcaption{GLFP distribution with $r=3000$, $r_0=500$.}
		\label{GLFP}
\end{minipage}
\caption{The RMSE of estimators when censoring rate changes}
	\label{three_models_RMSE}
\end{figure}

\begin{figure}[!ht]
\centering
\begin{minipage}[c]{0.24\textwidth}
\centering
\includegraphics[width=\textwidth]{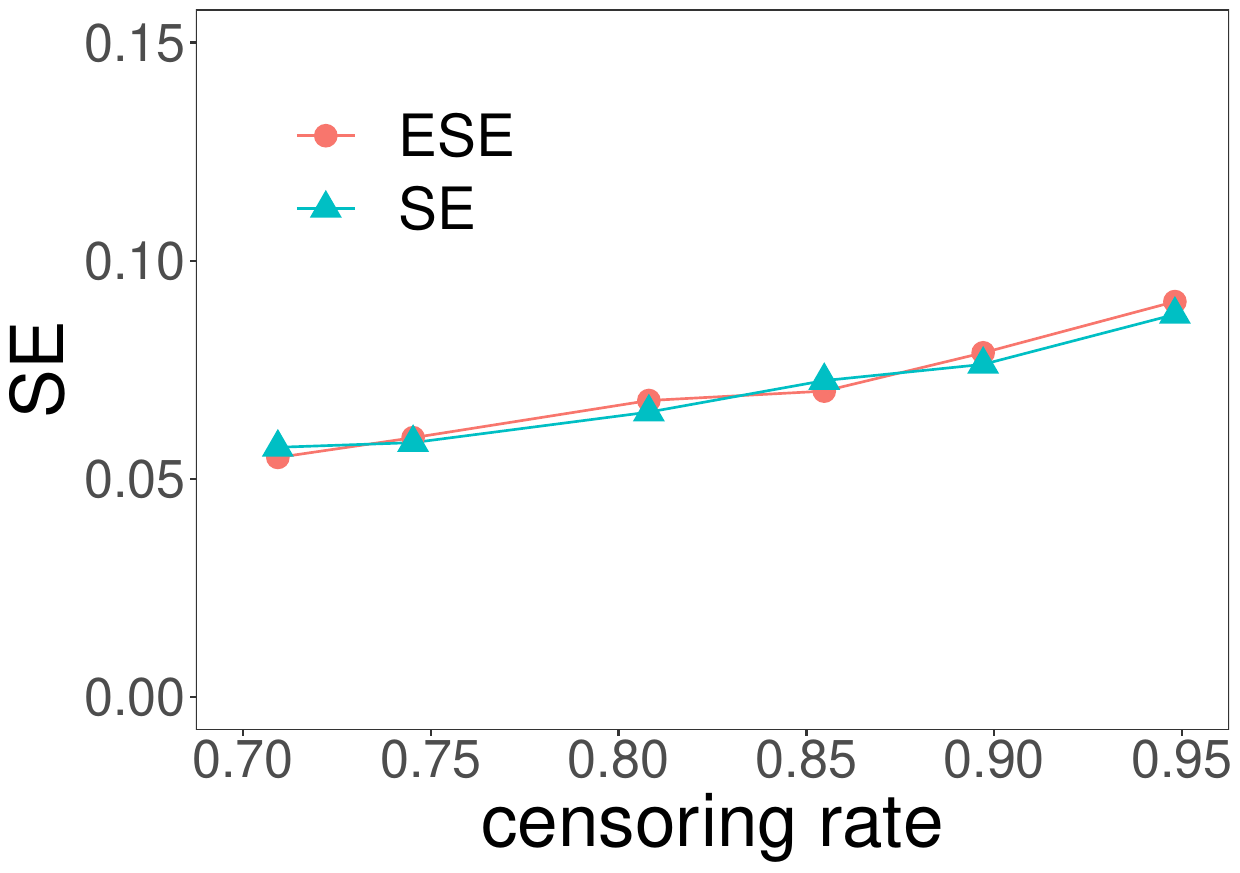}
\subcaption{$\beta$ for RDS}
        \label{figure:RDS-SE-beta}
\end{minipage} 
\begin{minipage}[c]{0.24\textwidth}
\centering
\includegraphics[width=\textwidth]{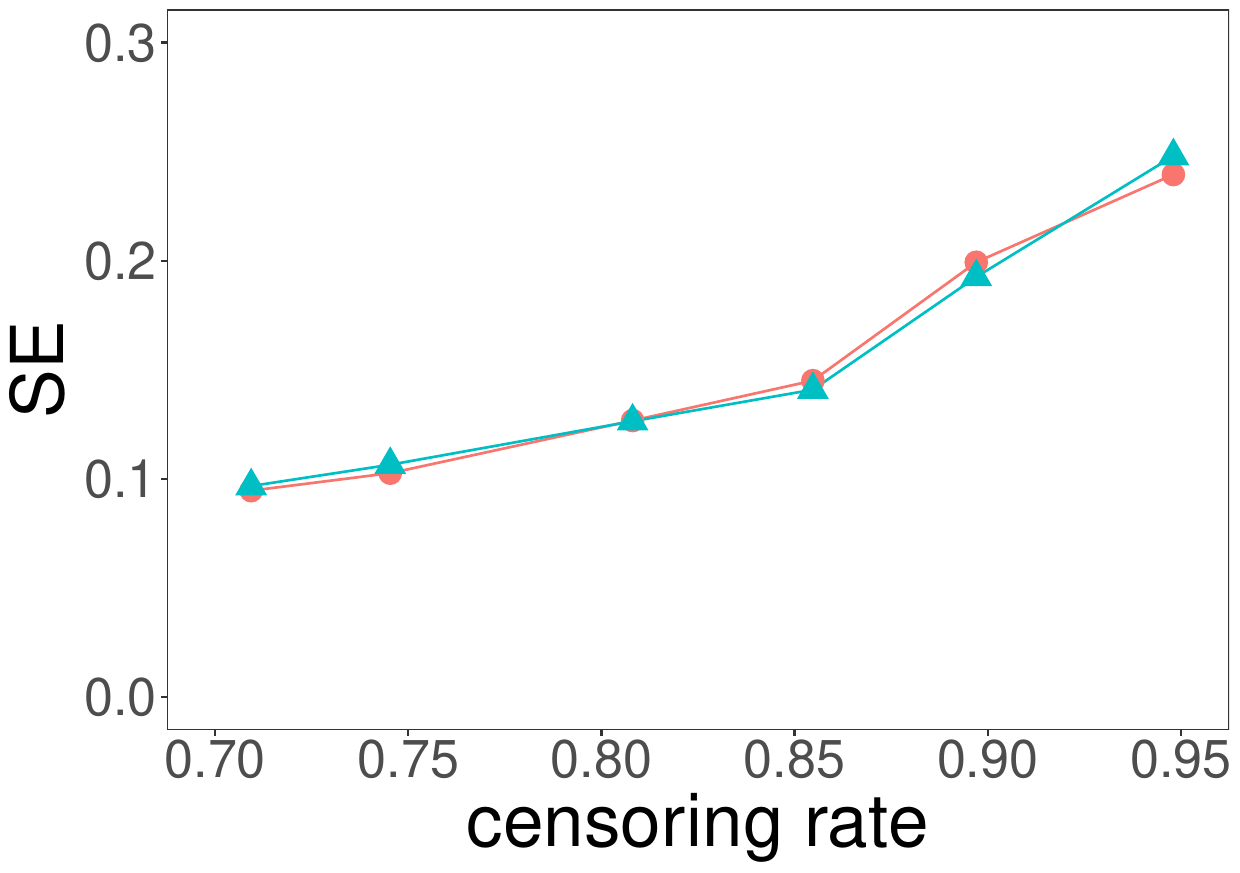}
\subcaption{$\eta$ for RDS}
        \label{figure:RDS-SE-eta}
\end{minipage} 
\begin{minipage}[c]{0.24\textwidth}
\centering
\includegraphics[width=\textwidth]{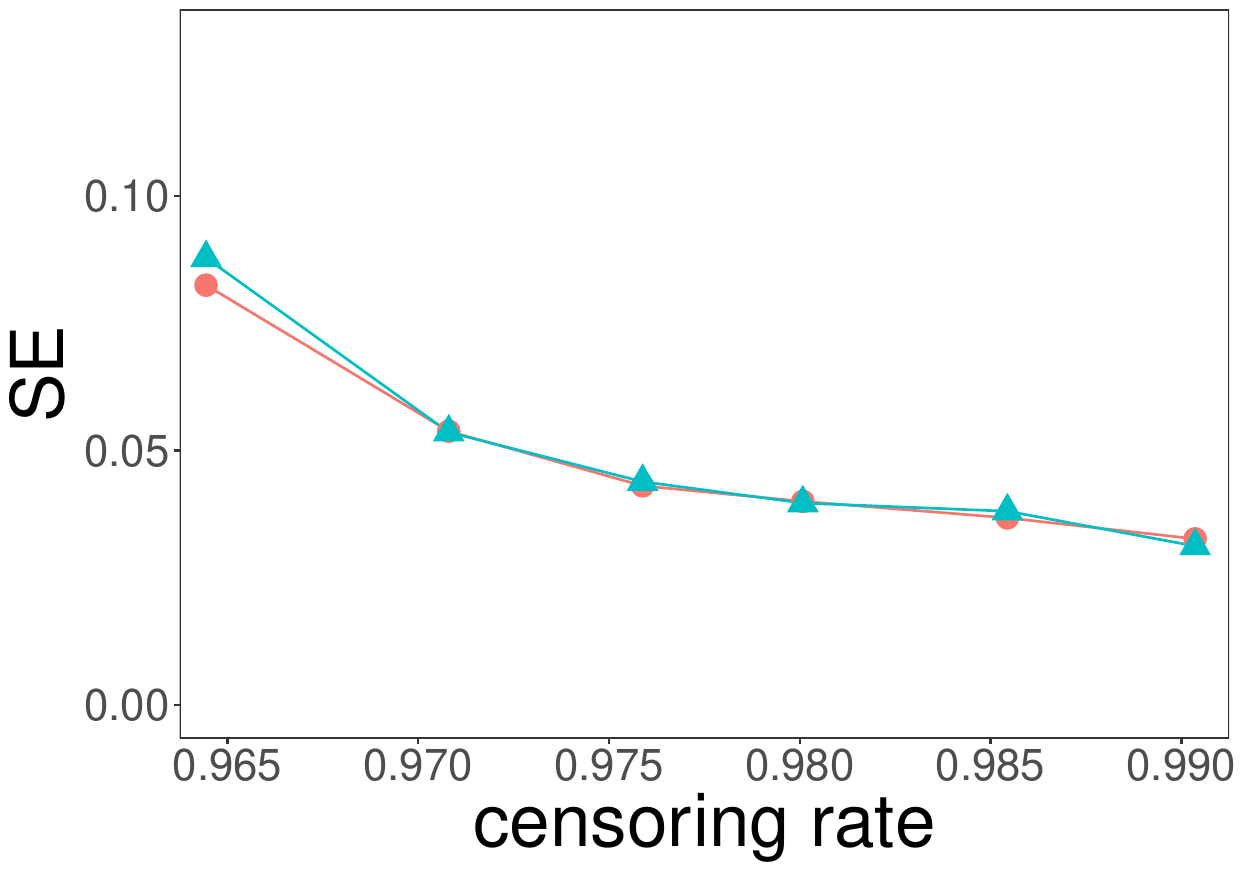}
\subcaption{$\beta$ for RDCS}
		\label{figure:RDCS-SE-beta}
\end{minipage}
\begin{minipage}[c]{0.24\textwidth}
\centering
\includegraphics[width=\textwidth]{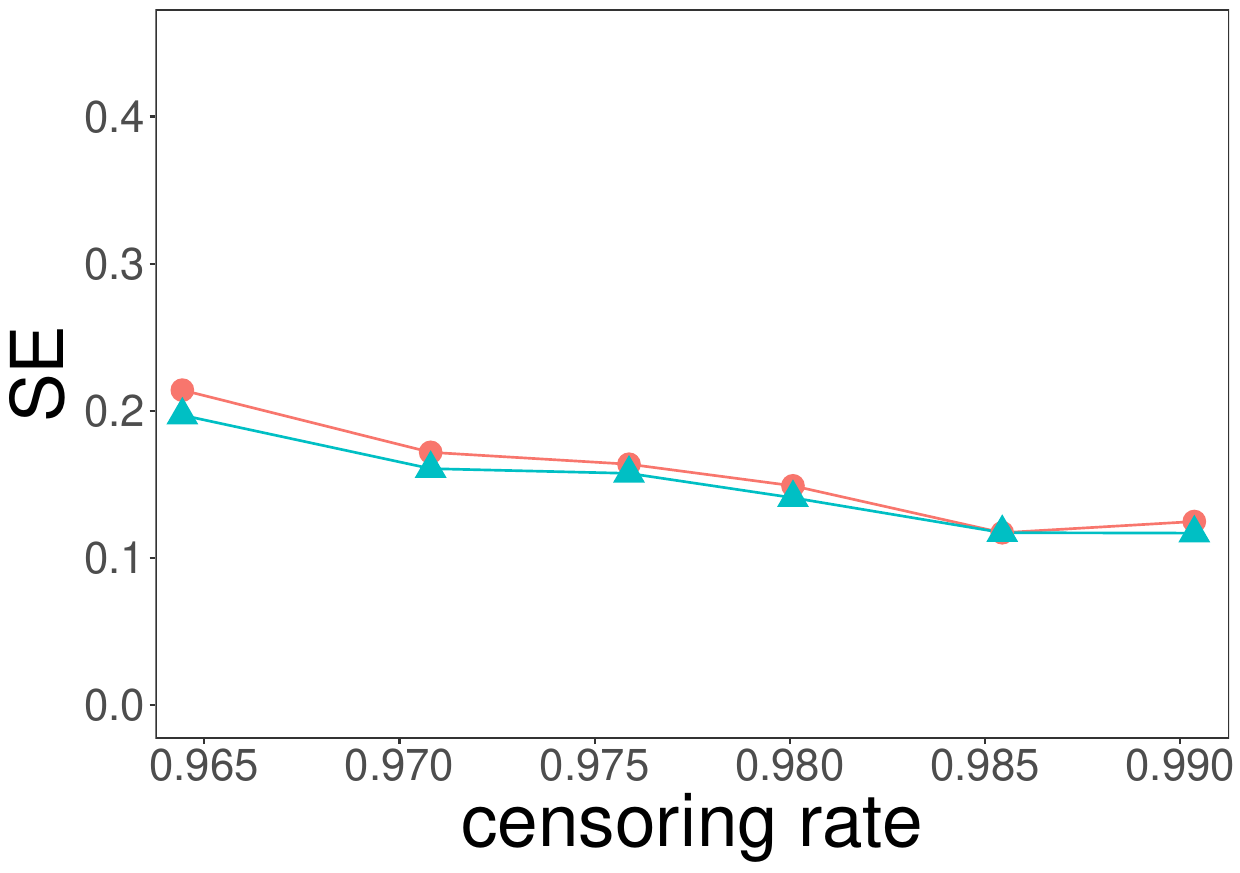}
\subcaption{$\eta$ for RDCS}
        \label{figure:RDCS-SE-eta}
\end{minipage} 
\caption{Comparison between ESE and SE}
	\label{figure:SE}
\end{figure}

\begin{table*}[!ht]
\centering
\caption{The simulation results of RDCS under Weibull distribution when $\alpha$ changes, where $r=1000$, $r_0=400$ and $n_0=1000$ are fixed, $n$ varies to change $\alpha$.}\label{RDCS_n}
\begin{tabular}{ccclllllccc}
\toprule[\heavyrulewidth]
\multirow{2}{*}{n} & \multirow{2}{*}{$\alpha$} & \multirow{2}{*}{r} & \multicolumn{1}{c}{\multirow{2}{*}{}} & \multicolumn{2}{c}{RMSE} & \multicolumn{2}{c}{Bias} & \multicolumn{3}{c}{Time(s)} \\
\cmidrule(lr){5-6} \cmidrule(lr){7-8} \cmidrule(lr){9-11}
 &  &  & \multicolumn{1}{c}{} & \multicolumn{1}{c}{RDCS} & \multicolumn{1}{c}{UNIF} & \multicolumn{1}{c}{RDCS} & \multicolumn{1}{c}{UNIF} & RDCS & UNIF & FULL \\ 
 \midrule[\heavyrulewidth]
\multirow{7}{*}{50000} & \multirow{7}{*}{0.98} & \multirow{2}{*}{1500} & $\beta$ & \textbf{0.0702} & 0.4217 & \textbf{0.0702} & 0.1300 & \multirow{2}{*}{0.079} & \multirow{2}{*}{0.015} & \multirow{7}{*}{1.109} \\
 &  &  & $\eta$ & \textbf{0.3043} & 1.7514 & \textbf{-0.3043} & 0.0419 &  &  &  \\
\noalign{\vskip 1.5mm}
 &  & \multirow{2}{*}{2000} & $\beta$ & \textbf{0.0702} & 0.3257 & \textbf{0.0702} & 0.0726 & \multirow{2}{*}{0.098} & \multirow{2}{*}{0.048} &  \\
 &  &  & $\eta$ & \textbf{0.3044} & 1.4373 & \textbf{-0.3044} & 0.0466 &  &  &  \\
\noalign{\vskip 1.5mm}
 &  & \multirow{2}{*}{2500} & $\beta$ & \textbf{0.0702} & 0.3065 & \textbf{0.0702} & 0.0824 & \multirow{2}{*}{0.104} & \multirow{2}{*}{0.060} &  \\
 &  &  & $\eta$ & \textbf{0.3044} & 1.1446 & \textbf{-0.3044} & -0.0447 &  &  &  \\ \hline
 
\multirow{7}{*}{$10^5$} & \multirow{7}{*}{0.99} & \multirow{2}{*}{1500} & $\beta$ & \textbf{0.0228} & 0.8186 & \textbf{0.0228} & 0.2096 & \multirow{2}{*}{0.113} & \multirow{2}{*}{0.018} & \multirow{7}{*}{1.938} \\
 &  &  & $\eta$ & \textbf{0.1703} & 7.0543 & \textbf{-0.1703} & 1.6671 &  &  &  \\
\noalign{\vskip 1.5mm}
 &  & \multirow{2}{*}{2000} & $\beta$ & \textbf{0.0227} & 0.5995 & \textbf{0.0227} & 0.1706 & \multirow{2}{*}{0.131} & \multirow{2}{*}{0.055} &  \\
 &  &  & $\eta$ & \textbf{0.1702} & 7.2795 & \textbf{-0.1702} & 1.1709 &  &  &  \\
\noalign{\vskip 1.5mm}
 &  & \multirow{2}{*}{2500} & $\beta$ & \textbf{0.0228} & 0.4730 & \textbf{0.0228} & 0.0838 & \multirow{2}{*}{0.137} & \multirow{2}{*}{0.067} &  \\
 &  &  & $\eta$ & \textbf{0.1702} & 2.8873 & \textbf{-0.1702} & 0.5443 &  &  &  \\ \hline
 
\multirow{7}{*}{$10^6$} & \multirow{7}{*}{0.999} & \multirow{2}{*}{1500} & $\beta$ & \textbf{0.0934} & 17.6393 & \textbf{-0.0933} & 5.2166 & \multirow{2}{*}{0.711} & \multirow{2}{*}{0.066} & \multirow{7}{*}{21.802} \\
 &  &  & $\eta$ & \textbf{0.7180} & 188.9546 & \textbf{0.7164} & 61.9186 &  &  &  \\
\noalign{\vskip 1.5mm}
 &  & \multirow{2}{*}{2000} & $\beta$ & \textbf{0.0932} & 6.7504 & \textbf{-0.0931} & 2.4213 & \multirow{2}{*}{0.755} & \multirow{2}{*}{0.091} &  \\
 &  &  & $\eta$ & \textbf{0.7139} & 231.5655 & \textbf{0.7130} & 85.5174 &  &  &  \\
\noalign{\vskip 1.5mm}
 &  & \multirow{2}{*}{2500} & $\beta$ & \textbf{0.0931} & 12.5579 & \textbf{-0.0930} & 2.4348 & \multirow{2}{*}{0.770} & \multirow{2}{*}{0.117} &  \\
 &  &  & $\eta$ & \textbf{0.7131} & 226.7329 & \textbf{0.7125} & 94.8492 &  &  & \\
 \bottomrule[\heavyrulewidth]
\end{tabular}
\end{table*}

To demonstrate how the proposed methods perform when the censoring rate changes, Figure \ref{exp}, \ref{weibull} and \ref{GLFP} plot the relationship between $\log$(RMSE) of RDCS, RDS, and uniform subsampling and censoring rate $\alpha$ of full data, under different sample size $r$ and sample size of pilot data $r_0$. We set the full data size $n = 50000$. Note that the RDCS method needs to collect all the uncensored data into the subsample ($r > n\alpha$). As $r$ is fixed, if $n$ is too large (e.g., $n = 10^6$), the appropriate censoring rate will be larger than 0.99. In that case, because the censoring rates of the other two methods RDS and uniform subsampling vary from 0.80 to 0.99, it is quite difficult to compare the plots of three alternative methods in one figure. Results of RDS and uniform subsampling are shown between $\alpha=0.80$ to $\alpha=0.99$ and the results of RDCS are shown between $\alpha=0.95$ to $\alpha=0.99$. We can find that RDS  performs consistently better than uniform subsampling. When the censoring rate is high enough to use the RDCS, RDCS has the smallest RMSE among the three. So RDCS is recommended when $\alpha$ is large. 

Additionally, we verify the variance estimation proposed in Section \ref{Variance Estimation} by comparing the mean estimated standard errors (ESE) and the empirical standard errors (SE). Let $r=2000$, $r_0=500$ and set the Weibull distribution model with true parameters $(\beta,\eta)=(2,4)$. All the experiments are repeated 500 times. Using RDS and RDCS, we calculate the SE of $\beta$ and $\eta$ respectively. Besides, in every repetition, the ESE are calculated using the diagonal value of $\frac{1}{r}\hat{\boldsymbol{V}}_{RDS}$ and $\frac{1}{r}\hat{\boldsymbol{V}}_{RDCS}$, and we compute the mean value as the final ESE. Figure \ref{figure:RDS-SE-beta}, \ref{figure:RDS-SE-eta}, \ref{figure:RDCS-SE-beta} and \ref{figure:RDCS-SE-eta} show the results under different $\alpha$. Note that with the increase of $\alpha$, the SE of RDS becomes large, but the SE of RDCS counter-intuitively goes down. Meanwhile, when $\alpha$ is near 0.96, there are some subtle differences between SE and ESE when we use the RDCS method. That is because the proposed estimation of covariance matrix $\hat{\boldsymbol{V}}_{RDCS}$ is based on the assumption that $n_0 \ll r$. However, this assumption does not fit well when $\alpha$ is relatively low and will be gradually satisfied when $\alpha$ becomes higher. Nevertheless, the maximum difference between ESE and SE is still minimal. Hence, the proposed ESE can still approximate SE quite well, which shows the validity of our theories.

To investigate the behavior of the alternative methods under a situation where the censoring rate approaches $1$, we fix $r_0=400$ and let $n=50000$, $10^5$ and $10^6$. 
The corresponding censoring rates are 0.98, 0.99, and 0.999 respectively.
We select $r=1500, 2000, 2500$ and $n_0=1000$, i.e., we always let the uncensored data of size $1000$ no matter how much data we have. 
The main purpose of conducting this simulation is to illustrate the importance of including all complete data for a subsampling algorithm when complete data is rare.
Under the Weibull distribution model, we present the estimation results of RDCS and uniform subsampling in Table \ref{RDCS_n}. With the increase of $n$, the RMSE and empirical bias of uniform subsampling-based estimation increase significantly, which is unacceptable. By contrast, the results of RDCS are dramatically better and \textit{stable}. This is because when $n_0$ remains constant, and the censoring rate approaches $1$ with $n\to \infty$, conventional methods will ignore most complete data as their appearance probability approaches $0$, leading to significant information loss. On the other hand, the RDCS is not affected by this aspect, as all the uncensored data are retained. Except for this, the mean calculation time of RDCS is much faster than that of full data.

\section{Conclusion}
In this article, two subsampling methods, RDS and RDCS, are proposed for efficiently estimating the lifetime model with massive reliability data. Asymptotic properties of these two subsampling-based estimators are derived. This is fundamental to calculating optimal subsampling probabilities. 
The corresponding algorithms are established through L-optimality to obtain the optimal subsampling probabilities, which also guarantee the asymptotic properties.
A real data example and simulation studies are employed to demonstrate the effectiveness of the proposed methods. In particular, the proposed methods are applied to three different lifetime distributions, and all the results outperform the uniform subsampling-based estimators in the aspects of RMSE and bias. As full data volume explodes, our methods become even faster than the uniform subsampling. In summary, the proposed methods dramatically reduce the computation burden for the massive dataset without loss of accuracy.

In comparing our methods to existing ones, notable distinctions arise: (\romannumeral1) \citet{yang2022} and \citet{Zuo2021} optimize sampling probabilities for specific survival models. 
They do not address highly censored cases, a gap our methods adeptly fill. 
(\romannumeral2) \citet{Wang2020rare} and \citet{Wang_2021} considered the subsampling methods for the logistic regression model with rare events. \citet{Wang2020rare} compared the efficiency between undersampling or oversampling and recommended using undersampling methods, while \citet{Wang_2021} obtained the optimal sampling probability and improved it by correcting log odds of sampled data. Besides, the data they dealt with are uncensored. However, our subsampling methods are designed to analyze heavily censored reliability data with lifetime distribution models.
(\romannumeral3) \citet{Zhang2022} and \citet{Keret2023} proposed the subsampling-based maximum partial likelihood estimation for the Cox proportional hazards model. The model assumes that the ratio of hazards of two subjects is constant and relies on the covariates. \citet{Zhang2022} studied the regular setting where uncensored data are not rare compared with the observed censoring times, and \citet{Keret2023} considered a more extreme case where uncensored data are rare.
In contrast, our methods focus on the characteristics of failure time and can be applied to various lifetime distribution models. Meanwhile, we derive the optimal subsampling methods for both the regular setting and the rare events case.

Nonetheless, our methods beckon further refinement. 
For example, in real-world applications, outliers exist in the raw data. The bias between estimators and true value will appear as the existence of outliers might affect the asymptotic properties of maximum likelihood estimators. Therefore, it is necessary to study the optimal subsampling methods with the presence of outliers.

\normalem 
\bibliographystyle{apalike}
\bibliography{arxiv.bib} 
\section{Appendix}
The supplementary materials include fout parts. Appendix \ref{appendix:Exmaples} verifies that the assumptions mentioned in the article can be satisfied under some commonly used distributions. 
Appendix \ref{appendix:Theorems} derive the asymptotic normality property between the estimators and the true value.
Appendix \ref{extra_sim} shows the impact of censoring on optimal subsampling probabilities through extra simulations.
Appendix \ref{appendix:proof} gives the proofs of all the theorems and propositions in this article.
\subsection{Justification for Assumptions In Section 3} \label{appendix:Exmaples}

\justifying{Exponential Distribution: $f(t;\theta) = \theta \exp(-\theta t)$, $F(t;\theta) = 1-\exp(-\theta x)$. $l_i(\theta) = (1-C_i)log\lambda - \lambda t_i + \lambda t_{il}$, $\frac{\partial l_i}{\partial \theta} = \frac{1-C_i}{\theta} - t_i + t_{il},\frac{\partial^2 l_i}{\partial \theta^2} = \frac{C_i -1}{\theta^2},\frac{\partial^3 l_i}{\partial \theta^3} = \frac{2(1-C_i)}{\theta^3}$.  If $\pi_i = 1/n$ for every $i$, 

$(A1)$: $M_t = \frac{1}{n} \sum_{i=1}^n -\frac{1-C_i}{\hat{\theta}^2}$ is negative, $\frac{1}{n^2}\sum_{i=1}^n \frac{1}{\pi_i} \ddot{l}_i(\hat{{\theta}})^2 = \frac{1}{n} \sum_{i=1}^n (\frac{1-C_i}{\hat{\theta}^2})^2 \leq \frac{1}{\hat{\theta}^4} = O_{p|\mathcal{D}_n}(1)$. 
    
$(A2)$: $|\frac{1}{n^2}\sum_{i=1}^n\frac{1}{\pi_i}\dot{l}_i^2(\hat{\theta})| = |\frac{1}{n} \sum_{i=1}^n [\frac{1-C_i}{\hat{\theta}^2} +(t_i-t_{il})^2 - \frac{2(1-C_i)(t_i-t_{il})}{\hat{\theta}}]| \leq \frac{1}{\hat{\theta}^2} + L^2 + \frac{2L}{\hat{\theta}} = O_{p|\mathcal{D}_n}(1)$, where $L$ is the time from the beginning to the end of the observation. 
    
$(A3)$: $\frac{1}{n}\Vert\frac{\partial^3 l_i(\hat{\theta})}{\partial \theta^3}\Vert = \frac{1}{n} \sum_{i=1}^n \frac{2(1-C_i)}{\hat{\theta}^3} \leq \frac{1}{\hat{\theta}^3} = O_{p|\mathcal{D}_n}(1)$. 
    
$(A4)$: $\frac{1}{n^{2+\delta}} \sum_{i=1}^n \frac{1}{\pi^{1+\delta}} \Vert \dot{l}_i (\hat{\theta})\Vert^{2+\delta} = \frac{1}{n} \sum_{i=1}^n |\frac{1-C_i}{\hat{\theta}} - (t_i - t_{il})|^{2+\delta} \leq (\frac{1}{\hat{\theta}} + L)^{2+\delta} = O_{p}(1)$. 
    
Hence, if the lifetime distribution is exponential distribution, Assumptions \ref{ass1}-\ref{ass4} can be easily satisfied.    }

\justifying{
    Weibull Distribution: $f(t;\boldsymbol{\theta}) = \frac{\beta}{\eta} (\frac{t}{\eta})^{\beta-1} \exp(-(\frac{t}{\eta})^\beta)$, $F(t;\boldsymbol{\theta}) = 1-\exp(-(\frac{t}{\eta})^\beta)$, $\boldsymbol{\theta} = (\beta, \eta)$. $l_i(\boldsymbol{\theta}) = (1- C_i)[log\beta - \beta log\eta +\beta log t_i - log t_i]+(\frac{t_{il}}{\eta})^\beta - (\frac{t_{i}}{\eta})^\beta$, 
        \begin{align*}
            &\frac{\partial l_i}{\partial \beta} = (1-C_i)(\frac{1}{\beta} - log\eta + log t_i) + (\frac{t_{il}}{\eta})^\beta log \frac{t_{il}}{\eta} - (\frac{t_{i}}{\eta})^\beta log \frac{t_{i}}{\eta}, \\
            &\frac{\partial l_i}{\partial \eta} = -(1-C_i)\frac{\beta}{\eta} + \frac{\beta t_i ^ \beta}{\eta ^{\beta + 1}} - \frac{\beta t_{il} ^ \beta}{\eta ^{\beta + 1}},\\
            &\frac{\partial^2 l}{\partial \beta^2 } = \frac{-(1-C_i)}{\beta^2}+\beta(\frac{t_{il}}{\eta})^{\beta-1}(log t_{il} - log\eta)\\
            &\quad-\beta(\frac{t_i}{\eta})^{\beta-1}(log t_i - log\eta),\\
            &\frac{\partial^2 l}{\partial \beta \partial \eta } = \frac{-(1-C_i)}{\eta} + \beta t_i ^\beta + \frac{\beta(\beta +1)t_{il}^\beta}{\eta^{\beta+2}} - \frac{\beta (\beta+1) t_i^\beta}{\eta^{\beta+2}},\\
            &\frac{\partial^2 l}{\partial \eta^2} = \frac{(1-C_i)\beta}{\eta^2} - \frac{\beta(\beta+1)t_i^\beta}{\eta^(\beta+2)} + \frac{\beta(\beta+1)t_{il}^\beta}{\eta^(\beta+2)}.\\
            &\frac{\partial^3 l}{\partial \beta^3} = \frac{(1-C_i)}{\eta^3} + log\frac{t_{il}}{\eta}[(\frac{t_{il}}{\eta})^{\beta-1} + \beta (\frac{t_{il}}{\eta})^{\beta-1} log\frac{t_{il}}{\eta}] - \\
            &log\frac{t_i}{\eta}[(\frac{t_i}{\eta})^{\beta-1} + \beta(\frac{t_i}{\eta})^{\beta-1}log\frac{t_i}{\eta}],\\
            &\frac{\partial^3 l}{\partial^2 \beta \partial \eta} = \frac{\beta(-\beta+1)t_{il}^{\beta-1}}{\eta^\beta} log\frac{t_{il}}{\eta} - \frac{\beta t_{il}^{\beta-1}}{\eta^\beta} \\
            &\quad- \frac{\beta(-\beta+1)t_{i}^{\beta-1}}{\eta^\beta} log\frac{t_{i}}{\eta} + \frac{\beta t_{i}^{\beta-1}}{\eta^\beta},\\
            &\frac{\partial^3 l}{\partial\beta\partial\eta^2} = \frac{1-C_i}{\eta^2} - \frac{\beta(\beta+1)(\beta+2)t_{il}^\beta}{\eta^{\beta+3}} + \frac{\beta(\beta+1)(\beta+2)t_i^\beta}{\eta^{\beta+3}},\\
            &\frac{\partial^3 l}{\partial\eta^3} = \frac{-2(1-C_i)\beta}{\eta^3}+\frac{\beta(\beta+1)(\beta+2)t_i^\beta}{\eta^{\beta+3}} \\
            &\quad- \frac{\beta(\beta+1)(\beta+2)t_{il}^\beta}{\eta^{\beta+3}}.
        \end{align*} \notag
    If $\pi_i = 1/n$ for every $i$, $t_0 \leq min(t_{il}) < max(t_i) \leq T_0$, $\boldsymbol{M}_t = \frac{1}{n} \sum_{i=1}^n \frac{\partial^2 l_i(\hat{\boldsymbol{\theta}})}{\partial \boldsymbol{\theta}^2}  $ is negative definite. Similarly, as $\hat{\boldsymbol{\theta}}$ is fixed and $t_{il}, t_i, i=1,2,...,n$ are finite, it is easy to see that Assumptions \ref{ass1}-\ref{ass4} are all satisfied.
    }

\subsection{Properties Between Estimators And True Value}\label{appendix:Theorems}
 
\begin{assumption} \label{ass9}
   The subsampling probabilities satisfy $\max_{i=1,2,...,n} (n\pi_i)^{-1} = O_p(1)$.
\end{assumption}
\begin{thm} \label{thm:general-true}
    When Theorem \ref{thm:genenral-mle} is established, assume $\boldsymbol{\Lambda}_g(\boldsymbol{\theta}_0)=\frac{1}{n^2}\sum_{i=1}^n\frac{1}{\pi_i}\dot{l}^2_i(\boldsymbol{\theta}_0)$ converges to $\boldsymbol{\Lambda}_{1g}(\boldsymbol{\theta}_0)$ as $n,r \rightarrow \infty$, under the Assumptions \ref{ass5}-\ref{ass7} and \ref{ass9},
    \begin{align}  
        \sqrt{r}\boldsymbol{V}_{0g}^{-\frac{1}{2}}(\tilde{\boldsymbol{\theta}}_g - \boldsymbol{\theta}_0)\xrightarrow{d} N(0,\boldsymbol{I}),
    \end{align} 
   where $\boldsymbol{V}_{0g}=\boldsymbol{M}_{0g}^{-1}\boldsymbol{\Lambda}_{0g}\boldsymbol{M}_{0g}^{-1}$, $\boldsymbol{M}_{0g} = \mathbb{E}\{ \ddot{l}(\boldsymbol{x}, \boldsymbol{\theta}_0 )\}$, $\boldsymbol{\Lambda}_{0g} = \boldsymbol{\Lambda}_{1g}(\boldsymbol{\theta}_0) + \gamma\boldsymbol{\Lambda}(\boldsymbol{\theta}_0)$ and $\gamma = \lim_{r,n\rightarrow\infty} r/n$. 
\end{thm} 
    
Theorem \ref{thm:general-true} shows that the unconditional convergence rate of $\tilde{\boldsymbol\theta}_g$ to the true parameter $\boldsymbol{\theta}_0$ is at the rate $O_p(r^{-1/2})$,  $\tilde{\boldsymbol\theta}_g$ can be close enough to the true parameter when $r\rightarrow \infty$. The convergence rate of the estimate is minimally influenced by the size of full data $n$ and dominated by the size of the subsampled subset. It is worth noting that the convergence rate of $\tilde{\boldsymbol{\theta}}_g - \boldsymbol{\theta}_0$ is the same as $\tilde{\boldsymbol{\theta}}_g - \hat{\boldsymbol{\theta}}$. 
In particular, $\tilde{\boldsymbol{\theta}}_g - \boldsymbol{\theta}_0 = \tilde{\boldsymbol{\theta}}_g - \hat{\boldsymbol{\theta}} + \hat{\boldsymbol{\theta}} - \boldsymbol{\theta}_0$, and it is well-known that $\hat{\boldsymbol{\theta}} - \boldsymbol{\theta}_0$ has the convergence rate of $O(n^{-1/2})$. Intuitively, when $r = o(n)$, the convergence rate of $\hat{\boldsymbol{\theta}} - \boldsymbol{\theta}_0$ can be neglected compared to the convergence rate of $\tilde{\boldsymbol{\theta}}_g - \hat{\boldsymbol{\theta}}$. Hence, it is rational and convenient to choose optimal subsampling probabilities merely through the properties between $\tilde{\boldsymbol\theta}_g$ and $ \hat{\boldsymbol\theta}$.

\begin{thm} \label{thm:censor-true}
\justifying{
    When Theorem 3 is established and $n,r \rightarrow \infty$, assume
    $\boldsymbol{\Lambda}_c(\boldsymbol{\theta}_0) = \frac{1}{n^2} \sum_{i=n_0+1}^n\\
    \frac{1}{\tilde{\pi}_i} \dot{l}_i(\boldsymbol{\theta}_0)^2 - \left[ \frac{1}{n}\sum_{i=n_0+1}^n \dot{l}_i(\boldsymbol{\theta}_0) \right]^2 \rightarrow \boldsymbol{\Lambda}_{1c}(\boldsymbol{\theta}_0).$
    Under the Assumptions \ref{ass5}-\ref{ass8}, we have}
    \begin{align}  
        \sqrt{r}\boldsymbol{V}_{0c}^{-\frac{1}{2}}(\tilde{\boldsymbol{\theta}}_c - \boldsymbol{\theta}_0)\xrightarrow{d} N(0,\boldsymbol{I}), 
    \end{align} 
   where $\boldsymbol{V}_{0c}=\boldsymbol{M}_{0g}^{-1}\boldsymbol{\Lambda}_{0c}\boldsymbol{M}_{0g}^{-1}$, $\boldsymbol{\Lambda}_{0c} = \boldsymbol{\Lambda}_{1c}(\boldsymbol{\theta}_0) + \gamma\boldsymbol{\Lambda}(\boldsymbol{\theta}_0)$ and $\gamma = \lim_{n,r \rightarrow \infty} r/n$. 
\end{thm} 

As stated, in practice, we let $r\ll n$. In that case, $\gamma=0$ and $\boldsymbol{\Lambda}_{0c} = \boldsymbol{\Lambda}_{1c}(\boldsymbol{\theta}_0)$. 

\subsection{Extra Simulations}\label{extra_sim}
\begin{figure}[htbp]
\centering
\begin{minipage}[c]{0.24\textwidth}
\centering
\includegraphics[width=\textwidth]{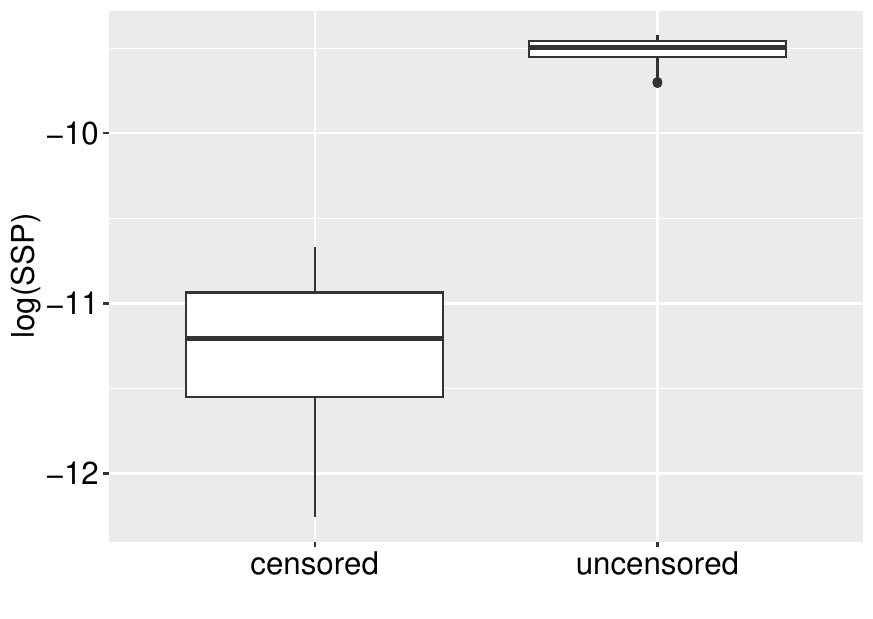}
\subcaption{Exponential distribution}
        \label{figure:SSP_exp}
\end{minipage} 
\begin{minipage}[c]{0.24\textwidth}
\centering
\includegraphics[width=\textwidth]{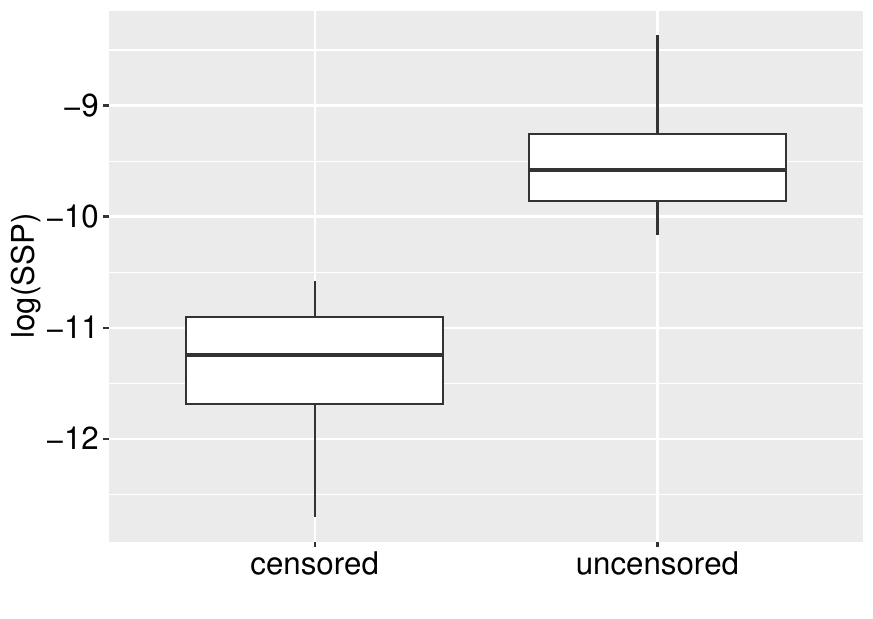}
\subcaption{Weibull distribution}
        \label{figure:SSP_wei}
\end{minipage} 

\caption{The distribution of SSP under two distributions, where $r=1000$, $r_0=200$ and $\alpha=0.90$.}
	\label{figure:SSP}
\end{figure}
To show the impact of censoring on optimal subsampling probabilities, we calculate the subsampling probabilities of all observations and compare the logarithm of subsampling probabilities between censored data and uncensored data. Figure \ref{figure:SSP_exp} and \ref{figure:SSP_wei} report the results under exponential distribution and weibull distribution respectively. We set $n=50000$, $r=1000$, $r_0=200$ and $\alpha=0.90$ for both two situations. As we can see, in general, uncensored data have larger subsampling probabilities than censored data so the uncensored data are more likely to be collected into the subsample.

\subsection{Proofs}\label{appendix:proof}
\subsubsection{Proof of Theorem \ref{thm:genenral-mle}}
From Proposition 2 of \cite{Wang2022comparative}, a sequence converging to 0 in conditional probability is equivalent to its converging to 0 in unconditional probability. Hence, $o_{p|\mathcal{D}_n}(1)$ is the same as $o_{p}(1)$ in the following proofs. The prove of Theorem \ref{thm:genenral-mle} need the lemma below:
    \begin{lemma}\label{lemma1}
        Under Assumptions \ref{ass1}, \ref{ass2} and \ref{ass3}, the following results hold:
        \begin{align}
            \boldsymbol{\tilde{M_g}} - \boldsymbol{M_g} = O_{p|\mathcal{D}_n}(r^{-\frac{1}{2}})\label{lemma1.1} \\
            \frac{1}{n} \frac{\partial l_g^*(\boldsymbol{\hat{\theta}})}{\partial\boldsymbol{\theta}} = O_{p|\mathcal{D}_n}(r^{-\frac{1}{2}}) \label{lemma1.2} 
        \end{align}
        where 
        \begin{align*}
            \boldsymbol{\tilde{M_g}}=\frac{1}{n} \frac{\partial^2 l_g^*(\boldsymbol{\hat{\theta}})}{\partial\boldsymbol{\theta}\partial\boldsymbol{\theta^\top}} = 
            \frac{1}{nr} \sum_{i=1}^r\frac{1}{\pi_i^*} \frac{\partial^2 l_i^*(\boldsymbol{\hat{\theta}})}{\partial\theta\partial\theta^\top} 
        \end{align*}
        \end{lemma}

\begin{proof}
    \begin{equation}\label{equl4}
        \mathbb{E}(\tilde{\boldsymbol{M_g}}|\mathcal{D}_n) = \frac{1}{nr}\sum_{i=1}^r\mathbb{E}\left(\frac{1}{\pi_i^*} \ddot{l}^*_i(\hat{\boldsymbol{\theta}})|\mathcal{D}_n\right)
        =\frac{1}{n}\sum_{i=1}^n\ddot{l}_i(\hat{\boldsymbol{\theta}})
        =\boldsymbol{M_g} 
    \end{equation}
    
Let $\tilde{\boldsymbol{M}}_g^{j_1j_2}$ denote the $j_1$-row, $j_2$-column entry of $\tilde{\boldsymbol{M}_g}$, $1\leq j_1, j_2\leq d$. Then
    \begin{equation}\label{equl5}
        \begin{aligned}
             Var\left(\frac{1}{\pi_i^*}\ddot{l}^*_i(\hat{\boldsymbol{\theta}})^{j_1j_2}|\mathcal{D}_n\right) &=\sum_{i=1}^n \pi_i\left(\frac{1}{\pi_i}\ddot{l}_i(\hat{\boldsymbol{\theta}})^{j_1j_2}-n\boldsymbol{M}_g^{j_1j_2}\right)^2 \\
            Var\left( \frac{1}{n}\tilde{\boldsymbol{M}}_g^{j_1j_2}|\mathcal{D}_n\right) &= \frac{1}{n^4r^2}\sum_{i=1}^n Var\left(\frac{1}{\pi_i^*}\ddot{l}^*_i(\hat{\boldsymbol{\theta}})^{j_1j_2}|\mathcal{D}_n\right)\\
            &= \frac{1}{r^2} \sum_{i=1}^n\pi_i\left(\frac{\ddot{l}_i(\hat{\boldsymbol{\theta}})^{j_1j_2}}{n\pi_i}- \boldsymbol{M}_g^{j_1j_2}  \right)^2\\
            &=\frac{1}{rn^2}\sum_{i=1}^n\frac{(\ddot{l}_i(\hat{\boldsymbol{\theta}})^{j_1j_2})^2}{\pi_i} - \frac{1}{r}(\boldsymbol{M}_g^{j_1j_2})^2\\
            &= O_{p|\mathcal{D}_n}\left(\frac{1}{r}\right) 
        \end{aligned}
    \end{equation}
    The last equality of equation \ref{equl5} holds by Assumption \ref{ass1}. Combining equation \ref{equl4}, \ref{equl5} and Chebyshev's inequality, we can derive equation \ref{lemma1.1}.

 Similarly, 
 \begin{equation}\label{lemmal6}
    \mathbb{E}\left(\frac{1}{n} \frac{\partial l_g^*(\hat{\boldsymbol{\theta}})}{\partial \boldsymbol{\theta}} | \mathcal{D}_n        \right)=\frac{1}{n} \sum_{i=1}^n \dot{l}_i(\hat{\mathcal{\theta}}) = 0 
\end{equation}
\begin{equation}\label{lemmal7}
    Var\left( \frac{1}{n} \frac{\partial l_g^*(\boldsymbol{\hat{\theta}})}{\partial\boldsymbol{\theta}} \right) = \frac{1}{n^2r} \sum_{i=1}^n \frac{\dot{l_i^2}(\boldsymbol{\hat{\theta}})}{\pi_i} = O_{p|\mathcal{D}_n}(r^{-\frac{1}{2}}) 
\end{equation}
The last equality of equation \ref{lemmal7} holds under Assumption \ref{ass2}. Using equation \ref{lemmal6}, \ref{lemmal7} and Chebyshev's inequality, the equation \ref{lemma1.2} is derived.

\end{proof}

The prove of Theorem \ref{thm:genenral-mle}:
\begin{proof} 
\begin{equation} \label{thms1}
    \begin{aligned}
        &\mathbb{E} \left\{ \frac{1}{n}  l_g^*(\boldsymbol{\theta}) - l_f(\boldsymbol{\theta}) | \mathcal{D}_n \right\}^2\\
        =& \frac{1}{r} \left[ \frac{1}{n^2} \sum_{i=1}^n \frac{l_i^2(\boldsymbol{\theta})}{\pi_i} -\left( \frac{1}{n} \sum_{i=1}^n l_i(\boldsymbol{\theta})^2 \right) \right] \\
        =& O_{p|\mathcal{D}_n}\left(\frac{1}{r}\right)     
    \end{aligned}
\end{equation}     

Equation \ref{thms1} holds when Assumption \ref{ass2} is true. Then,
\begin{equation}
    \frac{1}{n}l_g^*(\boldsymbol{\theta}) - l_f(\boldsymbol{\theta}) \xrightarrow{P|\mathcal{D}_n} 0 \notag
\end{equation}

The Theorem 5.9 of Van der Vaart(1998) shows: Let $\Psi_n$ be random vector-valued functions and let $\Psi$ be a fixed vector-valued function of $\boldsymbol{\theta}$ such that for every $\epsilon>0$, $ \sup_{\boldsymbol{\theta} \in \boldsymbol{\Theta}} \Vert \Psi_n(\hat{\boldsymbol{\theta}}) - \Psi(\boldsymbol{\theta}) \Vert\xrightarrow{P}0, \inf_{d(\boldsymbol{\theta}, \boldsymbol{\theta}_0)} \Vert \Psi(\boldsymbol{\theta})\Vert > 0 = \Vert \Psi(\boldsymbol{\theta}_0)\Vert$. Then any sequence of estimators $\hat{\boldsymbol{\theta}}$ such that $\Psi_n(\hat{\boldsymbol{\theta}}) = o_p(1)$ converges in probability to $\boldsymbol{\theta}_0$. 

Let $\Psi_n(\boldsymbol{\theta}) = \frac{1}{n} \dot{l}_g^*(\theta), \Psi(\boldsymbol{\theta})=\dot{l}_f(\theta)$, we have $\Vert \frac{1}{n} \dot{l}_g^*(\theta) - \dot{l}_f(\theta) \Vert \xrightarrow{P|\mathcal{D}_n}0$, $\dot{l}_f(\hat{\boldsymbol{\theta}}) = 0$ and $\dot{l}_g^*(\tilde{\boldsymbol{\theta}}_g)$. According to the above theorem, we have $\Vert \tilde{\boldsymbol{\theta}}_g - \hat{\boldsymbol{\theta}} \Vert = o_{p|\mathcal{D}_n}(1)$.

Now that we have shown the distance between $\tilde{\boldsymbol{\theta}}$ and $\hat{\boldsymbol{\theta}}$ is close when $r$ and $n$ is large, we aim to find out how close it is. 

Using the expression of Taylor expansion from Thomas S.Ferguson (1996), we have:
\begin{equation*} \label{thms2}
    0 = \frac{1}{n} \frac{\partial l_g^*(\tilde{\boldsymbol{\theta}}_g)}{\partial \theta_j} = \frac{1}{n}\frac{\partial l_g^*(\hat{\boldsymbol{\theta}})}{\partial \theta_j} + \frac{1}{n} \frac{\partial^2 l_g^*(\hat{\boldsymbol{\theta}})}{\partial \boldsymbol{\theta}^\top\partial \theta_j}(\tilde{\boldsymbol{\theta}_g} - \hat{\boldsymbol{\theta}}) + \frac{1}{n}{R_j} 
\end{equation*} 

where 
\begin{equation*} \label{thms3}
    R_j = (\tilde{\boldsymbol{\theta}_g} - \hat{\boldsymbol{\theta}})^\top \int_{0}^{1}\int_{0}^{1}\frac{\partial^3 l_g^* \left[\hat{\boldsymbol{\theta}} + uv(\tilde{\boldsymbol{\theta}}_g - \hat{\boldsymbol{\theta}}) \right]}{\partial\boldsymbol{\theta}\partial\boldsymbol{\theta}^\top\partial\theta_j}vdudv(\tilde{\boldsymbol{\theta}}_g - \hat{\boldsymbol{\theta}}) 
\end{equation*}

Under Assumption \ref{ass3}, we have:
\begin{equation*} \label{thms4}
    \left\Vert \int_{0}^{1}\int_{0}^{1}\frac{\partial^3 l_g^* \left[\hat{\boldsymbol{\theta}} + uv(\tilde{\boldsymbol{\theta}}_g - \hat{\boldsymbol{\theta}}) \right]}{\partial\boldsymbol{\theta}\partial\boldsymbol{\theta}^\top\partial\theta_j}vdudv \right\Vert   = O_{p|\mathcal{D}_n}(n) 
\end{equation*}

Hence,
\begin{equation*} \label{thms5}
    \frac{1}{n} R_j = O_{p|\mathcal{D}_n}(\Vert \tilde{\boldsymbol{\theta}}_g - \hat{\boldsymbol{\theta}}  \Vert^2) 
\end{equation*}

Then we have 
\begin{equation} \label{thms6}
    \tilde{\boldsymbol{\theta}}_g - \hat{\boldsymbol{\theta}} = -\tilde{\boldsymbol{M}_g}^{-1} \left[O_{p|\mathcal{D}_n}(\Vert \tilde{\boldsymbol{\theta}}_g - \hat{\boldsymbol{\theta}}  \Vert^2)+\frac{1}{n}\frac{\partial l_g^*(\hat{\boldsymbol{\theta}})}{\partial\boldsymbol{\theta^\top}}\right] 
\end{equation}

From equation \ref{lemma1.1}, $\tilde{\boldsymbol{M}_g}^{-1} = O_{p|\mathcal{D}_n}(1)$. This result together with $\Vert \tilde{\boldsymbol{\theta}}_g - \hat{\boldsymbol{\theta}} \Vert = o_{p|\mathcal{D}_n}(1)$, equation \ref{lemma1.2} and \ref{thms6} give rise to:

\begin{equation} \label{thms7}
    \tilde{\boldsymbol{\theta}}_g - \hat{\boldsymbol{\theta}} = O_{p|\mathcal{D}_n}(r^{-\frac{1}{2}}) + o_{p|\mathcal{D}_n}(\Vert \tilde{\boldsymbol{\theta}}_g - \hat{\boldsymbol{\theta}} \Vert) = O_{p|\mathcal{D}_n}(r^{-\frac{1}{2}})  
\end{equation}

Then, we move to the verification of asymptotic normality property. 

Let
\begin{equation}
    \frac{1}{n} \dot{l}_g^* (\hat{\boldsymbol{\theta}}) = \frac{1}{nr}\sum_{i=1}^r \frac{1}{\pi_i^*}\dot{l}^*_i(\hat{\boldsymbol{\theta}}) \triangleq \frac{1}{r} \sum_{i=1}^r k_i \notag
\end{equation}

Then, 
\begin{equation} 
    \mathbb{E}(k_i|\mathcal{D}_n) = \frac{1}{n}\sum_{i=1}^n \dot{l}_i(\hat{\boldsymbol{\theta}}) = 0 \notag
\end{equation}

Under Assumption \ref{ass2},
\begin{equation} 
    Var(k_i|\mathcal{D}_n)=\mathbb{E}\left(\frac{1}{n^2\pi_i^2} \dot{l_i}^2(\hat{\boldsymbol{\theta}})\right) = \frac{1}{n^2}\sum_{i=1}^n \frac{1}{\pi_i} \dot{l_i^2}(\hat{\boldsymbol{\theta}}) = O_{p|\mathcal{D}_n}(1) \notag
\end{equation}

As $\boldsymbol{\Lambda_g}=\frac{1}{rn^2}\sum_{i=1}^n\frac{1}{\pi_i}\dot{l}^2_i(\hat{\boldsymbol\theta})$, we have $Var(k_i|\mathcal{D}_n) = r\boldsymbol{\Lambda_g} = O_{p|\mathcal{D}_n}(1)$.

According to Lindeberg-Feller CLT, for each $n$, let $Y_{n,1},Y_{n,2},...,Y_{n,k_n}$ be independent random vectors with finite variance such that for every $\varepsilon > 0$,
\begin{equation}
\begin{aligned}
     \sum_{i=1}^{k_n}\mathbb{E} \Vert &Y_{n,i}\Vert^2 I(\Vert Y_{n,i}\Vert > \varepsilon)\rightarrow 0\\
     &\sum_{i=1}^{k_n} Cov(Y_{n,i})\rightarrow \Sigma
\end{aligned} \notag
\end{equation}

Then,
\begin{equation}
    \sum_{i=1}^{k_n}(Y_{n,i}-\mathbb{E}Y_{n,i})\xrightarrow{d} N(0,\Sigma) \notag
\end{equation}

Therefore, we verify the Lindeberg-Feller condition below. For every $\varepsilon > 0$ and some fixed $\delta >0$, 
\begin{equation*} \label{thms8}
    \begin{aligned}
        &\sum_{i=1}^r \mathbb{E} \left(\Vert r^{-\frac{1}{2}} k_i \Vert I(\Vert k_i \Vert > r^{\frac{1}{2}}\varepsilon)|\mathcal{D}_n \right)\\
        =&\frac{1}{r} \sum_{i=1}^r \mathbb{E}\left(\Vert k_i \Vert^2 I(\Vert k_i \Vert > r^{\frac{1}{2} }\varepsilon)|\mathcal{D}_n   \right)\\
        \leq&\frac{1}{r^{1+\frac{\delta}{2}} \varepsilon^\delta} \sum_{i=1}^r \mathbb{E}\left(\Vert k_i \Vert^{2+\delta} I(\Vert k_i \Vert > r^{\frac{1}{2} }\varepsilon)|\mathcal{D}_n   \right)\\
        \leq&\frac{1}{r^{1+\frac{\delta}{2}} \varepsilon^\delta}\sum_{i=1}^r \mathbb{E}\left(\Vert k_i \Vert^{2+\delta}|\mathcal{D}_n   \right)\\
        =& \frac{1}{r^{\frac{\delta}{2}} \varepsilon^\delta} \sum_{i=1}^r \mathbb{E} \left(\left\Vert \frac{1}{n\pi_i} \dot{l}_i (\hat{\boldsymbol{\theta}}) \right\Vert ^{2+\delta}|\mathcal{D}_n \right)\\
        =& \frac{1}{r^{\frac{\delta}{2}} \varepsilon^\delta} \frac{1}{n^{2+\delta}} \sum_{i=1}^n \frac{1}{\pi_i^{1+\delta}} \Vert \dot{l}_i (\hat{\boldsymbol{\theta}})\Vert^{2+\delta}\\
        =& o_p(1)
    \end{aligned} 
\end{equation*}

The last equality is based on Assumption \ref{ass4}. As the Lindeberg-Feller condition is satisfied, 
\begin{equation*} \label{thms9}
    \frac{1}{n}\boldsymbol{\Lambda_g}^{-\frac{1}{2}}\dot{l}_g^*(\hat{\boldsymbol{\theta}})=r^{-\frac{1}{2}}\left[Var(k_i|\mathcal{D}_n) \right] ^{-\frac{1}{2}}\sum_{i=1}^r k_i \rightarrow N(0,\boldsymbol{I}) 
\end{equation*}

From equation \ref{thms6} and \ref{thms7}, 
\begin{equation} \label{thms10}
     \tilde{\boldsymbol{\theta}}_g - \hat{\boldsymbol{\theta}} = -\frac{1}{n} \tilde{\boldsymbol{M}_g}^{-1} \dot{l}_g^*(\hat{\boldsymbol{\theta}}) + O_{p|\mathcal{D}_n}\left(\frac{1}{r}\right) 
\end{equation}

From equation \ref{lemma1.1}, 
\begin{equation} \label{thms11}
        \boldsymbol{\tilde{M_g}}^{-1} - \boldsymbol{M_g}^{-1} = - \boldsymbol{M_g}^{-1}(\boldsymbol{\tilde{M_g}} - \boldsymbol{M_g})\boldsymbol{\tilde{M_g}}^{-1} = O_{p|\mathcal{D}_n}\left(r^{-\frac{1}{2}}\right)
\end{equation}

From Assumption \ref{ass1} and $r\boldsymbol{V_t} = O_{p|\mathcal{D}_n}(1)$,
\begin{equation} \label{thms12}
    \boldsymbol{V}_g = \boldsymbol{M_g}^{-1} \boldsymbol{\Lambda}_g \boldsymbol{M_g}^{-1} =  \frac{1}{r}\boldsymbol{M_g}^{-1} (r\boldsymbol{\Lambda}_g) \boldsymbol{M_g}^{-1} = O_{p|\mathcal{D}_n}\left(\frac{1}{r}\right) 
\end{equation}

Based on equation \ref{thms10}, \ref{thms11} and \ref{thms12},
\begin{equation}
    \begin{aligned}
        &\boldsymbol{V}_g^{-\frac{1}{2}}(\tilde{\boldsymbol{\theta}}_g - \hat{\boldsymbol{\theta}})\\
        =& -\frac{1}{n}\boldsymbol{V}_g^{-\frac{1}{2}}\tilde{\boldsymbol{M}}_g^{-1} \dot{l}_g^*(\hat{\boldsymbol{\theta}}) + O_{p|\mathcal{D}_n}\left(r^{-\frac{1}{2}}\right) \\
        =& -\frac{1}{n}\boldsymbol{V}_g^{-\frac{1}{2}}\boldsymbol{M}_g^{-1} \dot{l}_g^*(\hat{\boldsymbol{\theta}}) -\frac{1}{n}\boldsymbol{V}_g^{-\frac{1}{2}}(\tilde{\boldsymbol{M}}_g^{-1} - \boldsymbol{M}_g^{-1}) \dot{l}_g^*(\hat{\boldsymbol{\theta}}) \\
        +& O_{p|\mathcal{D}_n}\left(r^{-\frac{1}{2}}\right)\\
        =& -\frac{1}{n}\boldsymbol{V}_g^{-\frac{1}{2}}\boldsymbol{M}_g^{-1} \boldsymbol{\Lambda}_g^{\frac{1}{2}}\boldsymbol{\Lambda}_g^{-\frac{1}{2}}\dot{l}_g^*(\hat{\boldsymbol{\theta}}) + O_{p|\mathcal{D}_n}\left(r^{-\frac{1}{2}}\right) \notag
    \end{aligned}
\end{equation}

Notice that
\begin{equation}
\begin{aligned}
    &(\boldsymbol{V}_g^{-\frac{1}{2}}\boldsymbol{M}_g^{-1} \boldsymbol{\Lambda}_g^{\frac{1}{2}})(\boldsymbol{V}_g^{-\frac{1}{2}}\boldsymbol{M}_g^{-1} \boldsymbol{\Lambda}_g^{\frac{1}{2}})^\top \\
    =& \boldsymbol{V}_g^{-\frac{1}{2}}\boldsymbol{M}_g^{-1} \boldsymbol{\Lambda}_g^{\frac{1}{2}} \boldsymbol{\Lambda}_g^{\frac{1}{2}}\boldsymbol{M}_g^{-1}\boldsymbol{V}_g^{-\frac{1}{2}} \\
    =&\boldsymbol{V}_g^{-\frac{1}{2}} \boldsymbol{V}_g \boldsymbol{V}_g^{-\frac{1}{2}} = \boldsymbol{I} \notag
\end{aligned}
\end{equation}

As a result
\begin{equation}
    \boldsymbol{V}_g^{-\frac{1}{2}}(\tilde{\boldsymbol{\theta}}_g - \hat{\boldsymbol{\theta}})\xrightarrow{d} N(0,\boldsymbol{I}) \notag
\end{equation}
\end{proof}

\subsubsection{Proof of Theorem \ref{thm:general-true}}
    A widely known result is needed for the proof of Theorem \ref{thm:general-true}, which can be seen in chapter 5 of \cite{Van_1998}:
\begin{lemma} \label{lemma:mle_true}
    Under Assumption \ref{ass5} and \ref{ass6}, if $\Lambda(\boldsymbol{\theta})$ is positive definite, we have
    \begin{equation*}
        \sqrt{n}(\hat{\boldsymbol{\theta}} - \boldsymbol{\theta}_0) \xrightarrow{d} N(\boldsymbol{0}, \boldsymbol{V}),
    \end{equation*}
    where $\boldsymbol{V} = \boldsymbol{M}^{-1}_{0g} \Lambda(\boldsymbol{\theta}_0) \boldsymbol{M}^{-1}_{0g}$.
\end{lemma}

The prove of Theorem \ref{thm:general-true}:
\begin{proof}
    We use the characteristic function to prove this result. Let $S_n = \sqrt{r}(\tilde{\boldsymbol{\theta}}_g - \hat{\boldsymbol{\theta}})$, $Y_n = \sqrt{r}(\hat{\boldsymbol{\theta}} - 
    \boldsymbol{\theta}_0)$. Then, $\sqrt{r}(\tilde{\boldsymbol{\theta}}_g - \boldsymbol{\theta}_0) = S_n + Y_n$.
    From Theorem \ref{thm:genenral-mle}, we have $\sqrt{r}(\tilde{\boldsymbol\theta}_g-\hat{\boldsymbol\theta})\xrightarrow{d}N(0,\boldsymbol{V}_g)$. Hence, 
    \begin{equation*}
        \mathbb{E}(e^{\mathrm{i}t^\top S_n}|\mathcal{D}_n) = e^{-0.5t^\top \boldsymbol{V}_g t} + o_{p|\mathcal{D}_n}(1) = e^{-0.5t^\top \boldsymbol{V}_g t} + o_{p}(1)
    \end{equation*}
    The last equation is from Proposition 2 of \cite{Wang2022comparative}.
    For every $k,l = 1,2,...,d$, because of Lipschitz continuity, 
    \begin{equation*}
    \begin{aligned}
        \left\vert \frac{1}{n} \sum_{i=1}^n \ddot{l}_{i(kl)}(\hat{\boldsymbol{\theta}}) - \frac{1}{n} \sum_{i=1}^n \ddot{l}_{i(kl)}(\boldsymbol{\theta}_0) \right\vert \leq& \left\Vert \hat{\boldsymbol{\theta}} - \boldsymbol{\theta}_0 \right\Vert\frac{1}{n} \sum_{i=1}^n \psi(\boldsymbol{x}_i)\\
       =& o_p(1)
    \end{aligned}
    \end{equation*}
    By law of large numbers, 
    \begin{equation*}
        \ddot{l}_f(\hat{\boldsymbol{\theta}}) = \ddot{l}_f(\boldsymbol{\theta}_0) + o_p(1) = \boldsymbol{M}_{0g}(\boldsymbol{\theta}_0) + o_p(1)
    \end{equation*}
    Next we prove $\Lambda_g(\hat{\boldsymbol{\theta}}) = \Lambda_{1g}(\boldsymbol{\theta}_0) + o_p(1)$.
    \begin{equation*}
        \begin{aligned}
            &\left\Vert \Lambda_g(\hat{\boldsymbol{\theta}}) - \Lambda_g(\boldsymbol{\theta}_0) \right\Vert \\
            =& \frac{1}{n} \left\Vert \sum_{i=1}^n \frac{\dot{l}_i(\hat{\boldsymbol{\theta}})^2 - \dot{l}_i(\boldsymbol{\theta}_0)^2 }{n\pi_i}\right\Vert \\
             \leq& \max_i\left(\frac{1}{n\pi_i}\right)\frac{1}{n}\sum_{i=1}^n \left\Vert \dot{l}_i(\hat{\boldsymbol{\theta}})^2 - \dot{l}_i(\hat{\boldsymbol{\theta}})\dot{l}_i(\boldsymbol{\theta}_0)^\top \right\Vert\\
             + &\max_i\left(\frac{1}{n\pi_i}\right)\frac{1}{n}\sum_{i=1}^n \left\Vert \dot{l}_i(\hat{\boldsymbol{\theta}})\dot{l}_i(\boldsymbol{\theta}_0)^\top -\dot{l}_i(\boldsymbol{\theta}_0)^2 \right\Vert\\
            =& \max_i\left(\frac{1}{n\pi_i}\right)\frac{1}{n} \sum_{i=1}^n \left\{\left\Vert \dot{l}_i(\hat{\boldsymbol{\theta}}) \right\Vert + \left\Vert\dot{l}_i(\boldsymbol{\theta}_0)\right\Vert\right\} \times \left\Vert \dot{l}_i(\hat{\boldsymbol{\theta}}) -  \dot{l}_i(\boldsymbol{\theta}_0) \right\Vert
        \end{aligned}
    \end{equation*}
    From Taylor expansion, we have
    \begin{equation*}
        \dot{l}_i(\hat{\boldsymbol{\theta}}) = \dot{l}_i(\boldsymbol{\theta}_0) + R_{i},
    \end{equation*}
    where $R_i = \int_0^1 \ddot{l}_i(\boldsymbol{\theta}_0 + u(\hat{\boldsymbol{\theta}} - \boldsymbol{\theta}_0))du$. From Assumption \ref{ass6}, $R_i$ satisfies
    \begin{equation*}
        \left\Vert R_i - \ddot{l}_i(\boldsymbol{\theta}_0) \right\Vert \leq d\int_0^1 u\psi(\boldsymbol{x}_i)\left\Vert \hat{\boldsymbol{\theta}} - \boldsymbol{\theta}_0 \right\Vert du = \frac{d}{2}\psi(\boldsymbol{x}_i)\left\Vert \hat{\boldsymbol{\theta}} - \boldsymbol{\theta}_0 \right\Vert
    \end{equation*}
    Then, we have
    \begin{equation*}
        \left\Vert \dot{l}_i(\hat{\boldsymbol{\theta}}) -  \dot{l}_i(\boldsymbol{\theta}_0)\right\Vert \leq \frac{d}{2}\psi(\boldsymbol{x}_i)\left\Vert \hat{\boldsymbol{\theta}} - \boldsymbol{\theta}_0 \right\Vert^2 +  \left\Vert \ddot{l}_i(\boldsymbol{\theta}_0)\right\Vert  \left\Vert\hat{\boldsymbol{\theta}} - \boldsymbol{\theta}_0\right\Vert.
    \end{equation*}
    Hence, 
    \begin{equation*}
        \begin{aligned}
            &\left\Vert \Lambda_g(\hat{\boldsymbol{\theta}}) - \Lambda_g(\boldsymbol{\theta}_0) \right\Vert\\
            \leq& \max_i\left(\frac{1}{n\pi_i}\right) \left\{ \frac{d}{2} \left\Vert \hat{\boldsymbol{\theta}} - \boldsymbol{\theta}_0 \right\Vert^2 \left[\frac{1}{n} \sum_{i=1}^n \left\Vert \dot{l}_i(\hat{\boldsymbol{\theta}}) \right\Vert \psi(\boldsymbol{x}_i) \right.\right.\\
             +&\left. \frac{1}{n} \sum_{i=1}^n \left\Vert \dot{l}_i(\boldsymbol{\theta}_0) \right\Vert\psi(\boldsymbol{x}_i) \right] + \left\Vert \hat{\boldsymbol{\theta}} - \boldsymbol{\theta}_0 \right\Vert \left[\frac{1}{n}\sum_{i=1}^n \left\Vert \dot{l}_i(\hat{\boldsymbol{\theta}})\right\Vert \left\Vert\ddot{l}_i(\boldsymbol{\theta}_0)\right\Vert  \right.\\
            +&\left.\left. \frac{1}{n}\sum_{i=1}^n \left\Vert \dot{l}_i(\boldsymbol{\theta}_0)\right\Vert \left\Vert\ddot{l}_i(\boldsymbol{\theta}_0)\right\Vert \right]\right\}
        \end{aligned}
    \end{equation*}
    Using H{\"o}lder's inequality, we have
    \begin{equation*}
    \begin{aligned}
        \frac{1}{n}\sum_{i=1}^n \left\Vert \dot{l}_i(\hat{\boldsymbol{\theta}}) \right\Vert \psi(\boldsymbol{x}_i) \leq& \left\{ \frac{1}{n} \sum_{i=1}^n \left\Vert \dot{l}_i(\hat{\boldsymbol{\theta}}) \right\Vert^4 \right\}^{\frac{1}{4}} \left\{ \frac{1}{n} \sum_{i=1}^n \psi(\boldsymbol{x}_i)^{\frac{4}{3}} \right\}^{\frac{3}{4}} \\
        =& O_p(1)
    \end{aligned}
    \end{equation*}
    Similarly, $\frac{1}{n}\sum_{i=1}^n \left\Vert \dot{l}_i(\boldsymbol{\theta}_0) \right\Vert \psi(\boldsymbol{x}_i)$, $\frac{1}{n}\sum_{i=1}^n \left\Vert \dot{l}_i(\hat{\boldsymbol{\theta}}) \right\Vert \left\Vert \ddot{l}_i(\boldsymbol{\theta}_0) \right\Vert$ and \\
    $\frac{1}{n}\sum_{i=1}^n \left\Vert \dot{l}_i(\boldsymbol{\theta}_0) \right\Vert \left\Vert \ddot{l}_i(\boldsymbol{\theta}_0) \right\Vert$ are all $O_p(1)$. Thus, $\left\Vert \Lambda_g(\hat{\boldsymbol{\theta}}) - \Lambda_g(\boldsymbol{\theta}_0) \right\Vert = o_p(1)$. Then we have
    \begin{equation*}
         \mathbb{E}(e^{\mathrm{i}t^\top S_n}|\mathcal{D}_n) = e^{-0.5t^\top \boldsymbol{M}_{0g}^{-1}\boldsymbol{\Lambda}_{1g}(\boldsymbol{\theta}_0)\boldsymbol{M}_{0g}^{-1} t} + o_{p}(1)
    \end{equation*}
    Using Lemma \ref{lemma:mle_true}, we have
    \begin{equation*}
        \mathbb{E}(e^{\mathrm{i}t^\top Y_n}) \rightarrow e^{-0.5\gamma t^\top\boldsymbol{M}_{0g}^{-1}\Lambda(\boldsymbol{\theta}_0)\boldsymbol{M}_{0g}^{-1} t}
    \end{equation*}
    \begin{equation*}
        \begin{aligned}
            &\left\vert \mathbb{E}\left\{ e^{\mathrm{i}t^\top (S_n+Y_n)} -  e^{\mathrm{i}t^\top Y_n}e^{-0.5\gamma t^\top\boldsymbol{M}_{0g}^{-1}\Lambda(\boldsymbol{\theta}_0)\boldsymbol{M}_{0g}^{-1} t}     \right\}\right\vert \\
            =&\left\vert    \mathbb{E}\left\{ \mathbb{E} \left[e^{\mathrm{i}t^\top (S_n+Y_n)} -  e^{\mathrm{i}t^\top Y_n}e^{-0.5\gamma t^\top\boldsymbol{M}_{0g}^{-1}\Lambda(\boldsymbol{\theta}_0)\boldsymbol{M}_{0g}^{-1}t}|\mathcal{D}_n \right]\right\}\right\vert \\
            =& \left\vert \mathbb{E} \left\{ e^{\mathrm{i}t^\top Y_n} \left[ \mathbb{E}(e^{\mathrm{i}t^\top S_n}|\mathcal{D}_n) - e^{-0.5\gamma t^\top\boldsymbol{M}_{0g}^{-1}\Lambda(\boldsymbol{\theta}_0)\boldsymbol{M}_{0g}^{-1}t} \right]\right\}\right\vert \\
            \leq& \mathbb{E} \left\{ \left\vert \mathbb{E} (e^{\mathrm{i}t^\top S_n}|\mathcal{D}_n) - e^{-0.5\gamma t^\top\boldsymbol{M}_{0g}^{-1}\Lambda(\boldsymbol{\theta}_0)\boldsymbol{M}_{0g}^{-1}t} \right\vert \right\} 
        \end{aligned}
    \end{equation*}
    By the dominated convergence theorem, we have
    \begin{equation*}
        \left\vert \mathbb{E}\left\{ e^{\mathrm{i}t^\top (S_n+Y_n)} -  e^{\mathrm{i}t^\top Y_n}e^{-0.5\gamma t^\top\boldsymbol{M}_{0g}^{-1}\Lambda(\boldsymbol{\theta}_0)\boldsymbol{M}_{0g}^{-1} t}\right\}\right\vert = o(1)
    \end{equation*}
    Thus, 
    \begin{equation*}
    \begin{aligned}
        &\sum_{n=1}^\mathbb{E} \left\{e^{\mathrm{i}t^\top (S_n+Y_n)} \right\} \\
        =& e^{\mathrm{i}t^\top Y_n}e^{-0.5\gamma t^\top\boldsymbol{M}_{0g}^{-1}\Lambda(\boldsymbol{\theta}_0)\boldsymbol{M}_{0g}^{-1} t} + o(1) \\
        \rightarrow&
        \mathbb{E}\left\{ e^{-0.5t^\top \boldsymbol{V}_{0g} t } \right\}
    \end{aligned}
    \end{equation*}
    
    Therefore, Theorem \ref{thm:general-true} is proved.
\end{proof}

\subsubsection{Proof of Theorem \ref{thm:censor-all}}
To prove Theorem \ref{thm:censor-all}, we need the following two lemmas.

\begin{lemma}\label{lemma3}
    Under Assumptions \ref{ass6} and \ref{ass9}, if $\left\Vert \tilde{\boldsymbol{\theta}}_c - \hat{\boldsymbol{\theta}} \right\Vert = o_p(1)$, conditional on $\mathcal{D}_n$, 
    \begin{equation*}
        B_c - \ddot{l}_f(\hat{\boldsymbol{\theta}}) = o_p(1),
    \end{equation*}
    where 
    \begin{equation*}
        \begin{aligned}
            &B_c = \frac{1}{n} \int_0^1 \ddot{l}^*_c \left(\hat{\boldsymbol{\theta}} + s(\tilde{\boldsymbol{\theta}}_c - \hat{\boldsymbol{\theta}})\right) ds, \\
            &\ddot{l}_f(\hat{\boldsymbol{\theta}}) = \frac{1}{n} \sum_{i=1}^n \ddot{l}_i(\hat{\boldsymbol{\theta}}) = \boldsymbol{M}_g.
        \end{aligned}
    \end{equation*}
\end{lemma}

\begin{proof}
Note that
\begin{equation*}
    \mathbb{E}\left( \frac{1}{n} \sum_{i=1}^r \omega_i^c \psi(\boldsymbol{x}^c_i) \Big| \mathcal{D}_n \right) = \frac{1}{n} \sum_{i=1}^n \psi(\boldsymbol{x}_i) = \mathbb{E}\left\{\psi(\boldsymbol{x})\right\} + o_p(1),
\end{equation*}
and
\begin{equation*}
    \begin{aligned}
        &Var \left( \frac{1}{n} \sum_{i=1}^r \omega_i^c \psi(\boldsymbol{x}^c_i) \Big| \mathcal{D}_n \right) \\
        =& \frac{1}{r-n_0}\sum_{i=n_0+1}^{n} \frac{1}{n^2\tilde{\pi}_i} \psi(\boldsymbol{x}_i)^2 \\
        \leq & \max_{i=n_0+1,...,n} \left( \frac{1}{n\tilde{\pi}_i} \right) \frac{\sum_{i=n_0+1}^n \psi(\boldsymbol{x}_i)^2}{n(r-n_0)}  \\
        =& O_p\left(\frac{1}{r-n_0}\right) = O_p\left(r^{-1}\right).
    \end{aligned}
\end{equation*}
Then we have
\begin{equation*}
    \frac{1}{n} \sum_{i=1}^r \omega_i^c \psi(\boldsymbol{x}^c_i) = O_{p|\mathcal{D}_n}(1)
\end{equation*}
For every $k,l =1,...,d$ and $s\in (0,1)$, from Assumption \ref{ass6}, we have
\begin{equation*}
\begin{aligned}
    &\left\vert \frac{1}{n} \ddot{l}^*_{c(kl)} \left(\hat{\boldsymbol{\theta}} + s(\tilde{\boldsymbol{\theta}}_c - \hat{\boldsymbol{\theta}})\right) -  \frac{1}{n} \ddot{l}^*_{c(kl)} \left( \hat{\boldsymbol{\theta}}\right) \right\vert \\
    =& s\left\Vert \tilde{\boldsymbol{\theta}}_c - \hat{\boldsymbol{\theta}} \right\Vert \times \frac{1}{n} \sum_{i=1}^r \omega_i^c \psi(\boldsymbol{x}^c_i) \\
    =&o_p(1) \times O_{p|\mathcal{D}_n}(1) \\
    =& o_p(1).
\end{aligned}
\end{equation*}
Based on Assumption \ref{ass6}, for an arbitrary fixed $\boldsymbol{\theta}$, we have
\begin{equation*}\label{66}
\begin{aligned}
    \frac{1}{n} \sum_{i=1}^n \ddot{l}^2_{i(kl)}(\hat{\boldsymbol{\theta}}) &\leq \frac{2}{n} \sum_{i=1}^2 \ddot{l}^2_{i(kl)}(\boldsymbol{\theta}) + \frac{2\left\Vert \hat{\boldsymbol{\theta}} - \boldsymbol{\theta} \right\Vert^2}{n} \sum_{i=1}^n \psi(\boldsymbol{x}_i)^2 \\
    &=O_p(1)+o_p(1)\\
    &= O_p(1)
\end{aligned}
\end{equation*}
According to this, we have
\begin{equation*}
\begin{aligned}
   &\mathbb{E} \left\{\frac{1}{n} \ddot{l}^*_{c(kl)} \left( \hat{\boldsymbol{\theta}}\right) \Big| \mathcal{D}_n \right\} = \frac{1}{n} \sum_{i=1}^n \ddot{l}_{i(kl)}(\hat{\boldsymbol{\theta}}) = \ddot{l}_{f(kl)}(\hat{\boldsymbol{\theta}})\\
   &Var\left\{\frac{1}{n} \ddot{l}^*_{c(kl)} \left( \hat{\boldsymbol{\theta}}\right) \Big| \mathcal{D}_n \right\}\\
   \leq& \frac{1}{r-n_0} \sum_{i=n_0+1}^n \frac{1}{n^2\tilde{\pi}_i} \ddot{l}_{i(kl)}(\hat{\boldsymbol{\theta}})^2 \\
   \leq& \frac{1}{n(r-n_0)}\max_{i=n_0+1,...,n}\left( \frac{1}{n\pi_i} \right)\sum_{i=n_0+1}^n \ddot{l}_{i(kl)}(\hat{\boldsymbol{\theta}})^2 \\
   =&O_p\left(\frac{1}{r-n_0}\right) = O_p(r^{-1})
\end{aligned}
\end{equation*}
Then, by Chebyshev's inequality,
\begin{equation*}
    \left\vert \frac{1}{n}\ddot{l}^*_{c(kl)}(\hat{\boldsymbol{\theta}}) -\frac{1}{n} \sum_{i=1}^n \ddot{l}_{i(kl)}(\hat{\boldsymbol{\theta}}) \right\vert = O_{P|\mathcal{D}_n}\left(r^{-\frac{1}{2}}\right) = o_{p|\mathcal{D}_n}(1).
\end{equation*}
Therefore, 
\begin{equation*}
\begin{aligned}
    &\left\vert B_{c(kl)} - \ddot{l}_{f(kl)}(\hat{\boldsymbol{\theta}}) \right\vert\\
    \leq& \int_0^1 \left\vert \frac{1}{n} \ddot{l}^*_{c(kl)} \left(\hat{\boldsymbol{\theta}} + s(\tilde{\boldsymbol{\theta}}_c - \hat{\boldsymbol{\theta}})\right) - \frac{1}{n} \sum_{i=1}^n \ddot{l}_{i(kl)}(\hat{\boldsymbol{\theta}}) \right\vert ds \\
    \leq&  \int_0^1 \left\vert\frac{1}{n} \sum_{i=1}^{n_0} \ddot{l}_{i(kl)}\left(\hat{\boldsymbol{\theta}} + s(\tilde{\boldsymbol{\theta}}_c - \hat{\boldsymbol{\theta}})\right) - \frac{1}{n} \sum_{i=1}^{n_0}\ddot{l}_{i(kl)}(\hat{\boldsymbol{\theta}}) \right\vert ds \\
    +& \int_0^1 \left\vert\frac{1}{n} \sum_{i=n_0+1}^r \frac{1}{(r-n_0)\tilde{\pi}_i} \ddot{l}_{i(kl)}\left(\hat{\boldsymbol{\theta}} + s(\tilde{\boldsymbol{\theta}}_c - \hat{\boldsymbol{\theta}})\right) \right.\\
    -&\left. \frac{1}{n}\sum_{i=n_0+1}^r \frac{1}{(r-n_0)\tilde{\pi}_i}\ddot{l}_{i(kl)}(\hat{\boldsymbol{\theta}}) \right\vert ds \\
    +& \left\vert \frac{1}{n}\sum_{i=n_0+1}^r \frac{1}{(r-n_0)\tilde{\pi}_i}\ddot{l}_{i(kl)}(\hat{\boldsymbol{\theta}}) - \frac{1}{n}\sum_{i=n_0+1}^r\ddot{l}_{i(kl)}(\hat{\boldsymbol{\theta}})
    \right\vert \\
    =&o_{p}(1) + o_p(1) + o_{p|\mathcal{D}_n}(1) = o_{p|\mathcal{D}_n}(1).
\end{aligned}
\end{equation*}
\end{proof}

\begin{lemma} \label{lemma4}
    Under Assumptions \ref{ass7} and \ref{ass9}, given $\mathcal{D}_n$ in probability, we have
    \begin{equation*}
       \boldsymbol{\Lambda}^{-\frac{1}{2}}_c \sum_{i=1}^r \frac{\sqrt{r}}{n} \omega_i^c\dot{l}^c_i(\hat{\boldsymbol{\theta}}) \xrightarrow{d} N(0,\boldsymbol{I})
    \end{equation*}
\end{lemma}
\begin{proof}
    We use the Lindeberg-Feller Theorem to prove \citep{Van_1998}.
    Firstly, note that
    \begin{equation*}
    \begin{aligned}
        &\mathbb{E}\left\{\sum_{i=1}^r\frac{\sqrt{r}}{n} \omega_i^c\dot{l}^c_i(\hat{\boldsymbol{\theta}}) \Big| \mathcal{D}_n\right\}= \frac{\sqrt{r}}{n}\sum_{i=1}^n \dot{l}_i(\hat{\boldsymbol{\theta}}) = 0,\\
        &Var\left\{\sum_{i=1}^r \frac{\sqrt{r}}{n} \omega_i^c\dot{l}^c_i(\hat{\boldsymbol{\theta}}) \Big| \mathcal{D}_n\right\} = \boldsymbol{\Lambda}_c\\
        =& \frac{1}{n^2} \sum_{i=n_0+1}^n \frac{1}{\tilde{\pi}_i} \dot{l}_i(\hat{\boldsymbol{\theta}})^2 - \left[ \frac{1}{n}\sum_{i=n_0+1}^n \dot{l}_i(\hat{\boldsymbol{\theta}}) \right]^2 \\
        \leq& \frac{1}{n}\max_{i=n_0+1,...,n}\left(\frac{1}{n\tilde{\pi}_i} \right) \sum_{i=n_0+1}^n \dot{l}_i(\hat{\boldsymbol{\theta}})^2 \\
        =&O_p(1).
    \end{aligned}
    \end{equation*}
    where the inequality is in the sense Loewner ordering. Besides, for $\forall \epsilon>0$ and some $\delta \in (0,2]$,
    \begin{equation*}
    \begin{aligned}
        &\mathbb{E}\left\{ \sum_{i=1}^r \left\Vert \frac{\sqrt{r}}{n}\omega_i^c \dot{l}^c_i(\hat{\boldsymbol{\theta}}) \right\Vert^2 I\left( \left\Vert \frac{\sqrt{r}}{n}\omega_i^c \dot{l}^c_i(\hat{\boldsymbol{\theta}}) \right\Vert>\epsilon \right) \Big| \mathcal{D}_n \right\} \\
        \leq & r^{1+\frac{\delta}{2}}n^{-2-\delta}\epsilon^{-\delta} \mathbb{E}\left\{ \sum_{i=1}^r \left\Vert \omega_i^c \dot{l}^c_i(\hat{\boldsymbol{\theta}}) \right\Vert^{2+\delta} \Big| \mathcal{D}_n \right\} \\
        \leq & \frac{r^{1+\frac{\delta}{2}}}{n^{2+\delta}\epsilon^{\delta}} \left[\sum_{i=1}^{n_0} \left\Vert \dot{l}^c_i(\hat{\boldsymbol{\theta}}) \right\Vert^{2+\delta} + \sum_{i=n_0+1}^n \frac{1}{[(r-n_0)\tilde{\pi}]^{1+\delta}_i} \left\Vert \dot{l}^c_i(\hat{\boldsymbol{\theta}}) \right\Vert^{2+\delta}\right] \\
        \leq & o_p(1)+ \frac{1}{r^{\frac{\delta}{2}}\epsilon^\delta} \max_{i=n_0+1,...,n}\left(\frac{1}{n\tilde{\pi}_i}\right)^{1+\delta} \frac{1}{n} \sum_{i=n_0+1}^n\left\Vert \dot{l}^c_i(\hat{\boldsymbol{\theta}}) \right\Vert^{2+\delta} \\
        =&O_p\left(r^{-\frac{\delta}{2}}\right)
    \end{aligned}
    \end{equation*}
    Hence, the Lindeberg-Feller conditions are satisfied. The central limit theorem shows that conditionally on $\mathcal{D}_n$,
     \begin{equation*}
       \boldsymbol{\Lambda}^{-\frac{1}{2}}_c \sum_{i=1}^r \frac{\sqrt{r}}{n} \omega_i^c\dot{l}^c_i(\hat{\boldsymbol{\theta}}) \xrightarrow{d} N(0,\boldsymbol{I}).
    \end{equation*}
\end{proof}

Now we use the Lemmas \ref{lemma3} and \ref{lemma4} to prove the Theorem \ref{thm:censor-all}.
\begin{proof}
    Directly calculation shows that, for any $\boldsymbol{\theta}$,
\begin{equation*}
\begin{aligned}
    &\mathbb{E} \left\{ \frac{1}{n}l^*_c(\boldsymbol{\theta}) \Big| \mathcal{D}_n \right\} = l_f(\boldsymbol{\theta}), \\
    &Var\left\{ \frac{1}{n}l^*_c(\boldsymbol{\theta}) \Big| \mathcal{D}_n \right\}\\
    &= \frac{1}{r-n_0}\left\{ \frac{1}{n}\sum_{i=n_0+1}^n \frac{1}{n\tilde{\pi}_i}\dot{l}_i^2(\hat{\boldsymbol{\theta}}) - \left[ \frac{1}{n}\sum_{i=n_0+1}^nl_i(\hat{\boldsymbol{\theta}}) \right]^2 \right\} 
\end{aligned} 
\end{equation*}
By Chebyshev's inequality, for any $\epsilon>0$,
\begin{equation*}
\begin{aligned}
     P\left\{\left\vert \frac{1}{n}l^*_c(\boldsymbol{\theta})-l_f(\boldsymbol{\theta}) \right\vert\geq \epsilon 
    \Big| \mathcal{D}_n \right\} &\leq \frac{Var\left\{ \frac{1}{n}l^*_c(\boldsymbol{\theta}) \Big| \mathcal{D}_n \right\}}{\epsilon^2} \\
    &=O_p\left(\frac{1}{r-n_0}\right) = O_p(r^{-1})
\end{aligned}
\end{equation*}
Therefore, for any $\boldsymbol{\theta}$,
\begin{equation*}
    \frac{1}{n}l^*_c(\boldsymbol{\theta})-l_f(\boldsymbol{\theta}) = o_{p|\mathcal{D}_n}(1)
\end{equation*}
Then, from Theorem 5.9 of \cite{Van_1998}, conditionally on $\mathcal{D}_n$,
\begin{equation*}
    \left\Vert \tilde{\boldsymbol{\theta}}_c - \hat{\boldsymbol{\theta}} \right\Vert = o_{p|\mathcal{D}_n}(1) = o_p(1).
\end{equation*}
By Taylor expansion,
\begin{equation*}
    0=\frac{1}{n}\dot{l}^*_c(\tilde{\boldsymbol{\theta}}_c)=\frac{1}{n}\dot{l}^*_c(\hat{\boldsymbol{\theta}})+B_c\cdot(\tilde{\boldsymbol{\theta}}_c-\hat{\boldsymbol{\theta}}),
\end{equation*}
By Lemma \ref{lemma3}, we have
\begin{equation*}
    0=\frac{1}{n}\dot{l}^*_c(\tilde{\boldsymbol{\theta}}_c)=\frac{1}{n}\dot{l}^*_c(\hat{\boldsymbol{\theta}})+\left[\ddot{l}_f(\hat{\boldsymbol{\theta}})+o_p(1)\right]\cdot(\tilde{\boldsymbol{\theta}}_c-\hat{\boldsymbol{\theta}}),
\end{equation*}
which leads to
\begin{equation*}
\begin{aligned}
    \tilde{\boldsymbol{\theta}}_c-\hat{\boldsymbol{\theta}} &= -\left[\ddot{l}_f(\hat{\boldsymbol{\theta}})+o_p(1)\right]^{-1}\frac{1}{n}\dot{l}^*_c(\hat{\boldsymbol{\theta}})\\
    &=-\left[\ddot{l}_f(\hat{\boldsymbol{\theta}})+o_p(1)\right]^{-1}\frac{1}{n}\sum_{i=1}^r\omega_i^c\dot{l}^c_i(\hat{\boldsymbol{\theta}}) \\
    &= -\frac{1}{\sqrt{r}} \left[\ddot{l}_f(\hat{\boldsymbol{\theta}})+o_p(1)\right]^{-1} \boldsymbol{\Lambda}_c^{\frac{1}{2}}\cdot \sqrt{r}\boldsymbol{\Lambda}_c^{-\frac{1}{2}}\frac{1}{n}\sum_{i=1}^r\omega_i^c\dot{l}^c_i(\hat{\boldsymbol{\theta}})
\end{aligned}
\end{equation*}
Note that $\ddot{l}_f(\hat{\boldsymbol{\theta}}) = \boldsymbol{M}_g$. From Lemma \ref{lemma4} and Slutsky's theorem, we can derive that conditionally on $\mathcal{D}_n$,
\begin{equation*}
    \sqrt{r} (\tilde{\boldsymbol{\theta}}_c - \hat{\boldsymbol{\theta}}) \xrightarrow{d} N(0,\boldsymbol{V}_c).
\end{equation*}
where
\begin{equation*}
\begin{aligned}
    \boldsymbol{V}_c &= \boldsymbol{M}_g^{-1} \boldsymbol{\Lambda}_c \boldsymbol{M}_g^{-1}\\
    \boldsymbol{\Lambda}_c &= \frac{1}{n^2} \sum_{i=n_0+1}^n \frac{1}{\tilde{\pi}_i} \dot{l}_i(\hat{\boldsymbol{\theta}})^2 - \left[ \frac{1}{n}\sum_{i=n_0+1}^n \dot{l}_i(\hat{\boldsymbol{\theta}}) \right]^2\\
    \boldsymbol{M}_g &= \frac{1}{n}\sum_{i=1}^n\ddot{l}_i(\hat{\boldsymbol\theta})
\end{aligned}
\end{equation*}

\end{proof}

\subsubsection{Proof of Theorem \ref{thm:censor-true}}
The proof of Theorem \ref{thm:censor-true} is similar to that of Theorem \ref{thm:general-true}.
\begin{proof}
Let $S_c = \sqrt{r}(\tilde{\boldsymbol{\theta}}_c - \hat{\boldsymbol{\theta}})$ and $Y_c = \sqrt{r}(\hat{\boldsymbol{\theta}} - \boldsymbol{\theta}_0)$, we have
\begin{equation*}
    \sqrt{r}(\tilde{\boldsymbol{\theta}}_c - \boldsymbol{\theta}_0) = S_c+Y_c.
\end{equation*}
Under Theorem \ref{thm:censor-all}, we have
\begin{equation*}
\begin{aligned}
    \mathbb{E}(e^{\mathrm{i}t^\top S_c} | \mathcal{D}_n) = e^{-0.5t^\top\boldsymbol{V}_ct} + o_{p|\mathcal{D}_n}(1)=e^{-0.5t^\top\boldsymbol{V}_ct} + o_{p}(1)
\end{aligned}
\end{equation*}
For every $k,l = 1,2,...,d$, from Assumption \ref{ass6}, we have
\begin{equation*}
\begin{aligned}
    \left\vert \frac{1}{n} \sum_{i=1}^n \ddot{l}_{i(kl)}(\hat{\boldsymbol{\theta}}) - \frac{1}{n}\sum_{i=1}^n\ddot{l}_{i(kl)}(\boldsymbol{\theta}_0) \right\vert \leq& \left\Vert \hat{\boldsymbol{\theta}} - \boldsymbol{\theta}_0 \right\Vert \cdot \frac{1}{n}\sum_{i=1}^n \psi(\boldsymbol{x}_i)\\
    =&o_p(1).
\end{aligned}
\end{equation*}
Hence, 
\begin{equation*}
    \ddot{l}_f(\hat{\boldsymbol{\theta}}) = \frac{1}{n} \sum_{i=1}^n \ddot{l}_i(\hat{\boldsymbol{\theta}}) = \frac{1}{n} \ddot{l}_f(\boldsymbol{\theta}_0) + o_p(1) = \mathbb{E}\left[\ddot{l}(\boldsymbol{x},\boldsymbol{\theta}_0)\right]+o_p(1).
\end{equation*}
Next, we consider the relation between $\boldsymbol{\Lambda}_c(\hat{\boldsymbol{\theta}})$ and $\boldsymbol{\Lambda}_{0c}(\boldsymbol{\theta}_0)$.
\begin{equation*}
\begin{aligned}
    &\left\Vert \boldsymbol{\Lambda}_c(\hat{\boldsymbol{\theta}})-\boldsymbol{\Lambda}_{0c}(\boldsymbol{\theta}_0) \right\Vert = \left\Vert \frac{1}{n}\sum_{i=n_0+1}^n \frac{1}{\tilde{\pi}_i} \dot{l}_i(\hat{\boldsymbol{\theta}})^2 - \left[ \frac{1}{n}\sum_{i=n_0+1}^n \dot{l}_i(\hat{\boldsymbol{\theta}})\right]^2 \right.\\
    &\left. - \frac{1}{n}\sum_{i=n_0+1}^n \frac{1}{\tilde{\pi}_i} \dot{l}_i(\boldsymbol{\theta}_0)^2 +\left[ \frac{1}{n}\sum_{i=n_0+1}^n \dot{l}_i(\boldsymbol{\theta}_0)\right]^2 \right\Vert \\
    &\leq \frac{1}{n} \left\Vert \sum_{i=n_0+1}^n\frac{1}{n\tilde{\pi}_i}\left[\dot{l}_i(\hat{\boldsymbol{\theta}})^2 - \dot{l}_i(\boldsymbol{\theta}_0)^2\right] \right\Vert + \\
    &\frac{1}{n^2}\left\Vert \left[\sum_{i=n_0+1}^n \dot{l}_i(\hat{\boldsymbol{\theta}})\right]^2 - \left[\sum_{i=n_0+1}^n \dot{l}_i(\boldsymbol{\theta}_0)\right]^2  \right\Vert.
\end{aligned}
\end{equation*}
We analyse the two parts of the last equality respectively. 
\begin{equation*}
\begin{aligned}
    &\frac{1}{n} \left\Vert \sum_{i=n_0+1}^n\frac{1}{n\tilde{\pi}_i}\left[\dot{l}_i(\hat{\boldsymbol{\theta}})^2 - \dot{l}_i(\boldsymbol{\theta}_0)^2\right] \right\Vert \\
    & \leq \max_{i=n_0+1,...,n} \left(\frac{1}{n\tilde{\pi}_i} \right)\left[\frac{1}{n}\sum_{i=n_0+1}^n\left\Vert \dot{l}_i(\hat{\boldsymbol{\theta}})^2 - \dot{l}_i(\hat{\boldsymbol{\theta}})\dot{l}_i(\boldsymbol{\theta}_0)\right\Vert \right.\\
    &\left. + \frac{1}{n}\sum_{i=n_0+1}^n\left\Vert\dot{l}_i(\hat{\boldsymbol{\theta}})\dot{l}_i(\boldsymbol{\theta}_0)- \dot{l}_i(\boldsymbol{\theta}_0)^2  \right\Vert\right] \\
    & = \max_{i=n_0+1,...,n} \left(\frac{1}{n\tilde{\pi}_i} \right) \frac{1}{n}\sum_{i=n_0+1}^n \left[\left\Vert \dot{l}_i(\hat{\boldsymbol{\theta}}) \right\Vert +\left\Vert \dot{l}_i(\boldsymbol{\theta}_0) \right\Vert\right]\cdot\\
    &\quad\left\Vert \dot{l}_i(\hat{\boldsymbol{\theta}})-\dot{l}_i(\boldsymbol{\theta}_0) \right\Vert
\end{aligned}
\end{equation*}
By Taylor expansion, we have
\begin{equation*}
\begin{aligned}
    \dot{l}_i(\hat{\boldsymbol{\theta}}) =  \dot{l}_i(\boldsymbol{\theta}_0) + (\hat{\boldsymbol{\theta}}-\boldsymbol{\theta}_0)B_{ci}\\
    B_{ci}=\int_0^1 \ddot{l}_i(\boldsymbol{\theta}_0+s(\hat{\boldsymbol{\theta}}-\boldsymbol{\theta}_0))ds
\end{aligned}
\end{equation*}
Then, using Assumption \ref{ass6},
\begin{equation*}
\begin{aligned}
    \left\Vert B_{ci}- \ddot{l}_i(\boldsymbol{\theta}_0) \right\Vert &= \int_0^1\left[\ddot{l}_i(\boldsymbol{\theta}_0+s(\hat{\boldsymbol{\theta}}-\boldsymbol{\theta}_0)) - \ddot{l}_i(\boldsymbol{\theta}_0)\right]ds \\
    &\leq d\int_0^1 \left\Vert s(\hat{\boldsymbol{\theta}}-\boldsymbol{\theta}_0) \right\Vert \psi(\boldsymbol{x}_i)ds \\
    &= 0.5d\psi(\boldsymbol{x}_i) \left\Vert \hat{\boldsymbol{\theta}}-\boldsymbol{\theta}_0 \right\Vert
\end{aligned}
\end{equation*}
Hence, 
\begin{equation*}
\begin{aligned}
    &\left\Vert \dot{l}_i(\hat{\boldsymbol{\theta}}) -  \dot{l}_i(\boldsymbol{\theta}_0) \right\Vert = \left\Vert \hat{\boldsymbol{\theta}}-\boldsymbol{\theta}_0 \right\Vert \left\Vert B_{ci}\right\Vert \\
    &\leq 0.5d\psi(\boldsymbol{x}_i) \left\Vert \hat{\boldsymbol{\theta}}-\boldsymbol{\theta}_0 \right\Vert^2 + \left\Vert \ddot{l}_i(\boldsymbol{\theta}_0) \right\Vert \left\Vert \hat{\boldsymbol{\theta}}-\boldsymbol{\theta}_0 \right\Vert 
\end{aligned}
\end{equation*}
Therefore, 
\begin{equation*}
\begin{aligned}
    &\frac{1}{n} \left\Vert \sum_{i=n_0+1}^n\frac{1}{n\tilde{\pi}_i}\left[\dot{l}_i(\hat{\boldsymbol{\theta}})^2 - \dot{l}_i(\boldsymbol{\theta}_0)^2\right] \right\Vert \\
    &\leq \max_{i=n_0+1,...,n} \left(\frac{1}{n\tilde{\pi}_i} \right) \frac{1}{n} \sum_{i=n_0+1}^n \\
    &\left[ 0.5d\psi(\boldsymbol{x}_i) \left\Vert \hat{\boldsymbol{\theta}}-\boldsymbol{\theta}_0 \right\Vert^2 + \left\Vert \ddot{l}_i(\boldsymbol{\theta}_0) \right\Vert \left\Vert \hat{\boldsymbol{\theta}}-\boldsymbol{\theta}_0 \right\Vert  \right]\cdot \\
    & \left[\left\Vert \dot{l}_i(\hat{\boldsymbol{\theta}}) \right\Vert +\left\Vert \dot{l}_i(\boldsymbol{\theta}_0) \right\Vert\right] \\
    &\leq \max_{i=n_0+1,...,n} \left(\frac{1}{n\tilde{\pi}_i} \right) \left\{ 0.5d\left\Vert \hat{\boldsymbol{\theta}}-\boldsymbol{\theta}_0 \right\Vert^2 \right.\\
    &\left.\left[\sum_{i=n_0+1}^n \left\Vert \dot{l}_i(\boldsymbol{\theta}_0) \right\Vert \psi(\boldsymbol{x}_i) + \sum_{i=n_0+1}^n \left\Vert \dot{l}_i(\hat{\boldsymbol{\theta}}) \right\Vert \psi(\boldsymbol{x}_i)\right]\right. \\
    &\left.+ \left\Vert \hat{\boldsymbol{\theta}}-\boldsymbol{\theta}_0 \right\Vert \left[ \frac{1}{n}\sum_{i=n_0+1}^n \left\Vert \ddot{l}_i(\boldsymbol{\theta}_0) \right\Vert \left\Vert \dot{l}_i(\hat{\boldsymbol{\theta}}) \right\Vert +\right.\right.\\
    &\left.\left.\frac{1}{n}\sum_{i=n_0+1}^n \left\Vert \ddot{l}_i(\boldsymbol{\theta}_0) \right\Vert \left\Vert \dot{l}_i(\boldsymbol{\theta}_0) \right\Vert \right]
    \right\}
\end{aligned}
\end{equation*}
Using H{\"o}lder's inequality, as $n_0$ is finite, $n\rightarrow\infty$, we have
\begin{equation*}
\begin{aligned}
    &\frac{1}{n}\sum_{i=n_0+1}^n \left\Vert \dot{l}_i(\hat{\boldsymbol{\theta}}) \right\Vert \psi(\boldsymbol{x}_i) \\
    &\leq \frac{n-n_0}{n}\left\{ \frac{1}{n-n_0} \sum_{i=n_0+1}^n \left\Vert \dot{l}_i(\hat{\boldsymbol{\theta}}) \right\Vert^4 \right\}^{\frac{1}{4}} \\
    &\cdot\left\{ \frac{1}{n-n_0} \sum_{i=n_0+1}^n \psi(\boldsymbol{x}_i)^{\frac{4}{3}} \right\}^{\frac{3}{4}}\\
    &= O_p(1)
\end{aligned}
\end{equation*}
Similarly, $\frac{1}{n}\sum_{i=n_0+1}^n \left\Vert \dot{l}_i(\boldsymbol{\theta}_0) \right\Vert \psi(\boldsymbol{x}_i)$, \\
$\frac{1}{n}\sum_{i=n_0+1}^n \left\Vert \dot{l}_i(\hat{\boldsymbol{\theta}}) \right\Vert \left\Vert \ddot{l}_i(\boldsymbol{\theta}_0) \right\Vert$ and \\
$\frac{1}{n}\sum_{i=n_0+1}^n \left\Vert \dot{l}_i(\boldsymbol{\theta}_0) \right\Vert \left\Vert \ddot{l}_i(\boldsymbol{\theta}_0) \right\Vert$ are all $O_p(1)$. Hence,
\begin{equation*}
    \frac{1}{n} \left\Vert \sum_{i=n_0+1}^n\frac{1}{n\tilde{\pi}_i}\left[\dot{l}_i(\hat{\boldsymbol{\theta}})^2 - \dot{l}_i(\boldsymbol{\theta}_0)^2\right] \right\Vert = o_p(1).
\end{equation*}
Besides, by Assumption \ref{ass7}
\begin{equation*}
\begin{aligned}
    &\frac{1}{n^2}\left\Vert \left[\sum_{i=n_0+1}^n \dot{l}_i(\hat{\boldsymbol{\theta}})\right]^2 - \left[\sum_{i=n_0+1}^n \dot{l}_i(\boldsymbol{\theta}_0)\right]^2\right\Vert\\
    &\leq \left\Vert \frac{1}{n}\sum_{i=n_0+1}^n \dot{l}_i(\boldsymbol{\theta}_0)+\frac{1}{n}\sum_{i=n_0+1}^n\dot{l}_i(\hat{\boldsymbol{\theta}}) \right\Vert \cdot 
    \left\Vert \frac{1}{n}\sum_{i=n_0+1}^n \dot{l}_i(\hat{\boldsymbol{\theta}})-\right.\\
    &\left.\frac{1}{n}\sum_{i=n_0+1}^n\dot{l}_i(\boldsymbol{\theta}_0) \right\Vert \\
    &\leq O_p(1) \cdot \left(\frac{1}{n}\sum_{i=n_0+1}^n 0.5d\psi(\boldsymbol{x}_i)\left\Vert \hat{\boldsymbol{\theta}}-\boldsymbol{\theta}_0\right\Vert^2 +\right.\\
    &\left.\frac{1}{n}\sum_{i=n_0+1}^n\left\Vert \hat{\boldsymbol{\theta}}-\boldsymbol{\theta}_0\right\Vert\left\Vert \ddot{l}_i(\boldsymbol{\theta}_0) \right\Vert \right)\\
    &=O_p(1) \cdot o_p(1) = o_p(1).
\end{aligned}
\end{equation*}
Hence, 
\begin{equation*}
\begin{aligned}
    &\left\Vert \boldsymbol{\Lambda}_c(\hat{\boldsymbol{\theta}})-\boldsymbol{\Lambda}_{0c}(\boldsymbol{\theta}_0) \right\Vert \\
    &\leq \frac{1}{n} \left\Vert \sum_{i=n_0+1}^n\frac{1}{n\tilde{\pi}_i}\left[\dot{l}_i(\hat{\boldsymbol{\theta}})^2 - \dot{l}_i(\boldsymbol{\theta}_0)^2\right] \right\Vert +  \\
    &\frac{1}{n^2}\left\Vert \left[\sum_{i=n_0+1}^n \dot{l}_i(\hat{\boldsymbol{\theta}})\right]^2 - \left[\sum_{i=n_0+1}^n \dot{l}_i(\boldsymbol{\theta}_0)\right]^2  \right\Vert = o_p(1).
\end{aligned}
\end{equation*}
Then, 
\begin{equation*}
    \mathbb{E}(e^{\mathrm{i}t^\top S_c}|\mathcal{D}_n) = e^{-0.5t^\top \boldsymbol{M}_{0g}^{-1}\Lambda_{1c}(\boldsymbol{\theta}_0)\boldsymbol{M}_{0g}^{-1} t} + o_{p}(1)
\end{equation*}
Note that  $Y_c = \sqrt{\frac{r}{n}}\sqrt{n}(\hat{\boldsymbol{\theta}} - \boldsymbol{\theta}_0)$. Using the property of $\hat{\boldsymbol{\theta}}$, we have
\begin{equation*}
    \mathbb{E}(e^{\mathrm{i}t^\top Y_c}) \rightarrow e^{-0.5\gamma t^\top\boldsymbol{M}_{0g}^{-1}\Lambda(\boldsymbol{\theta}_0)\boldsymbol{M}_{0g}^{-1} t}
\end{equation*}
\begin{equation*}
    \begin{aligned}
         &\left\vert \mathbb{E}\left\{ e^{\mathrm{i}t^\top (S_c+Y_c)} -  e^{\mathrm{i}t^\top Y_c}e^{-0.5\gamma t^\top\boldsymbol{M}_{0g}^{-1}\Lambda(\boldsymbol{\theta}_0)\boldsymbol{M}_{0g}^{-1} t}     \right\}\right\vert \\
        =&\left\vert   \mathbb{E}\left\{ \mathbb{E} \left[e^{\mathrm{i}t^\top (S_c+Y_c)} -  e^{\mathrm{i}t^\top Y_c}e^{-0.5\gamma t^\top\boldsymbol{M}_{0g}^{-1}\Lambda(\boldsymbol{\theta}_0)\boldsymbol{M}_{0g}^{-1}t}|\mathcal{D}_n \right]\right\}\right\vert \\
        =& \left\vert \mathbb{E} \left\{ e^{\mathrm{i}t^\top Y_c} \left[ \mathbb{E}(e^{\mathrm{i}t^\top S_c}|\mathcal{D}_n) - e^{-0.5\gamma t^\top\boldsymbol{M}_{0g}^{-1}\Lambda(\boldsymbol{\theta}_0)\boldsymbol{M}_{0g}^{-1}t} \right]\right\}\right\vert \\
        \leq& \mathbb{E} \left\{ \left\vert \mathbb{E} (e^{\mathrm{i}t^\top S_c}|\mathcal{D}_n) - e^{-0.5\gamma t^\top\boldsymbol{M}_{0g}^{-1}\Lambda(\boldsymbol{\theta}_0)\boldsymbol{M}_{0g}^{-1}t} \right\vert \right\} 
    \end{aligned}
\end{equation*}
By the dominated convergence theorem, we have
\begin{equation*}
    \left\vert \mathbb{E}\left\{ e^{\mathrm{i}t^\top (S_c+Y_c)} -  e^{\mathrm{i}t^\top Y_c}e^{-0.5\gamma t^\top\boldsymbol{M}_{0g}^{-1}\Lambda(\boldsymbol{\theta}_0)\boldsymbol{M}_{0g}^{-1} t}\right\}\right\vert = o(1)
\end{equation*}
Thus, 
\begin{equation*}
\begin{aligned}
    \mathbb{E} \left\{e^{\mathrm{i}t^\top (S_c+Y_c)} \right\} = &e^{\mathrm{i}t^\top Y_c}e^{-0.5\gamma t^\top\boldsymbol{M}_{0g}^{-1}\Lambda(\boldsymbol{\theta}_0)\boldsymbol{M}_{0g}^{-1} t} + o(1) \\
    \rightarrow&
    \mathbb{E}\left\{ e^{-0.5t^\top \boldsymbol{V}_{0c} t } \right\}
\end{aligned}
\end{equation*}
Therefore, Theorem \ref{thm:censor-true} is proved.

\end{proof}

\subsubsection{Proof of Theorem \ref{ssp:general A-opt}}
\begin{proof}
As $\boldsymbol{V}_g = \boldsymbol{M_g}^{-1} \boldsymbol{\Lambda_g}\boldsymbol{M_g}^{-1}= \boldsymbol{M_g}^{-1} [\frac{1}{rn^2}\sum_{i=1}^n \frac{1}{\pi_i} \dot{l}^2_i(\hat{\boldsymbol{\theta}})]\boldsymbol{M_g}^{-1}$, 
    \begin{equation}
        \begin{aligned}
            tr(\boldsymbol{V}_g) &= \frac{1}{rn^2}tr \left[\boldsymbol{M_g}^{-1}\sum_{i=1}^n \frac{1}{\pi_i} \dot{l}^2_i(\hat{\boldsymbol{\theta}})\boldsymbol{M_g}^{-1} \right] \\
            &= \frac{1}{rn^2}\sum_{i=1}^n \frac{1}{\pi_i}tr(\boldsymbol{M_g}^{-1} \dot{l}_i(\hat{\boldsymbol{\theta}}) \dot{l}_i(\hat{\boldsymbol{\theta}})^\top \boldsymbol{M_g}^{-1})\\
            &= \frac{1}{rn^2}\left(\sum_{i=1}^n \frac{1}{\pi_i} \Vert\boldsymbol{M_g}^{-1} \dot{l}_i(\hat{\boldsymbol{\theta}})\Vert^2\right)\left(\sum_{i=1}^n \pi_i \right) \\
            &\geq \frac{1}{rn^2}\left(\sum_{i=1}^n \Vert\boldsymbol{M_g}^{-1} \dot{l}_i(\hat{\boldsymbol{\theta}})\Vert\right)^2
        \end{aligned}
    \end{equation}
    
The last inequality is based on the Cauchy-Schwarz inequality and it holds if and only if when $\pi_i \propto \Vert\boldsymbol{M_g}^{-1} \dot{l}_i(\hat{\boldsymbol{\theta}})\Vert$.
\end{proof}

\subsubsection{Proof of Theorem \ref{ssp:general L-opt}}
\begin{proof}
    \begin{equation*}
        \begin{aligned}
            tr(\boldsymbol{\Lambda}_g) &= \frac{1}{rn^2}tr \left[\sum_{i=1}^n \frac{1}{\pi_i} \dot{l}^2_i(\hat{\boldsymbol{\theta}}) \right] \\
            &= \frac{1}{rn^2}\left(\sum_{i=1}^n \frac{1}{\pi_i} \Vert \dot{l}_i(\hat{\boldsymbol{\theta}})\Vert^2\right)\left(\sum_{i=1}^n \pi_i \right) \\
            &\geq \frac{1}{rn^2}\left(\sum_{i=1}^n \Vert\dot{l}_i(\hat{\boldsymbol{\theta}})\Vert\right)^2
        \end{aligned}
    \end{equation*}
The last inequality is based on the Cauchy-Schwarz inequality and it holds if and only if when $\pi_i \propto \Vert \dot{l}_i(\hat{\boldsymbol{\theta}})\Vert$.
\end{proof}

\subsubsection{Proof of Proposition \ref{thm:RDS-mle}}
To prove Proposition \ref{thm:RDS-mle}, we need the following two lemmas.

\begin{lemma}\label{lemma5}
    Under Assumptions \ref{ass6}, if $\left\Vert \tilde{\boldsymbol{\theta}}_{RDS} - \hat{\boldsymbol{\theta}} \right\Vert = o_p(1)$, conditional on $\mathcal{D}_n$, 
    \begin{equation*}
        B_p - \ddot{l}_f(\hat{\boldsymbol{\theta}}) = o_p(1),
    \end{equation*}
    where 
    \begin{equation*}
        \begin{aligned}
            &B_p = \frac{1}{n} \int_0^1 \ddot{l}^* \left(\hat{\boldsymbol{\theta}} + s(\tilde{\boldsymbol{\theta}}_{RDS} - \hat{\boldsymbol{\theta}})\right) ds, \\
            &\ddot{l}_f(\hat{\boldsymbol{\theta}}) = \frac{1}{n} \sum_{i=1}^n \ddot{l}_i(\hat{\boldsymbol{\theta}}) = \boldsymbol{M}_g.
        \end{aligned}
    \end{equation*}
\end{lemma}

\begin{proof}
For every $k,l =1,...,d$ and $s\in (0,1)$, from Assumption \ref{ass7}, we have
\begin{equation*}
\begin{aligned}
    &\left\vert \frac{1}{n} \ddot{l}^*_g{(kl)} \left(\hat{\boldsymbol{\theta}} + s(\tilde{\boldsymbol{\theta}}_{RDS} - \hat{\boldsymbol{\theta}})\right) -  \frac{1}{n} \ddot{l}^*_{g(kl)} \left( \hat{\boldsymbol{\theta}}\right) \right\vert \\
    =& s\left\Vert \tilde{\boldsymbol{\theta}}_{RDS} - \hat{\boldsymbol{\theta}} \right\Vert \max_{i=1,...,n}\left(\frac{1}{n\pi_{\xi,i}^{RDS}(\tilde{\boldsymbol{\theta}}_p)}\right)\times \frac{1}{r} \sum_{i=1}^r \psi(\boldsymbol{x}^*_i) \\
    =&o_p(1) \times O_{p|\mathcal{D}_n,\tilde{\boldsymbol{\theta}}_p}(1) \\
    =& o_p(1).
\end{aligned}
\end{equation*}
Based on Assumption \ref{ass6}, for an arbitrary fixed $\boldsymbol{\theta}$, we have
\begin{equation*}
\begin{aligned}
    \frac{1}{n} \sum_{i=1}^n \ddot{l}^2_{i(kl)}(\hat{\boldsymbol{\theta}}) &\leq \frac{2}{n} \sum_{i=1}^2 \ddot{l}^2_{i(kl)}(\boldsymbol{\theta}) + \frac{2\left\Vert \hat{\boldsymbol{\theta}} - \boldsymbol{\theta} \right\Vert^2}{n} \sum_{i=1}^n \psi(\boldsymbol{x}_i)^2 \\
    &=O_p(1)+o_p(1)\\
    &= O_p(1)
\end{aligned}
\end{equation*}
According to this, we have
\begin{equation*}
\begin{aligned}
   &\mathbb{E} \left\{\frac{1}{n} \ddot{l}^*_{g(kl)} \left( \hat{\boldsymbol{\theta}}\right) \Big| \mathcal{D}_n \right\} = \frac{1}{n} \sum_{i=1}^n \ddot{l}_{i(kl)}(\hat{\boldsymbol{\theta}}) = \ddot{l}_{f(kl)}(\hat{\boldsymbol{\theta}})\\
   &Var\left\{\frac{1}{n} \ddot{l}^*_{g(kl)} \left( \hat{\boldsymbol{\theta}}\right) \Big| \mathcal{D}_n \right\} \\
   &\leq \frac{1}{r} \sum_{i=1}^n \frac{1}{n^2\pi_{\xi,i}^{RDS}(\tilde{\boldsymbol{\theta}}_p)} \ddot{l}_{i(kl)}(\hat{\boldsymbol{\theta}})^2 \\
   &\leq \frac{1}{nr}\max_{i=1,...,n}\left( \frac{1}{n\pi_{\xi,i}^{RDS}(\tilde{\boldsymbol{\theta}}_p)} \right)\sum_{i=1}^n \ddot{l}_{i(kl)}(\hat{\boldsymbol{\theta}})^2 \\
   &= O_p(r^{-1})
\end{aligned}
\end{equation*}
Then, by Chebyshev's inequality,
\begin{equation*}
    \left\vert \frac{1}{n}\ddot{l}^*_{g(kl)}(\hat{\boldsymbol{\theta}}) -\frac{1}{n} \sum_{i=1}^n \ddot{l}_{i(kl)}(\hat{\boldsymbol{\theta}}) \right\vert = O_{P|\mathcal{D}_n,\tilde{\boldsymbol{\theta}}_p}\left(r^{-\frac{1}{2}}\right) = o_{p|\mathcal{D}_n,\tilde{\boldsymbol{\theta}}_p}(1).
\end{equation*}
Therefore, 
\begin{equation*}
\begin{aligned}
    &\left\vert B_{p(kl)} - \ddot{l}_{f(kl)}(\hat{\boldsymbol{\theta}}) \right\vert\\
    &\leq \int_0^1 \left\vert \frac{1}{n} \ddot{l}^*_{g(kl)} \left(\hat{\boldsymbol{\theta}} + s(\tilde{\boldsymbol{\theta}}_{RDS} - \hat{\boldsymbol{\theta}})\right) - \frac{1}{n} \sum_{i=1}^n \ddot{l}_{i(kl)}(\hat{\boldsymbol{\theta}}) \right\vert ds \\
    & \leq \int_0^1 \left\vert\frac{1}{n} \ddot{l}^*_{g(kl)} \left(\hat{\boldsymbol{\theta}} + s(\tilde{\boldsymbol{\theta}}_{RDS} - \hat{\boldsymbol{\theta}})\right) - \frac{1}{r}\sum_{i=1}^r \frac{\ddot{l}^*_{i(kl)}(\hat{\boldsymbol{\theta}})}{n\pi_{\xi,i}^{RDS*}(\tilde{\boldsymbol{\theta}}_p)}\right\vert \\
    &+\left\vert \frac{1}{r}\sum_{i=1}^r \frac{\ddot{l}^*_{i(kl)}(\hat{\boldsymbol{\theta}})}{n\pi_{\xi,i}^{RDS*}(\tilde{\boldsymbol{\theta}}_p)} - \frac{1}{n} \sum_{i=1}^{n}\ddot{l}_{i(kl)}(\hat{\boldsymbol{\theta}})\right\vert ds\\
    &=o_{p}(1)+o_{p|\mathcal{D}_n,\tilde{\boldsymbol{\theta}}_p}(1)= o_{p|\mathcal{D}_n,\tilde{\boldsymbol{\theta}}_p}(1).
\end{aligned}
\end{equation*}
\end{proof}

\begin{lemma} \label{lemma6}
    Under Assumptions \ref{ass7} and \ref{ass9}, given $\mathcal{D}_n$ in probability, we have
    \begin{equation*}
       {\boldsymbol{\Lambda}_{RDS}^0}^{-\frac{1}{2}} \frac{\sqrt{r}}{rn}  \sum_{i=1}^{r}\frac{1}{\pi_{\xi,i}^{RDS}(\tilde{\boldsymbol{\theta}}_p)} \dot{l}_{i}^*(\boldsymbol{\theta})\xrightarrow{d} N(0,\boldsymbol{I}),
    \end{equation*}
    where  ${\boldsymbol{\Lambda}_{RDS}^0}=\frac{1}{n^2} \sum_{i=1}^n \frac{1}{\pi^{RDS}_{\xi,i}(\tilde{\boldsymbol{\theta}}_p)} \dot{l}_i(\hat{\boldsymbol{\theta}})^2$
\end{lemma}
\begin{proof}
    We use the Lindeberg-Feller Theorem to prove \citep{Van_1998}.
    Firstly, let
    \begin{equation*}
        \frac{\sqrt{r}}{rn}  \sum_{i=1}^{r} \dot{l}_{i}^*(\boldsymbol{\theta})/\pi_{\xi,i}^{RDS}(\tilde{\boldsymbol{\theta}}_p) \triangleq \frac{1}{\sqrt{r}} \sum_{i=1}^r \boldsymbol{g}^*_i
    \end{equation*}
    Then, by direct calculation,
    \begin{equation*}
    \begin{aligned}
        \mathbb{E}\left\{\boldsymbol{g}^*_i | \mathcal{D}_n, \tilde{\boldsymbol{\theta}}_p\right\} &= \frac{1}{n}\sum_{i=1}^n \dot{l}_i(\hat{\boldsymbol{\theta}}) = 0,\\
        Var\left\{\boldsymbol{g}^*_i | \mathcal{D}_n, \tilde{\boldsymbol{\theta}}_p\right\} &=  {\boldsymbol{\Lambda}_{RDS}^0}= \frac{1}{n^2} \sum_{i=1}^n \frac{1}{\pi^{RDS}_{\xi,i}(\tilde{\boldsymbol{\theta}}_p)} \dot{l}_i(\hat{\boldsymbol{\theta}})^2 \\
        &=O_p(1).
    \end{aligned}
    \end{equation*}
    Besides, for $\forall \epsilon>0$ and some $\delta \in (0,2]$,
    \begin{equation*}
    \begin{aligned}
        &\frac{1}{r}\sum_{i=1}^r\mathbb{E}\left\{\left\Vert \boldsymbol{g}^*_i\right\Vert^2 I\left( \left\Vert \boldsymbol{g}^*_i\right\Vert>\sqrt{r}\epsilon \right) \Big| \mathcal{D}_n,\tilde{\boldsymbol{\theta}}_p \right\} \\
        \leq & \frac{1}{r^{1+\frac{\delta}{2}}\epsilon^{\delta}} \sum_{i=1}^r\mathbb{E}\left\{\left\Vert \boldsymbol{g}^*_i\right\Vert^{2+\delta} \Big| \mathcal{D}_n,\tilde{\boldsymbol{\theta}}_p \right\}\\
        \leq & \frac{1}{r^{\frac{\delta}{2}}\epsilon^{\delta}n} \sum_{i=1}^n  \frac{\left\Vert\dot{l}_i(\hat{\boldsymbol{\theta}}) \right\Vert^{2+\delta}}{\left[n\pi^{RDS}_{\xi,i}(\tilde{\boldsymbol{\theta}}_p)\right]^{1+\delta}}\\
        =&O_p\left(r^{-\frac{\delta}{2}}\right)=o_p(1)
    \end{aligned}
    \end{equation*}
    Hence, the Lindeberg-Feller conditions are satisfied. The central limit theorem shows that conditionally on $\mathcal{D}_n$ and $\tilde{\boldsymbol{\theta}}_p$,
     \begin{equation*}
       {\boldsymbol{\Lambda}^0_{RDS}}^{-\frac{1}{2}} \frac{\sqrt{r}}{rn}  \sum_{i=1}^{r}\frac{1}{\pi_{\xi,i}^{RDS*}(\tilde{\boldsymbol{\theta}}_p)} \dot{l}_{i}^*(\boldsymbol{\theta})\xrightarrow{d} N(0,\boldsymbol{I})
    \end{equation*}
\end{proof}

Now we use the Lemmas \ref{lemma5} and \ref{lemma6} to prove:
\begin{equation*}
        \sqrt{r}{\boldsymbol{V}_{RDS}^0}^{-\frac{1}{2}}(\tilde{\boldsymbol{\theta}}_{RDS} - \hat{\boldsymbol{\theta}}) \xrightarrow{d} N(0,\boldsymbol{I}),
    \end{equation*}
where
    \begin{equation*}            \boldsymbol{V}_{RDS}^0=\boldsymbol{M}_{g}^{-1}\boldsymbol{\Lambda}_{RDS}^0\boldsymbol{M}_{g}^{-1},
    \end{equation*}
and
    \begin{equation*}
        \boldsymbol{\Lambda}_{RDS}^0=\frac{1}{n^2}\sum_{i=1}^n \frac{1}{\pi_{\xi,i}^{RDS}(\tilde{\boldsymbol{\theta}}_p)}\dot{l}_i(\hat{\boldsymbol{\theta}})^2.
    \end{equation*}
\begin{proof}
    Directly calculation shows that, for any $\boldsymbol{\theta}$,
\begin{equation*}
\begin{aligned}
    &\mathbb{E} \left\{ \frac{1}{n}l^*_g(\boldsymbol{\theta}) \Big| \mathcal{D}_n,\tilde{\boldsymbol{\theta}}_p \right\} = l_f(\boldsymbol{\theta}), \\
    &Var\left\{ \frac{1}{n}l^*_g(\boldsymbol{\theta}) \Big| \mathcal{D}_n \right\} = O_p(r^{-1})
\end{aligned} 
\end{equation*}
By Chebyshev's inequality, for any $\boldsymbol{\theta}$,
\begin{equation*}
    \frac{1}{n}l^*_g(\boldsymbol{\theta})-l_f(\boldsymbol{\theta}) = o_{p|\mathcal{D}_n,\tilde{\boldsymbol{\theta}}_p}(1)
\end{equation*}
Then, from Theorem 5.9 of \cite{Van_1998}, 
\begin{equation*}
    \left\Vert \tilde{\boldsymbol{\theta}}_{RDS} - \hat{\boldsymbol{\theta}} \right\Vert = o_p(1).
\end{equation*}
By Taylor expansion,
\begin{equation*}
    0=\frac{1}{n}\dot{l}^*_g(\tilde{\boldsymbol{\theta}}_{RDS})=\frac{1}{n}\dot{l}^*_g(\hat{\boldsymbol{\theta}})+B_p\cdot(\tilde{\boldsymbol{\theta}}_{RDS}-\hat{\boldsymbol{\theta}}),
\end{equation*}
By Lemma \ref{lemma5}, we have
\begin{equation*}
    0=\frac{1}{n}\dot{l}^*_g(\tilde{\boldsymbol{\theta}}_{RDS})=\frac{1}{n}\dot{l}^*_g(\hat{\boldsymbol{\theta}})+\left[\ddot{l}_f(\hat{\boldsymbol{\theta}})+o_p(1)\right]\cdot(\tilde{\boldsymbol{\theta}}_{RDS}-\hat{\boldsymbol{\theta}}),
\end{equation*}
which leads to
\begin{equation*}
\begin{aligned}
    &\tilde{\boldsymbol{\theta}}_{RDS}-\hat{\boldsymbol{\theta}} \\
    =& -\left[\ddot{l}_f(\hat{\boldsymbol{\theta}})+o_p(1)\right]^{-1}\frac{1}{n}\dot{l}^*_g(\hat{\boldsymbol{\theta}})\\
    =& -\frac{1}{\sqrt{r}} \left[\ddot{l}_f(\hat{\boldsymbol{\theta}})+o_p(1)\right]^{-1} {\boldsymbol{\Lambda}_{RDS}^0}^{\frac{1}{2}}\cdot \sqrt{r}{\boldsymbol{\Lambda}_{RDS}^0}^{-\frac{1}{2}}\\
    &\cdot\frac{1}{nr}\sum_{i=1}^r\frac{1}{\pi_{\xi,i}^{RDS*}(\tilde{\boldsymbol{\theta}}_p)}\dot{l}^*_i(\hat{\boldsymbol{\theta}})
\end{aligned}
\end{equation*}
Note that $\ddot{l}_f(\hat{\boldsymbol{\theta}}) = \boldsymbol{M}_g$. From Lemma \ref{lemma6} and Slutsky's theorem, we can derive that conditionally on $\mathcal{D}_n$ and $\tilde{\boldsymbol{\theta}}_p$,
\begin{equation*}
    \sqrt{r} {\boldsymbol{V}^0_{RDS}}^{-\frac{1}{2}}(\tilde{\boldsymbol{\theta}}_{RDS} - \hat{\boldsymbol{\theta}}) \xrightarrow{d} N(0,\boldsymbol{I}).
\end{equation*}

\end{proof}

Next, we will show the difference between $\boldsymbol{\Lambda}_{RDS}^0$ and $\boldsymbol{\Lambda}_{RDS}$. By direct calculation,

\begin{equation}\label{107}
\begin{aligned}
&\left\Vert \boldsymbol{\Lambda}_{RDS}^0 - \boldsymbol{\Lambda}_{RDS} \right\Vert \\
=& \left\Vert \frac{1}{n^2}\sum_{i=1}^n \frac{\dot{l}_i(\hat{\boldsymbol{\theta}})^2}{(1-\xi)\pi_i^{RDS}(\tilde{\boldsymbol{\theta}}_p)+\frac{\xi}{n}} - \frac{1}{n^2}\sum_{i=1}^n \frac{\dot{l}_i(\hat{\boldsymbol{\theta}})^2}{(1-\xi)\pi_i^{RDS}(\hat{\boldsymbol{\theta}})+\frac{\xi}{n}} \right\Vert\\
\leq& \frac{1}{n^2}\sum_{i=1}^n \left\Vert \dot{l}_i(\hat{\boldsymbol{\theta}}) \right\Vert^2\left\vert \frac{1}{(1-\xi)\pi_i^{RDS}(\tilde{\boldsymbol{\theta}}_p)+\frac{\xi}{n}}-\frac{1}{(1-\xi)\pi_i^{RDS}(\hat{\boldsymbol{\theta}})+\frac{\xi}{n}} \right\vert \\
<&\frac{1}{{\xi}^2}\sum_{i=1}^n \left\Vert \dot{l}_i(\hat{\boldsymbol{\theta}}) \right\Vert^2 \left\vert \pi_i^{RDS}(\tilde{\boldsymbol{\theta}}_p) - \pi_i^{RDS}(\hat{\boldsymbol{\theta}}) \right\vert.
\end{aligned}
\end{equation}

And
\begin{equation*} 
\begin{aligned}
&\left\vert \pi_i^{RDS}(\tilde{\boldsymbol{\theta}}_p) - \pi_i^{RDS}(\hat{\boldsymbol{\theta}}) \right\vert \\
=& \left\vert \frac{\Vert  \dot{l}_i(\tilde{\boldsymbol{\theta}}_p)\Vert}{\sum_{i=1}^n \Vert \dot{l}_i(\tilde{\boldsymbol{\theta}}_p) \Vert} - \frac{\Vert \dot{l}_i(\hat{\boldsymbol{\theta}}) \Vert}{\sum_{i=1}^n \Vert  \dot{l}_i(\hat{\boldsymbol{\theta}}) \Vert} \right\vert \\
\leq & \left\Vert \dot{l}_i(\hat{\boldsymbol{\theta}})  \right\Vert \cdot \frac{\sum_{j=1}^n \left\vert \Vert \dot{l}_j(\tilde{\boldsymbol{\theta}}_p)\Vert - \Vert \dot{l}_j(\hat{\boldsymbol{\theta}}) \Vert \right\vert}{\sum_{j=1}^n \Vert\dot{l}_j(\tilde{\boldsymbol{\theta}}_p)\Vert \sum_{j=1}^n \Vert\dot{l}_j(\hat{\boldsymbol{\theta}})\Vert}\\
+& \frac{\left\vert \Vert \dot{l}_i(\tilde{\boldsymbol{\theta}}_p)\Vert - \Vert \dot{l}_i(\hat{\boldsymbol{\theta}}) \Vert \right\vert}{\sum_{j=1}^n \Vert\dot{l}_j(\tilde{\boldsymbol{\theta}}_p)\Vert} \\
\triangleq& E_{i1} + E_{i2}.
\end{aligned}
\end{equation*}

By Assumption \ref{ass6}, for $j=1,2,...,n$
\begin{equation*}
\begin{aligned}
&\left\vert \Vert \dot{l}_j(\tilde{\boldsymbol{\theta}}_p)\Vert - \Vert \dot{l}_j(\hat{\boldsymbol{\theta}}) \Vert \right\vert 
\leq  \left\Vert \dot{l}_j(\tilde{\boldsymbol{\theta}}_p) -  \dot{l}_j(\hat{\boldsymbol{\theta}}) \right\Vert  
\\
\leq&  \sqrt{\sum_{k=1}^d \left[ \dot{l}_{j(k)}(\tilde{\boldsymbol{\theta}}_p) -  \dot{l}_{j(k)}(\hat{\boldsymbol{\theta}})\right]^2} \\
\leq & \sum_{k=1}^d \left\vert \dot{l}_{j(k)}(\tilde{\boldsymbol{\theta}}_p) -  \dot{l}_{j(k)}(\hat{\boldsymbol{\theta}}) \right\vert \leq  \sum_{k=1}^d \left\vert \ddot{l}_{j(k)}^{\top}(\acute{\boldsymbol{\theta}}_k)(\tilde{\boldsymbol{\theta}}_p - \hat{\boldsymbol{\theta}}) \right\vert \\
\leq& 
\Vert \tilde{\boldsymbol{\theta}}_p - \hat{\boldsymbol{\theta}} \Vert \sum_{k=1}^d \Vert \ddot{l}_{j(k)}(\acute{\boldsymbol{\theta}}_k)\Vert,
\end{aligned}
\end{equation*}
where $\dot{l}_{j(k)}$ is the $k$th element of $\dot{l}_j$, $\ddot{l}_{j(k)}$ is the $k$th column of $\ddot{l}_{j}$ and all the $\acute{\boldsymbol{\theta}}_k$ are between $\hat{\boldsymbol{\theta}}$ and $\tilde{\boldsymbol{\theta}}_p$.

Let $\zeta_j \triangleq\sum_{k=1}^d \Vert \ddot{l}_{j(k)}(\acute{\boldsymbol{\theta}}_k)\Vert$, we have
$$E_{i1} \leq  \frac{\Vert \dot{l}_i(\hat{\boldsymbol{\theta}}) \Vert \Vert \tilde{\boldsymbol{\theta}}_p - \hat{\boldsymbol{\theta}} \Vert \sum_{j=1}^n \zeta_j}{\sum_{j=1}^n \Vert\dot{l}_j(\tilde{\boldsymbol{\theta}}_p)\Vert \sum_{j=1}^n \Vert\dot{l}_j(\hat{\boldsymbol{\theta}})\Vert},$$
and 
$$E_{i2}\leq \frac{\Vert \tilde{\boldsymbol{\theta}}_p - \hat{\boldsymbol{\theta}} \Vert \zeta_i}{\sum_{j=1}^n \Vert\dot{l}_j(\tilde{\boldsymbol{\theta}}_p)\Vert}.$$

From Equation \ref{66} and Assumption \ref{ass6}, we have
\begin{equation*}
\begin{aligned}
&\frac{1}{n}\sum_{j=1}^n \zeta_i^2 \\
\leq& \frac{d}{n}\sum_{j=1}^n\sum_{k=1}^d \Vert \ddot{l}_{j(k)}(\acute{\boldsymbol{\theta}}_k) \Vert^2 \\
=& \frac{d}{n}\sum_{j=1}^n\sum_{k=1}^d \sum_{i=1}^d \ddot{l}_{j(ki)}(\acute{\boldsymbol{\theta}}_k)^2 \\
\leq& \frac{d}{n}\sum_{j=1}^n\sum_{k=1}^d \sum_{i=1}^d \left( 2\ddot{l}_{j(ki)}(\hat{\boldsymbol{\theta}})^2 + 2\psi(x_j)^2\Vert \tilde{\boldsymbol{\theta}}_p - \hat{\boldsymbol{\theta}} \Vert^2 \right) \\
=& O_p(1).
\end{aligned}
\end{equation*}

This also indicates that $\frac{1}{n}\sum_{i=1}^n \zeta_i = O_p(1)$.
Hence, 
\begin{equation}\label{equ111}
\sum_{i=1}^n \Vert\dot{l}_i(\hat{\boldsymbol{\theta}})\Vert^2 E_{i1} = \Vert \tilde{\boldsymbol{\theta}}_p - \hat{\boldsymbol{\theta}} \Vert \frac{1}{n}\sum_{i=1}^n \Vert\dot{l}_i(\hat{\boldsymbol{\theta}})\Vert^2\cdot O_p(1)=O_p(\Vert \tilde{\boldsymbol{\theta}}_p - \hat{\boldsymbol{\theta}} \Vert),
\end{equation}

\begin{equation}\label{equ112}
\begin{aligned}
\sum_{i=1}^n \Vert\dot{l}_i(\hat{\boldsymbol{\theta}})\Vert^2 E_{i2} &\leq \Vert \tilde{\boldsymbol{\theta}}_p - \hat{\boldsymbol{\theta}} \Vert\frac{1}{n}\sum_{i=1}^n \Vert\dot{l}_i(\hat{\boldsymbol{\theta}})\Vert^2\zeta_i \cdot O_p(1) \\
&\leq  \sqrt{\frac{1}{n}\sum_{i=1}^n \Vert \dot{l}_i(\hat{\boldsymbol{\theta}}) \Vert^4}\cdot \sqrt{\frac{1}{n}\sum_{i=1}^n \zeta_i^2}\cdot O_p(\Vert \tilde{\boldsymbol{\theta}}_p - \hat{\boldsymbol{\theta}}\Vert)\\
&= O_p(\Vert \tilde{\boldsymbol{\theta}}_p - \hat{\boldsymbol{\theta}} \Vert)
\end{aligned}
\end{equation} 

When $r_0,r,n\rightarrow \infty$, based on Equations \ref{107}, \ref{equ111} and \ref{equ112}, we have
$$\Vert \boldsymbol{\Lambda}_{RDS}^0 - \boldsymbol{\Lambda}_{RDS} \Vert =  O_p(\Vert \tilde{\boldsymbol{\theta}}_p - \hat{\boldsymbol{\theta}} \Vert) = o_p(1).$$

Therefore, given $\mathcal{D}_n$ and $\tilde{\boldsymbol{\theta}}_p$, as $r_0,r,n\rightarrow \infty$, we have 

$$\sqrt{r}\boldsymbol{V}_{RDS}^{-\frac{1}{2}}(\tilde{\boldsymbol{\theta}}_{RDS} - \hat{\boldsymbol{\theta}}) \xrightarrow{d} N(0,\boldsymbol{I}).$$

\subsubsection{Proof of Theorem \ref{ssp:cen-sub}}
Note that if we want to minimize the trace of 
\begin{equation*}
    \boldsymbol{\Lambda}_c = \frac{1}{n^2} \sum_{i=n_0+1}^n \frac{1}{\tilde{\pi}_i} \dot{l}_i(\hat{\boldsymbol{\theta}})^2 - \left[ \frac{1}{n}\sum_{i=n_0+1}^n \dot{l}_i(\hat{\boldsymbol{\theta}}) \right]^2,
\end{equation*}
we only need to minimize
\begin{equation*}
    \frac{1}{n^2} \sum_{i=n_0+1}^n \frac{1}{\tilde{\pi}_i} \dot{l}_i(\hat{\boldsymbol{\theta}})^2 \triangleq \tilde{\boldsymbol{\Lambda}}_c.
\end{equation*}
Then,
\begin{equation*}
\begin{aligned}
    tr(\tilde{\boldsymbol{\Lambda}}_c) &=\frac{1}{n^2}\sum_{i=n_0+1}^n\frac{1}{\tilde{\pi}_i} \left\Vert\dot{l}_i(\hat{\boldsymbol{\theta}})\right\Vert^2 \\
    &= \frac{1}{n^2}\left(\sum_{i=n_0+1}^n\frac{1}{\tilde{\pi}_i} \left\Vert\dot{l}_i(\hat{\boldsymbol{\theta}})\right\Vert^2 \right) \left(\sum_{i=n_0+1}^n \tilde{\pi}_i\right)\\
    &\geq  \frac{1}{n^2} \left(\sum_{i=n_0+1}^n \left\Vert\dot{l}_i(\hat{\boldsymbol{\theta}})\right\Vert \right)^2
\end{aligned}
\end{equation*}

The last inequality is based on the Cauchy-Schwarz inequality and it holds if and only if when $\tilde{\pi}_i \propto \Vert \dot{l}_i(\hat{\boldsymbol{\theta}})\Vert$.

\subsubsection{Proof of Proposition \ref{thm:RDCS-mle}}
The proof of Proposition \ref{thm:RDCS-mle} needs the following two lemmas. The derivation of these two lemmas is similar to lemmas \ref{lemma3} and \ref{lemma4}, so we do not go into details.  For convenience, we denote $\omega_{\xi,i}^{RDCS}$ as $\omega_{\xi,i}^{RDCS}(\tilde{\boldsymbol{\theta}}_{cp})$ and $\tilde{\pi}_{\xi,i}^{RDCS}$ as $\tilde{\pi}_{\xi,i}^{RDCS}(\tilde{\boldsymbol{\theta}}_{cp})$. 

\begin{lemma}\label{lemma7}
    Under Assumptions \ref{ass6} and \ref{ass9}, if $\left\Vert \tilde{\boldsymbol{\theta}}_{RDCS} - \hat{\boldsymbol{\theta}} \right\Vert = o_p(1)$, conditional on $\mathcal{D}_n$ and $\tilde{\boldsymbol{\theta}}_{cp}$, 
    \begin{equation*}
        B_{cp} - \ddot{l}_f(\hat{\boldsymbol{\theta}}) = o_p(1),
    \end{equation*}
    where 
    \begin{equation*}
        \begin{aligned}
            &B_{cp} = \frac{1}{n} \int_0^1 \ddot{l}^*_c \left(\hat{\boldsymbol{\theta}} + s(\tilde{\boldsymbol{\theta}}_{RDCS} - \hat{\boldsymbol{\theta}})\right) ds, \\
            &\ddot{l}_f(\hat{\boldsymbol{\theta}}) = \frac{1}{n} \sum_{i=1}^n \ddot{l}_i(\hat{\boldsymbol{\theta}}) = \boldsymbol{M}_g.
        \end{aligned}
    \end{equation*}
\end{lemma}

\begin{proof}
Note that
\begin{equation*}
\begin{aligned}
    \mathbb{E}\left( \frac{1}{n} \sum_{i=1}^r \omega_{\xi,i}^{c,RDCS} \psi(\boldsymbol{x}^c_i) \Big| \mathcal{D}_n,\tilde{\boldsymbol{\theta}}_{cp} \right) &= \frac{1}{n} \sum_{i=1}^n \psi(\boldsymbol{x}_i)\\
    &= \mathbb{E}\left\{\psi(\boldsymbol{x})\right\} + o_p(1),
\end{aligned}
\end{equation*}
and
\begin{equation*}
    \begin{aligned}
        &Var \left( \frac{1}{n} \sum_{i=1}^r \omega_{\xi,i}^{c,RDCS}\psi(\boldsymbol{x}^c_i) \Big| \mathcal{D}_n ,\tilde{\boldsymbol{\theta}}_{cp}\right) \\
        =& \frac{1}{r-n_0}\sum_{i=n_0+1}^{n} \frac{1}{n^2\tilde{\pi}^{RDCS}_{\xi,i}} \psi(\boldsymbol{x}_i)^2 \\
        \leq & \max_{i=n_0+1,...,n} \left( \frac{1}{n\tilde{\pi}^{RDCS}_{\xi,i}} \right) \frac{\sum_{i=n_0+1}^n \psi(\boldsymbol{x}_i)^2}{n(r-n_0)}  \\
        =& O_p\left(\frac{1}{r-n_0}\right) = O_p\left(r^{-1}\right).
    \end{aligned}
\end{equation*}
Then we have
\begin{equation*}
    \frac{1}{n} \sum_{i=1}^r \omega_{\xi,i}^{c,RDCS} \psi(\boldsymbol{x}^c_i) = O_{p|\mathcal{D}_n,\tilde{\boldsymbol{\theta}}_{cp}}(1)
\end{equation*}
For every $k,l =1,...,d$ and $s\in (0,1)$, from Assumption \ref{ass7}, we have
\begin{equation*}
\begin{aligned}
    &\left\vert \frac{1}{n} \ddot{l}^*_{c(kl)} \left(\hat{\boldsymbol{\theta}} + s(\tilde{\boldsymbol{\theta}}_{RDCS} - \hat{\boldsymbol{\theta}})\right) -  \frac{1}{n} \ddot{l}^*_{c(kl)} \left( \hat{\boldsymbol{\theta}}\right) \right\vert\\
    &= s\left\Vert \tilde{\boldsymbol{\theta}}_{RDCS} - \hat{\boldsymbol{\theta}} \right\Vert \times \frac{1}{n} \sum_{i=1}^r \omega_{\xi,i}^{c,RDCS} \psi(\boldsymbol{x}^c_i) \\
    &= o_p(1).
\end{aligned}
\end{equation*}
From the proof of Theorem \ref{thm:censor-all}, we have
\begin{equation*}
\begin{aligned}
    \frac{1}{n} \sum_{i=1}^n \ddot{l}^2_{i(kl)}(\hat{\boldsymbol{\theta}})
    &= O_p(1)\\
   \mathbb{E} \left\{\frac{1}{n} \ddot{l}^*_{c(kl)} \left( \hat{\boldsymbol{\theta}}\right) \Big| \mathcal{D}_n,\tilde{\boldsymbol{\theta}}_{cp} \right\} &= \frac{1}{n} \sum_{i=1}^n \ddot{l}_{i(kl)}(\hat{\boldsymbol{\theta}}) = \ddot{l}_{f(kl)}(\hat{\boldsymbol{\theta}})\\
   Var\left\{\frac{1}{n} \ddot{l}^*_{c(kl)} \left( \hat{\boldsymbol{\theta}}\right) \Big| \mathcal{D}_n ,\tilde{\boldsymbol{\theta}}_{cp}\right\}& = O_p(r^{-1})
\end{aligned}
\end{equation*}
Then, by Chebyshev's inequality,
\begin{equation*}
    \left\vert \frac{1}{n}\ddot{l}^*_{c(kl)}(\hat{\boldsymbol{\theta}}) -\frac{1}{n} \sum_{i=1}^n \ddot{l}_{i(kl)}(\hat{\boldsymbol{\theta}}) \right\vert = O_{P|\mathcal{D}_n,\tilde{\boldsymbol{\theta}}_{cp}}\left(r^{-\frac{1}{2}}\right) = o_{p}(1).
\end{equation*}
Therefore, 
\begin{equation*}
\begin{aligned}
    &\left\vert B_{cp(kl)} - \ddot{l}_{f(kl)}(\hat{\boldsymbol{\theta}}) \right\vert \leq \int_0^1 \left\vert \frac{1}{n} \ddot{l}^*_{c(kl)} \left(\hat{\boldsymbol{\theta}} + s(\tilde{\boldsymbol{\theta}}_{RDCS} - \hat{\boldsymbol{\theta}})\right) -\right.\\
    &\left. \frac{1}{n} \sum_{i=1}^n \ddot{l}_{i(kl)}(\hat{\boldsymbol{\theta}}) \right\vert ds \\
    & \leq \int_0^1 \left\vert\frac{1}{n} \sum_{i=1}^{n_0} \ddot{l}^c_{i(kl)}\left(\hat{\boldsymbol{\theta}} + s(\tilde{\boldsymbol{\theta}}_c - \hat{\boldsymbol{\theta}})\right) - \frac{1}{n} \sum_{i=1}^{n_0}\ddot{l}^c_{i(kl)}(\hat{\boldsymbol{\theta}}) \right\vert ds \\
    &+ \int_0^1 \left\vert\frac{1}{n} \sum_{i=n_0+1}^r \frac{1}{(r-n_0)\tilde{\pi}_{\xi,i}^{c,RDCS}} \ddot{l}^c_{i(kl)}\left(\hat{\boldsymbol{\theta}} + s(\tilde{\boldsymbol{\theta}}_{RDCS} - \hat{\boldsymbol{\theta}})\right) \right.\\
    &\left.- \frac{1}{n}\sum_{i=n_0+1}^r \frac{1}{(r-n_0)\tilde{\pi}_{\xi,i}^{c,RDCS}}\ddot{l}^c_{i(kl)}(\hat{\boldsymbol{\theta}}) \right\vert ds \\
    &+ \left\vert \frac{1}{n}\sum_{i=n_0+1}^r \frac{1}{(r-n_0)\tilde{\pi}_{\xi,i}^{c,RDCS}}\ddot{l}^c_{i(kl)}(\hat{\boldsymbol{\theta}}) - \frac{1}{n}\sum_{i=n_0+1}^r\ddot{l}^c_{i(kl)}(\hat{\boldsymbol{\theta}})
    \right\vert \\
    &=o_{p}(1) + o_p(1) + o_{p}(1) = o_{p}(1).
\end{aligned}
\end{equation*}
\end{proof}

\begin{lemma} \label{lemma8}
    Under Assumptions \ref{ass7} and \ref{ass9}, given $\mathcal{D}_n$ and $\tilde{\boldsymbol{\theta}}_{cp}$ in probability, we have
    \begin{equation*}
       {\boldsymbol{\Lambda}^0_{RDCS}}^{-\frac{1}{2}} \sum_{i=1}^r \frac{\sqrt{r}}{n} \omega_{\xi,i}^{c,RDCS}\dot{l}^c_i(\hat{\boldsymbol{\theta}}) \xrightarrow{d} N(0,\boldsymbol{I}),
    \end{equation*}
    where
    \begin{equation*}
    \begin{aligned}
        \boldsymbol{\Lambda}^0_{RDCS} 
        &= \frac{1}{n^2} \sum_{i=n_0+1}^n \frac{1}{\tilde{\pi}_{\xi,i}^{RDCS}(\tilde{\boldsymbol{\theta}}_{cp})} \dot{l}_i(\hat{\boldsymbol{\theta}})^2 - \left[ \frac{1}{n}\sum_{i=n_0+1}^n \dot{l}_i(\hat{\boldsymbol{\theta}}) \right]^2\\
        &= \frac{1}{n^2} \sum_{i=n_0+1}^n \frac{1}{\tilde{\pi}_{\xi,i}^{RDCS}(\hat{\boldsymbol{\theta}})} \left[\frac{\partial \log(1-F(t_i,\hat{\boldsymbol{\theta}}))}{\partial \boldsymbol{\theta}} \right.\\
        &\quad- \left.\frac{\partial \log(1-F(t_{il},\hat{\boldsymbol{\theta}}))}{\partial \boldsymbol{\theta}}\right]^2  \\
        &-\left[ \frac{1}{n}\sum_{i=n_0+1}^n \frac{\partial \log(1-F(t_i,\hat{\boldsymbol{\theta}}))}{\partial \boldsymbol{\theta}} -  \frac{\partial \log(1-F(t_{il},\hat{\boldsymbol{\theta}}))}{\partial \boldsymbol{\theta}}\right]^2.
    \end{aligned}
    \end{equation*}
\end{lemma}
\begin{proof}
    We use the Lindeberg-Feller Theorem to prove \citep{Van_1998}.
    Note that
    \begin{equation*}
    \begin{aligned}
        &\mathbb{E}\left\{\sum_{i=1}^r\frac{\sqrt{r}}{n} \omega_{\xi,i}^{c,RDCS}\dot{l}^c_i(\hat{\boldsymbol{\theta}}) \Big| \mathcal{D}_n,\tilde{\boldsymbol{\theta}}_{cp}\right\} = \frac{\sqrt{r}}{n}\sum_{i=1}^n \dot{l}_i(\hat{\boldsymbol{\theta}}) = 0,\\
        &Var\left\{\sum_{i=1}^r \frac{\sqrt{r}}{n} \omega_{\xi,i}^{c,RDCS}\dot{l}^c_i(\hat{\boldsymbol{\theta}}) \Big| \mathcal{D}_n,\tilde{\boldsymbol{\theta}}_{cp}\right\} \\
        =& \boldsymbol{\Lambda}_{RDCS}^0 \\
        =& \frac{1}{n^2} \sum_{i=n_0+1}^n \frac{1}{\tilde{\pi}^{RDCS}_{\xi,i}} \dot{l}_i(\hat{\boldsymbol{\theta}})^2 - \left[ \frac{1}{n}\sum_{i=n_0+1}^n \dot{l}_i(\hat{\boldsymbol{\theta}}) \right]^2 \\
        \leq& \frac{1}{n}\max_{i=n_0+1,...,n}\left(\frac{1}{n\tilde{\pi}^{RDCS}_{\xi,i}} \right) \sum_{i=n_0+1}^n \dot{l}_i(\hat{\boldsymbol{\theta}})^2 \\
        =&O_p(1).
    \end{aligned}
    \end{equation*}
    where the inequality is in the sense Loewner ordering. Besides, for $\forall \epsilon>0$ and some $\delta \in (0,2]$,
    \begin{equation*}
    \begin{aligned}
        &\mathbb{E}\left\{ \sum_{i=1}^r \left\Vert \frac{\sqrt{r}}{n}\omega_{\xi,i}^{c,RDCS} \dot{l}^c_i(\hat{\boldsymbol{\theta}}) \right\Vert^2\right. \\
        &\quad \cdot\left. I\left( \left\Vert \frac{\sqrt{r}}{n}\omega_{\xi,i}^{c,RDCS} \dot{l}^c_i(\hat{\boldsymbol{\theta}}) \right\Vert>\epsilon \right) \Big| \mathcal{D}_n,\tilde{\boldsymbol{\theta}}_{cp} \right\} \\
        \leq & r^{1+\frac{\delta}{2}}n^{-2-\delta}\epsilon^{-\delta} \mathbb{E}\left\{ \sum_{i=1}^r \left\Vert \omega_{\xi,i}^{c,RDCS} \dot{l}^c_i(\hat{\boldsymbol{\theta}}) \right\Vert^{2+\delta} \right.\\
        & \quad \left. I\left( \left\Vert \omega_{\xi,i}^{c,RDCS} \dot{l}^c_i(\hat{\boldsymbol{\theta}}) \right\Vert>\frac{n}{\sqrt{r}}\epsilon \right) \Big| \mathcal{D}_n,\tilde{\boldsymbol{\theta}}_{cp} \right\}\\
        \leq & r^{1+\frac{\delta}{2}}n^{-2-\delta}\epsilon^{-\delta} \mathbb{E}\left\{ \sum_{i=1}^r \left\Vert \omega_{\xi,i}^{c,RDCS} \dot{l}^c_i(\hat{\boldsymbol{\theta}}) \right\Vert^{2+\delta} \Big| \mathcal{D}_n,\tilde{\boldsymbol{\theta}}_{cp} \right\} \\
        \leq & \frac{r^{1+\frac{\delta}{2}}}{n^{2+\delta}\epsilon^{\delta}} \left[\sum_{i=1}^{n_0} \left\Vert \dot{l}^c_i(\hat{\boldsymbol{\theta}}) \right\Vert^{2+\delta} \right.\\
        & \quad \left.+ \frac{1}{(r-n_0)^{1+\delta}}\sum_{i=n_0+1}^n \left(\frac{1}{\tilde{\pi}_{\xi,i}^{RDCS}}\right)^{1+\delta} \left\Vert \dot{l}^c_i(\hat{\boldsymbol{\theta}}) \right\Vert^{2+\delta}\right] \\
        \leq & o_p(1)+ \frac{1}{r^{\frac{\delta}{2}}\epsilon^\delta} \max_{i=n_0+1,...,n}\left(\frac{1}{n\tilde{\pi}^{RDCS}_{\xi,i}}\right)^{1+\delta} \\
        & \quad\cdot\frac{1}{n} \sum_{i=n_0+1}^n\left\Vert \dot{l}^c_i(\hat{\boldsymbol{\theta}}) \right\Vert^{2+\delta} \\
        =&O_p\left(r^{-\frac{\delta}{2}}\right)
    \end{aligned}
    \end{equation*}
    Hence, the Lindeberg-Feller conditions are satisfied. The central limit theorem shows that conditionally on $\mathcal{D}_n$ and $\tilde{\boldsymbol{\theta}}_{cp}$,
      \begin{equation*}
       {\boldsymbol{\Lambda}^0_{RDCS}}^{-\frac{1}{2}}\sum_{i=1}^r \frac{\sqrt{r}}{n} \omega_{\xi,i}^{c,RDCS}\dot{l}^c_i(\hat{\boldsymbol{\theta}}) \xrightarrow{d} N(0,\boldsymbol{I})
    \end{equation*}
\end{proof}

Now we use the Lemmas \ref{lemma7} and \ref{lemma8} to prove:
\begin{align*}
    \sqrt{r}{\boldsymbol{V}^0_{RDCS}}^{-\frac{1}{2}}(\tilde{\boldsymbol{\theta}}_{RDCS} - \hat{\boldsymbol{\theta}}) \xrightarrow{d}N(0,\boldsymbol{I}),
\end{align*}
where
\begin{equation*}
    \begin{aligned}
        \boldsymbol{V}^0_{RDCS} &= \boldsymbol{M}_g^{-1} \boldsymbol{\Lambda}^0_{RDCS} \boldsymbol{M}_g^{-1}
    \end{aligned}
\end{equation*}

\begin{proof}
    Directly calculation shows that, for any $\boldsymbol{\theta}$,
\begin{equation*}
\begin{aligned}
    \mathbb{E} \left\{ \frac{1}{n}l^*_c(\boldsymbol{\theta}) \Big| \mathcal{D}_n,\tilde{\boldsymbol{\theta}}_{cp} \right\} &= l_f(\boldsymbol{\theta}), \\
    Var\left\{ \frac{1}{n}l^*_c(\boldsymbol{\theta}) \Big| \mathcal{D}_n,\tilde{\boldsymbol{\theta}}_{cp}\right\} &= O_p(r^{-1})
\end{aligned} 
\end{equation*}
By Chebyshev's inequality, for any $\boldsymbol{\theta}$,
\begin{equation*}
    \frac{1}{n}l^*_c(\boldsymbol{\theta})-l_f(\boldsymbol{\theta}) = o_{p|\mathcal{D}_n}(1)
\end{equation*}
Then, from Theorem 5.9 of \cite{Van_1998}, conditionally on $\mathcal{D}_n$,
\begin{equation*}
    \left\Vert \tilde{\boldsymbol{\theta}}_{RDCS} - \hat{\boldsymbol{\theta}} \right\Vert = o_p(1).
\end{equation*}
By Taylor expansion,
\begin{equation*}
    0=\frac{1}{n}\dot{l}^*_c(\tilde{\boldsymbol{\theta}}_{RDCS})=\frac{1}{n}\dot{l}^*_c(\hat{\boldsymbol{\theta}})+B_{cp}\cdot(\tilde{\boldsymbol{\theta}}_{RDCS}-\hat{\boldsymbol{\theta}}),
\end{equation*}
By Lemma \ref{lemma7}, we have
\begin{equation*}
    0=\frac{1}{n}\dot{l}^*_c(\tilde{\boldsymbol{\theta}}_{RDCS})=\frac{1}{n}\dot{l}^*_c(\hat{\boldsymbol{\theta}})+\left[\ddot{l}_f(\hat{\boldsymbol{\theta}})+o_p(1)\right]\cdot(\tilde{\boldsymbol{\theta}}_{RDCS}-\hat{\boldsymbol{\theta}}),
\end{equation*}
which leads to
\begin{equation*}
\begin{aligned}
    &\tilde{\boldsymbol{\theta}}_{RDCS}-\hat{\boldsymbol{\theta}}\\
    =& -\left[\ddot{l}_f(\hat{\boldsymbol{\theta}})+o_p(1)\right]^{-1}\frac{1}{n}\dot{l}^*_c(\hat{\boldsymbol{\theta}})\\
    =& - \left[\ddot{l}_f(\hat{\boldsymbol{\theta}})+o_p(1)\right]^{-1} {\boldsymbol{\Lambda}^0_{RDCS}}^{\frac{1}{2}}\cdot{\boldsymbol{\Lambda}^0_{RDCS}}^{-\frac{1}{2}}\frac{1}{n}\sum_{i=1}^r\omega_{\xi,i}^{c,RDCS}\dot{l}^c_i(\hat{\boldsymbol{\theta}}).
\end{aligned}
\end{equation*}
Note that $\ddot{l}_f(\hat{\boldsymbol{\theta}}) = \boldsymbol{M}_g$. From Lemma \ref{lemma8} and Slutsky's theorem, we can derive that conditionally on $\mathcal{D}_n$ and $\tilde{\boldsymbol{\theta}}_{pc}$:
   \begin{align}
        \sqrt{r}{\boldsymbol{V}_{RDCS}^0}^{-\frac{1}{2}}(\tilde{\boldsymbol{\theta}}_{RDCS} - \hat{\boldsymbol{\theta}}) \xrightarrow{d}N(0,\boldsymbol{I}).
    \end{align}

\end{proof}

Next, we calculate the difference between $\boldsymbol{\Lambda}_{RDCS}^0$ and $\boldsymbol{\Lambda}_{RDCS}$. By direct calculation,

\begin{equation*}
\begin{aligned}
&\left\Vert \boldsymbol{\Lambda}_{RDCS}^0 - \boldsymbol{\Lambda}_{RDCS} \right\Vert \\
=& \left\Vert \frac{1}{n^2}\sum_{i=n_0+1}^n \frac{\dot{l}_i(\hat{\boldsymbol{\theta}})^2}{(1-\xi_c)\tilde{\pi}_i^{RDCS}(\tilde{\boldsymbol{\theta}}_{cp})+\frac{\xi_c}{n-n_0}} \right.\\
\quad-& \left. \frac{1}{n^2}\sum_{i=n_0+1}^n \frac{\dot{l}_i(\hat{\boldsymbol{\theta}})^2}{(1-\xi_c)\tilde{\pi}_i^{RDCS}(\hat{\boldsymbol{\theta}})+\frac{\xi_c}{n}} \right\Vert\\
\leq& \frac{1}{n^2}\sum_{i=n_0+1}^n \left\Vert \dot{l}_i(\hat{\boldsymbol{\theta}}) \right\Vert^2\left\vert \frac{1}{(1-\xi_c)\tilde{\pi}_i^{RDCS}(\tilde{\boldsymbol{\theta}}_{cp})+\frac{\xi_c}{n-n_0}}- \right. \\
\quad-& \left.\frac{1}{(1-\xi_c)\tilde{\pi}_i^{RDCS}(\hat{\boldsymbol{\theta}})+\frac{\xi_c}{n-n_0}} \right\vert \\
<&\left(\frac{n-n_0}{n\xi_c}\right)^2\sum_{i=n_0+1}^n \left\Vert \dot{l}_i(\hat{\boldsymbol{\theta}}) \right\Vert^2 \left\vert \tilde{\pi}_i^{RDCS}(\tilde{\boldsymbol{\theta}}_{cp}) - \tilde{\pi}_i^{RDCS}(\hat{\boldsymbol{\theta}}) \right\vert.
\end{aligned}
\end{equation*}

And
\begin{equation*} 
\begin{aligned}
&\left\vert \tilde{\pi}_i^{RDCS}(\tilde{\boldsymbol{\theta}}_{cp}) - \tilde{\pi}_i^{RDCS}(\hat{\boldsymbol{\theta}}) \right\vert \\
=& \left\vert \frac{\Vert  \dot{l}_i(\tilde{\boldsymbol{\theta}}_{cp})\Vert}{\sum_{i=n_0+1}^n \Vert \dot{l}_i(\tilde{\boldsymbol{\theta}}_{cp}) \Vert} - \frac{\Vert \dot{l}_i(\hat{\boldsymbol{\theta}}) \Vert}{\sum_{i=n_0+1}^n \Vert  \dot{l}_i(\hat{\boldsymbol{\theta}}) \Vert} \right\vert \\
\leq & \left\Vert \dot{l}_i(\hat{\boldsymbol{\theta}})  \right\Vert \cdot \frac{\sum_{j=n_0+1}^n \left\vert \Vert \dot{l}_j(\tilde{\boldsymbol{\theta}}_{cp})\Vert - \Vert \dot{l}_j(\hat{\boldsymbol{\theta}}) \Vert \right\vert}{\sum_{j=n_0+1}^n \Vert\dot{l}_j(\tilde{\boldsymbol{\theta}}_{cp})\Vert \sum_{j=n_0+1}^n \Vert\dot{l}_j(\hat{\boldsymbol{\theta}})\Vert} \\
+ & \frac{\left\vert \Vert \dot{l}_i(\tilde{\boldsymbol{\theta}}_p)\Vert - \Vert \dot{l}_i(\hat{\boldsymbol{\theta}}) \Vert \right\vert}{\sum_{j=n_0+1}^n \Vert\dot{l}_j(\tilde{\boldsymbol{\theta}}_p)\Vert} \\
\triangleq& E^c_{i1} + E^c_{i2}.
\end{aligned}
\end{equation*}

By Assumption \ref{ass6}, for $j=1,2,...,n$
\begin{equation*}
\begin{aligned}
&\left\vert \Vert \dot{l}_j(\tilde{\boldsymbol{\theta}}_{cp})\Vert - \Vert \dot{l}_j(\hat{\boldsymbol{\theta}}) \Vert \right\vert 
\leq  \left\Vert \dot{l}_j(\tilde{\boldsymbol{\theta}}_{cp}) -  \dot{l}_j(\hat{\boldsymbol{\theta}}) \right\Vert  
\\
\leq&  \sqrt{\sum_{k=1}^d \left[ \dot{l}_{j(k)}(\tilde{\boldsymbol{\theta}}_{cp}) -  \dot{l}_{j(k)}(\hat{\boldsymbol{\theta}})\right]^2} \\
\leq & \sum_{k=1}^d \left\vert \dot{l}_{j(k)}(\tilde{\boldsymbol{\theta}}_{cp}) -  \dot{l}_{j(k)}(\hat{\boldsymbol{\theta}}) \right\vert \leq  \sum_{k=1}^d \left\vert \ddot{l}_{j(k)}^{\top}(\acute{\boldsymbol{\theta}}_{ck})(\tilde{\boldsymbol{\theta}}_{cp} - \hat{\boldsymbol{\theta}}) \right\vert \\
\leq& 
\Vert \tilde{\boldsymbol{\theta}}_{cp} - \hat{\boldsymbol{\theta}} \Vert \sum_{k=1}^d \Vert \ddot{l}_{j(k)}(\acute{\boldsymbol{\theta}}_{ck})\Vert,
\end{aligned}
\end{equation*}
where $\dot{l}_{j(k)}$ is the $k$th element of $\dot{l}_j$, $\ddot{l}_{j(k)}$ is the $k$th column of $\ddot{l}_{j}$ and all the $\acute{\boldsymbol{\theta}}_{ck}$ are between $\hat{\boldsymbol{\theta}}$ and $\tilde{\boldsymbol{\theta}}_{cp}$.

Let $\zeta_{cj} \triangleq\sum_{k=1}^d \Vert \ddot{l}_{j(k)}(\acute{\boldsymbol{\theta}}_{ck})\Vert$, we have
$$E^c_{i1} \leq  \frac{\Vert \dot{l}_i(\hat{\boldsymbol{\theta}}) \Vert \Vert \tilde{\boldsymbol{\theta}}_{cp} - \hat{\boldsymbol{\theta}} \Vert \sum_{j=n_0+1}^n \zeta_{cj}}{\sum_{j=n_0+1}^n \Vert\dot{l}_j(\tilde{\boldsymbol{\theta}}_{cp})\Vert \sum_{j=n_0+1}^n \Vert\dot{l}_j(\hat{\boldsymbol{\theta}})\Vert},$$
and 
$$E^c_{i2}\leq \frac{\Vert \tilde{\boldsymbol{\theta}}_{cp} - \hat{\boldsymbol{\theta}} \Vert \zeta_{ci}}{\sum_{j=n_0+1}^n \Vert\dot{l}_j(\tilde{\boldsymbol{\theta}}_{cp})\Vert}.$$

From Equation \ref{66} and Assumption \ref{ass6}, we have
\begin{equation*}
\frac{1}{n}\sum_{j=n_0+1}^n \zeta_{ci}^2 \leq \frac{1}{n}\sum_{j=1}^n \zeta_{ci}^2 \leq O_p(1).
\end{equation*}

This also indicates that $\frac{1}{n}\sum_{i=n_0+1}^n \zeta_{ci} = O_p(1)$.
Hence, 

\begin{equation*}
\begin{aligned}
    \sum_{i=n_0+1}^n \Vert\dot{l}_{i}(\hat{\boldsymbol{\theta}})\Vert^2 E_{i1} &= \Vert \tilde{\boldsymbol{\theta}}_{cp} - \hat{\boldsymbol{\theta}} \Vert \frac{1}{n}\sum_{i=n_0+1}^n \Vert\dot{l}_i(\hat{\boldsymbol{\theta}})\Vert^2\cdot O_p(1)\\
    &=O_p(\Vert \tilde{\boldsymbol{\theta}}_{cp} - \hat{\boldsymbol{\theta}} \Vert),
\end{aligned}
\end{equation*}

\begin{equation*}
\begin{aligned}
&\sum_{i=n_0+1}^n \Vert\dot{l}_i(\hat{\boldsymbol{\theta}})\Vert^2 E_{i2} \\
\leq& \Vert \tilde{\boldsymbol{\theta}}_{cp} - \hat{\boldsymbol{\theta}} \Vert\frac{1}{n}\sum_{i=n_0+1}^n \Vert\dot{l}_i(\hat{\boldsymbol{\theta}})\Vert^2\zeta_{ci} \cdot O_p(1) \\
\leq&  \sqrt{\frac{1}{n}\sum_{i=n_0+1}^n \Vert \dot{l}_i(\hat{\boldsymbol{\theta}}) \Vert^4}\cdot \sqrt{\frac{1}{n}\sum_{i=n_0+1}^n \zeta_{ci}^2}\cdot O_p(\Vert \tilde{\boldsymbol{\theta}}_{cp} - \hat{\boldsymbol{\theta}}\Vert)\\
=& O_p(\Vert \tilde{\boldsymbol{\theta}}_{cp} - \hat{\boldsymbol{\theta}} \Vert)
\end{aligned}
\end{equation*} 

Thus, when $r_0,r,n\rightarrow \infty$, we have
$$\Vert \boldsymbol{\Lambda}_{RDCS}^0 - \boldsymbol{\Lambda}_{RDCS} \Vert =  O_p(\Vert \tilde{\boldsymbol{\theta}}_{cp} - \hat{\boldsymbol{\theta}} \Vert) = o_p(1).$$

Therefore, given $\mathcal{D}_n$ and $\tilde{\boldsymbol{\theta}}_{cp}$, as $r_0,r,n\rightarrow \infty$, we have 

$$\sqrt{r}\boldsymbol{V}_{RDCS}^{-\frac{1}{2}}(\tilde{\boldsymbol{\theta}}_{RDCS} - \hat{\boldsymbol{\theta}}) \xrightarrow{d} N(0,\boldsymbol{I}).$$

\newpage
\bibhang=1.7pc
\bibsep=2pc
\fontsize{9}{14pt plus.8pt minus .6pt}\selectfont
\renewcommand\bibname{\large \bf References}
\expandafter\ifx\csname
natexlab\endcsname\relax\def\natexlab#1{#1}\fi
\expandafter\ifx\csname url\endcsname\relax
  \def\url#1{\texttt{#1}}\fi
\expandafter\ifx\csname urlprefix\endcsname\relax\def\urlprefix{URL}\fi


\end{document}